\documentclass[
reprint,
superscriptaddress,
nofootinbib,
aps, prx
]{revtex4-1}

\usepackage{custom}
\usepackage{mathtools}
\usepackage[caption=false]{subfig}
\usepackage{tikz-cd}
\usepackage{multirow}
\usepackage{url}
\usepackage{hyperref}
\usepackage{cancel}

\usepackage{booktabs}

\usepackage{enumitem}

\newtheorem{Lemma}{Lemma}
\newtheorem{corol}{Corollary}
\newtheorem{definition}{Definition}

\newtheorem{example}{Example}
\newtheorem{nonexample}{Non-example}

\def\be{\begin{equation}}
\def\ee{\end{equation}}
\def\ba{\begin{eqnarray}}
\def\ea{\end{eqnarray}}

\newcommand\q{\quad}
\def\Nl{{\mathchoice
{\setbox0=\hbox{$\displaystyle\rm N$}\hbox{\hbox to0pt
{\kern0.4\wd0\vrule height0.9\ht0\hss}\box0}}
{\setbox0=\hbox{$\textstyle\rm N$}\hbox{\hbox to0pt
{\kern0.4\wd0\vrule height0.9\ht0\hss}\box0}}
{\setbox0=\hbox{$\scriptstyle\rm N$}\hbox{\hbox to0pt
{\kern0.4\wd0\vrule height0.9\ht0\hss}\box0}}
{\setbox0=\hbox{$\scriptscriptstyle\rm N$}\hbox{\hbox to0pt
{\kern0.4\wd0\vrule height0.9\ht0\hss}\box0}}}}
\def\Zl{{\mathchoice
{\setbox0=\hbox{$\displaystyle\rm Z$}\hbox{\hbox to0pt
{\kern0.4\wd0\vrule height0.9\ht0\hss}\box0}}
{\setbox0=\hbox{$\textstyle\rm Z$}\hbox{\hbox to0pt
{\kern0.4\wd0\vrule height0.9\ht0\hss}\box0}}
{\setbox0=\hbox{$\scriptstyle\rm Z$}\hbox{\hbox to0pt
{\kern0.4\wd0\vrule height0.9\ht0\hss}\box0}}
{\setbox0=\hbox{$\scriptscriptstyle\rm Z$}\hbox{\hbox to0pt
{\kern0.4\wd0\vrule height0.9\ht0\hss}\box0}}}}
\def\Ql{{\mathchoice
{\setbox0=\hbox{$\displaystyle\rm Q$}\hbox{\hbox to0pt
{\kern0.4\wd0\vrule height0.9\ht0\hss}\box0}}
{\setbox0=\hbox{$\textstyle\rm Q$}\hbox{\hbox to0pt
{\kern0.4\wd0\vrule height0.9\ht0\hss}\box0}}
{\setbox0=\hbox{$\scriptstyle\rm Q$}\hbox{\hbox to0pt
{\kern0.4\wd0\vrule height0.9\ht0\hss}\box0}}
{\setbox0=\hbox{$\scriptscriptstyle\rm Q$}\hbox{\hbox to0pt
{\kern0.4\wd0\vrule height0.9\ht0\hss}\box0}}}}
\def\Rl{{\mathchoice
{\setbox0=\hbox{$\displaystyle\rm R$}\hbox{\hbox to0pt
{\kern0.4\wd0\vrule height0.9\ht0\hss}\box0}}
{\setbox0=\hbox{$\textstyle\rm R$}\hbox{\hbox to0pt
{\kern0.4\wd0\vrule height0.9\ht0\hss}\box0}}
{\setbox0=\hbox{$\scriptstyle\rm R$}\hbox{\hbox to0pt
{\kern0.4\wd0\vrule height0.9\ht0\hss}\box0}}
{\setbox0=\hbox{$\scriptscriptstyle\rm R$}\hbox{\hbox to0pt
{\kern0.4\wd0\vrule height0.9\ht0\hss}\box0}}}}
\def\Cl{{\mathchoice
{\setbox0=\hbox{$\displaystyle\rm C$}\hbox{\hbox to0pt
{\kern0.4\wd0\vrule height0.9\ht0\hss}\box0}}
{\setbox0=\hbox{$\textstyle\rm C$}\hbox{\hbox to0pt
{\kern0.4\wd0\vrule height0.9\ht0\hss}\box0}}
{\setbox0=\hbox{$\scriptstyle\rm C$}\hbox{\hbox to0pt
{\kern0.4\wd0\vrule height0.9\ht0\hss}\box0}}
{\setbox0=\hbox{$\scriptscriptstyle\rm C$}\hbox{\hbox to0pt
{\kern0.4\wd0\vrule height0.9\ht0\hss}\box0}}}}
\def\Hl{{\mathchoice
{\setbox0=\hbox{$\displaystyle\rm H$}\hbox{\hbox to0pt
{\kern0.4\wd0\vrule height0.9\ht0\hss}\box0}}
{\setbox0=\hbox{$\textstyle\rm H$}\hbox{\hbox to0pt
{\kern0.4\wd0\vrule height0.9\ht0\hss}\box0}}
{\setbox0=\hbox{$\scriptstyle\rm H$}\hbox{\hbox to0pt
{\kern0.4\wd0\vrule height0.9\ht0\hss}\box0}}
{\setbox0=\hbox{$\scriptscriptstyle\rm H$}\hbox{\hbox to0pt
{\kern0.4\wd0\vrule height0.9\ht0\hss}\box0}}}}
\def\Ol{{\mathchoice
{\setbox0=\hbox{$\displaystyle\rm O$}\hbox{\hbox to0pt
{\kern0.4\wd0\vrule height0.9\ht0\hss}\box0}}
{\setbox0=\hbox{$\textstyle\rm O$}\hbox{\hbox to0pt
{\kern0.4\wd0\vrule height0.9\ht0\hss}\box0}}
{\setbox0=\hbox{$\scriptstyle\rm O$}\hbox{\hbox to0pt
{\kern0.4\wd0\vrule height0.9\ht0\hss}\box0}}
{\setbox0=\hbox{$\scriptscriptstyle\rm O$}\hbox{\hbox to0pt
{\kern0.4\wd0\vrule height0.9\ht0\hss}\box0}}}}
\newcommand{\ca}{\mathcal A}

\newcommand{\cc}{\mathcal C}
\newcommand{\cd}{\mathcal D}

\newcommand{\cg}{\mathcal G}
\newcommand{\ch}{\mathcal H}

\newcommand{\cl}{\mathcal L}
\newcommand{\cm}{\mathcal M}

\newcommand{\cp}{\mathcal P}
\newcommand{\cq}{\mathcal Q}
\newcommand{\calr}{\mathcal R}

\newcommand{\ct}{\mathcal T}

  \newcommand{\Fs}{\mathfrak{S}}

\newcommand{\intsum}{\mathclap{\displaystyle\int}\mathclap{\textstyle\sum}}

\def\nn{\nonumber}

\newcommand{\eqa}{\begin{eqnarray}}
\newcommand{\neqa}{\end{eqnarray}}

\def\la{\langle}

\def\f{\frac}

\usepackage{bbm}

\def\q{{\quad}}

\definecolor{green(munsell)}{rgb}{0.0, 0.66, 0.47}
\definecolor{Blu}{rgb}{0.0, 0.4, 0.7}

\newcommand\proofapp{The proof is given in Appendix~\ref{app:proof}.}

\begin{document}

\title{Relational  Dynamics with Periodic Clocks}

\author{Leonardo Chataignier}
\email{lchataig@cbpf.br}
\altaffiliation{Current address: Centro Brasileiro de Pesquisas F\'{i}sicas, Rua Dr. Xavier Sigaud 150, CEP: 22290-180, Rio de Janeiro, RJ, Brazil}
\affiliation{Department of Physics and EHU Quantum Center, University of the Basque Country UPV/EHU, Barrio Sarriena s/n, 48940 Leioa, Spain}

\author{Philipp A. H\"{o}hn}
\email{philipp.hoehn@oist.jp}
\affiliation{Okinawa Institute of Science and Technology Graduate University, Onna, Okinawa 904 0495, Japan}

\author{Maximilian P. E. Lock}
\email{maximilian.paul.lock@tuwien.ac.at}
\affiliation{Atominstitut, TU Wien, 1020 Vienna, Austria}
\affiliation{Institute for Quantum Optics and Quantum Information (IQOQI),
\\Austrian Academy of Sciences, 1090 Vienna, Austria}

\author{Fabio M. Mele}
\email{fmele1@lsu.edu}
\altaffiliation{Current address: Department of Physics and Astronomy, Louisiana State University, Baton Rouge, LA 70803, USA}
\affiliation{Department of Physics \& Astronomy, Western University, N6A3K7, London ON, Canada}
\affiliation{Okinawa Institute of Science and Technology Graduate University, Onna, Okinawa 904 0495, Japan}

\date{\today}

\begin{abstract}
We discuss a systematic way in which a relational dynamics can be established relative to periodic clocks both in the classical and quantum theories, emphasising the parallels between them.~We show that:~(1) classical and quantum relational observables that encode the value of a quantity relative to a periodic clock are only invariant along the gauge orbits generated by the Hamiltonian constraint if the quantity itself is periodic, and otherwise the observables are only transiently invariant per clock cycle (this implies, in particular, that counting winding numbers does not lead to invariant observables relative to the periodic clock);~(2) the quantum relational observables can be obtained from a \emph{partial} group averaging procedure over a single clock cycle;~(3) there is an equivalence (`trinity') between the quantum theories based on the quantum relational observables of the clock-neutral picture of Dirac quantisation, the relational Schr\"odinger picture of the Page-Wootters formalism, and the relational Heisenberg picture that follows from quantum deparametrisation, all three taken relative to periodic clocks (implying that the dynamics in all three is necessarily periodic);~(4) in the context of periodic clocks, the original Page-Wootters definition of conditional probabilities fails for systems that have a continuous energy spectrum and, using the equivalence between the Page-Wootters and the clock-neutral, gauge-invariant formalism, must be suitably updated.~Finally, we show how a system evolving periodically with respect to a periodic clock can evolve monotonically with respect to an aperiodic clock, without inconsistency.~The presentation is illustrated by several examples, and we conclude with a brief comparison to other approaches in the literature that also deal with relational descriptions of periodic clocks.
\end{abstract}

\maketitle
\onecolumngrid
\newpage
\tableofcontents
\newpage
\twocolumngrid

\section{Introduction}
One of the central issues in quantum gravity and quantum foundations is the development of a theory in which time and space are not arbitrary, unphysical coordinates, but are instead defined from physical entities (dynamical reference frames).~Such a theory would provide a relational account of dynamics, in the sense that the evolution would be defined relative to the readings of physical clocks.~While a great deal of effort has been devoted to the construction of models of relational dynamics, most of the literature focuses on aperiodic (uniform or monotonic) clocks, the values of which change monotonically without repetition.~With the exception of some approaches that address different kinds of periodic reference frames, a systematic framework for dealing with periodic clocks is currently lacking.~In the following, we address this gap.

Most clocks in our everyday lives involve some form of periodicity, such as the old-fashioned, analog wristwatches.~If working properly, these watches allow us to locate events relative to the position of their pointers, which always return to a given position (unless the watch breaks down).~Of course, practical clocks usually incorporate explicit cycle-counting mechanisms, and thus are not strictly periodic systems.~A sense of continuous flow of time despite the periodicity of the 24-hour day is obtained by the use of calendar days, which serve as ``winding numbers'' that count the repetitions of periodic clocks, and allow us to faithfully track the passage of time.

More generally, the use of harmonic oscillators to approximate the behavior of dynamical systems near their potential minima also shows that physical periodic clocks are widespread.~In fact, periodic clocks might be of use beyond our ordinary, everyday world in the context of classical and quantum gravity.~As there is no preferred time in general relativity, the use of periodic physical fields to track the passage of time (e.g., along an observer's worldline) may be justified in certain situations.~For example, the dynamics of a homogeneous scalar field that is conformally coupled to a closed Friedmann metric can be described in terms of a pair of suitably defined harmonic oscillators \cite{Kiefer:1989va,Kiefer:1990ms,Kiefer:1993cqg}.~In the quantum theory, this leads to a simple but nontrivial model in quantum cosmology.~Periodic clocks  have also recently been employed to begin exploring constraints on the (im)possibility of time travel in relativistic quantum settings \cite{Alonso-Serrano:2023gir} and on the existence of a fundamental period of time \cite{Wendel:2020hqv}, as well as to mimic gravitational time dilation in finite-dimensional quantum systems that are amenable to laboratory implementation \cite{Cafasso:2024zqa}.~Therefore, the development of a systematic formalism to describe a relational quantum dynamics relative to quantum periodic clocks should be rather useful in a series of applications, in laboratory situations and beyond.

In the present article, we present such a systematic treatment of a relational dynamics with periodic clocks, both in the classical and quantum theories.~The presentation is filled with illustrative examples and a particular emphasis is placed on the parallels between the classical and quantum cases.~Concretely, we consider reparametrisation-invariant theories that describe the dynamics of a certain set of degrees of freedom (referred to as the `system') and a set of periodic degrees of freedom.~For simplicity, we assume that there are no interactions between the two sets of degrees of freedom, as this would typically also ruin the periodicity of the clock.~The Hamiltonian is then a sum of Hamiltonians for each set of variables, and it is constrained to vanish due to time-reparametrisation invariance.~We then discuss the construction of relational Dirac observables (observables that are gauge invariant and encode the value of a system quantity relative to a dynamical reference frame) that refer to the periodic clock.

Starting with the classical theory, we first provide a general definition of periodic and aperiodic clocks, and we discuss the definition and use of winding numbers, which can be used to define a monotonic clock from a periodic one.~This monotonic clock can be seen as an ``unwound'' version of the periodic quantity, and it unravels the dynamics relative to the periodic clock.

One of the central results of this article is that the relational observable that describes a system quantity $f$ relative to the unwound clock is generally only invariant during a single cycle of the periodic clock.~This means that such an observable is only transiently invariant, and a true invariant along the entire gauge orbits is obtained only if the described system quantity $f$ is itself periodic.~This holds both classically and in the quantum theory.~In particular, we see that, as far as the invariant relational observables are concerned, winding numbers do not have an intrinsic meaning in a purely relational setting with a periodic clock and without additional counting degrees of freedom.

Quantum periodic clocks are modeled using covariant $\rm{U}(1)$ positive operator-valued measures (POVMs), which can encompass both ideal and non-ideal clocks (i.e., clocks that exhibit quantum states that are not perfectly distinguishable).~In short, periodic quantum clocks are $\rm{U}(1)$-quantum reference frames.~The construction of the quantum relational observables relative to periodic clocks is shown to involve a group averaging \emph{a priori} not with respect to the group generated by the full Hamiltonian constraint (which may be the translation group), but rather only with respect to $\rm{U}(1)$.~One can produce invariants via this partial averaging procedure by selecting system quantities that are themselves periodic.~In this way, one obtains invariant operators that can be interpreted as the physical observables of the clock-neutral picture of Dirac quantisation. 

Instead, averaging over the full group generated by the constraint results (i) in divergences avoided by the partial group average when  the full group is the translation group, and (ii) in well-defined Dirac observables when the full group is $\rm{U}(1)$.~In case (ii) this happens on account of averaging over all configurations that the periodic clock cannot resolve (i.e.\ over its isotropy group), which amounts to automatically projecting the system observables onto their periodic versions and inserting them into the partial group average.~This means that in case (i) the partial group averaging is better behaved and in case (ii) it yields the same gauge-invariant observables as the full group average.~This justifies using only the partial average over a single clock cycle and injecting periodic system observables.

We also show that the clock-neutral picture is equivalent to two other formulations of relational dynamics that are frequently used in the literature:~the Page--Wootters (PW) formalism \cite{pageEvolutionEvolutionDynamics1983,Page:1984qt,woottersTimeReplacedQuantum1984}, which defines a relational Schr\"odinger picture, and the relational Heisenberg picture constructed by symmetry reduction (`quantum deparametrisation').~The equivalence between these three approaches, established via invertible quantum reduction maps, is an extension of the `trinity of relational quantum dynamics' previously established for monotonic clocks \cite{Hoehn:2019owq,HLSrelativistic} to the case of periodic clocks. (For an extension to quantum reference frames associated with general symmetry groups, see \cite{delaHamette:2021oex}.) The `monotonic trinity' has also been expanded to parametrised field theory \cite{Hoehn:2023axh}, was key in understanding recent works on gravitational algebras and entropies in perturbative quantum gravity \cite{Chandrasekaran:2022cip} in terms of quantum reference frames \cite{DeVuyst:2024pop,DVEHK2}, and helped to clarify aspects of time evolution in group field theories \cite{Calcinari:2024pek}, a nonperturbative approach to quantum gravity.

One of the main results of the `periodic trinity' in this work is that the PW prescription to define conditional probabilities must be modified for systems that have a continuous energy spectrum, as the original PW proposal yields ill-defined probabilities in this case.~In fact, we show that for periodic clocks, unlike the case of monotonic clocks, the PW conditional inner product defined as the expectation value of the `projector' onto a clock reading in physical states, but evaluated in the kinematical inner product, is equivalent to the physical inner product only when the group generated by the constraint is compact, and so isomorphic to $\rm{U}(1)$.~In the non-compact case, however, the usual PW conditional inner product diverges.~Therefore, as required by the equivalence with the clock-neutral, gauge-invariant Dirac picture, the correct conditional probability densities should be instead defined in terms of expectation values of physical `projection' operators in the \emph{physical} inner product which is shown to be well-defined also for the non-compact case and to reproduce the already known results in that case.

Finally, given the above differences and subtleties that arise in the relational dynamics with periodic clocks compared to aperiodic ones, we discuss the situation in which both types of clock are present.~In particular, using the quantum reference frame transformation from the perspective of one clock to the perspective of the other, we show how the relational descriptions of the system relative to periodic and aperiodic clocks are compatible with one another as should be expected from them being reductions of the clock-neutral picture.

The paper is organised as follows.~Sec.~\ref{Sec2} is devoted to the discussion of classical relational dynamics.~After introducing the setup of the work and a reparametrisation-invariant, intrinsic definition of periodic and aperiodic clocks, the construction of the unravelled monotonic clock from a periodic one, relational observables, and their transient invariance are discussed.~In Sec.~\ref{sec_covpovm}, we move to the quantum theory and construct covariant POVMs to model quantum periodic clocks.~The clock-neutral Dirac quantised picture and quantum relational observables are discussed in Sec.~\ref{sec_Dirac} and~\ref{sec_observables-periodic}, respectively.~In Sec.~\ref{sec_reduction} we discuss the relational Schr\"odinger picture, obtained via Page-Wootters reduction, and the relational Heisenberg picture, obtained via quantum deparametrisation.~In particular, we demonstrate the equivalence between these two formulations of relational dynamics with periodic clocks and with the clock-neutral Dirac formulation, and utilise it to provide the correct definition of Page-Wootters conditional probabilities.~In Sec.~\ref{sec_clockchanges} we discuss how to switch between the description relative to periodic and aperiodic clocks when both types of clock are present.~We conclude with a comparison between our work and previous literature on periodic clocks in Sec.~\ref{Sec:discussion}.~A brief summary and some final remark are reported in Sec.~\ref{Sec:conclusion}.~The presentation is supplemented with various appendices which contain technical details and proofs of the main results.~Throughout the text, we illustrate our findings with many explicit examples.

\section{Classical relational dynamics with periodic clocks}
\label{Sec2}

\subsection{Preliminaries}

Suppose we are given a reparametrisation-invariant action ${\mathfrak{S}=\int_{\cm} \,ds\,L(q^a,\dot{q}^a)}$ for a composite system on a $D$-dimensional configuration space $\cq_{\rm kin}$, where $\dot q^a$ denotes differentiation with respect to $s$ and $a $ ranges through $1,\ldots, D$.~$\cm$ is a one-dimensional manifold encoding the time direction and reparametrisation-invariance means that the action is invariant under diffeomorphisms of $\cm$:~the Lagrangian transforms as a scalar density $L(q^a,\dot q^a)\mapsto L(q^a,dq^a/d\tilde s)\,d\tilde s/ds$ under a reparametrisation $s\mapsto\tilde s(s)$. Upon Legendre transformation, one finds the Hamiltonian in the form $H=N(s)\,C_H$, where $N(s)$ is an arbitrary (lapse) function and $C_H$ is a so-called Hamiltonian constraint
\begin{align}
C_H=\sum_{a=1}^D\,p_a\,\dot q^a-L(q^a,\dot q^a)\approx0.\nn
\end{align}
It has to vanish on account of the reparametrisation invariance of  $L(q^a,\dot q^a)$ and defines the constraint surface $\cc$ in the kinematical phase space $\cp_{\rm kin}$ (which in this case will simply be $T^*\cq_{\rm kin}$).~The symbol $\approx$ henceforth denotes a so-called \emph{weak equality}, i.e.\ one that only holds on the constraint surface $\cc\subset\cp_{\rm kin}$ \cite{diracLecturesQuantumMechanics1964,Henneaux:1992ig}. We shall see quantum analogs of this later.

We are free to set $N(s)=1$, upon which the Hamiltonian $H$ coincides with the constraint  $C_H$.~The dynamical equations it generates on the kinematical phase space read
\begin{align}
\f{d f}{d s} \ce \{f,C_H\} , \nn 
\end{align}
for an arbitrary function $f: \cp_{\rm kin} \to \mathbb{R}$ and define a flow on $\cp_{\rm kin}$, ${\alpha_{C_H}^s:\mathbb{R}\rightarrow\cp_{\rm kin}}$, with flow parameter $s$.~In any neighbourhood where $f$ is analytic, this flow transforms it as\footnote{For notational simplicity, we are suppressing here the dependence on the phase space point $x\in\cp_{\rm kin}$ in the argument.}
\begin{align}
\label{alpha}
f \mapsto \alpha_{C_H}^s\cdot f &\ce
\sum_{n=0}^{\infty}\,\f{s^n}{n!}\,\{f,C_H \}_n ,
\end{align}
where $\{f,C_H\}_{n+1} \ce \{ \{f,C_H \}_{n},C_H\}$ is the iterated Poisson bracket subject to the convention $\{f,C_H\}_0\ce f$. 

When restricted to $\cc$, this flow constitutes the phase space image of the action of active diffeomorphisms on $\cm$, which are equivalent to the passive diffeomorphisms $s\mapsto\tilde s(s)$.~Since this is a gauge symmetry of the action, any dynamical trajectory is thus also a gauge orbit, in line with the fact that $C_H$ is a first class constraint.~The evolution $f(s)$ in the gauge parameter $s$ is therefore not physical \emph{per se} and in this article we will instead adopt the relational approach \cite{Rovelli:1989jn,Rovelli:1990jm,Rovelli:1990ph,Rovelli:1990pi,rovelliQuantumGravity2004,thiemannModernCanonicalQuantum2008,Dittrich:2005kc,dittrichPartialCompleteObservables2007,Dittrich:2006ee,Dittrich:2007jx,Tambornino:2011vg,Goeller:2022rsx} to constructing a gauge-invariant, i.e.\ reparametrisation-invariant dynamics.~Reparametrisation-invariant information is encoded in Dirac observables $F:\cc\to\mathbb{R}$, which satisfy $\{F,C_H\}\approx0$.~Specifically, we will be interested in so-called relational Dirac observables\,---\,or evolving constants of motion\,---\,which capture how degrees of freedom of interest evolve relative to a choice of a dynamical clock observable along the orbits generated by $C_H$ in $\cc$.~The clock observable will thus define a dynamical coordinate along these orbits and thereby constitute a temporal reference frame for the remaining degrees of freedom. 

The temporal manifold $\cm$ underlying the action $\Fs$ will determine `how far' in the gauge parameter $s$ and, in turn, over which range of dynamical clock readings we may consider the ensuing dynamics.~In particular, when using a periodic clock, the properties of $\cm$ will affect for how many clock cycles we may consider the relational evolution.~We shall require $\cm$ to be (i) connected, so that we will have a continuous `flow of time'; (ii) Hausdorff, so that points (i.e.\ states) on the dynamical orbits are distinguishable; and (iii) without boundary, so that we have a future and past inextendible dynamics rather than a special initial or final endpoint of it (although the latter condition could be easily relaxed).~A standard theorem \cite{onemf} shows that under these conditions, $\cm$ must be homeomorphic to either $\mathbb{R}$ or $S^1$, depending on whether $\cm$ is compact or non-compact.~Given that we have reparametrisation invariance, $\cm$ is then also diffeomorphic to either $\mathbb{R}$ or $S^1$.

\subsection{Decomposition into clock and evolving system}

As indicated above, we will be interested in partitioning our composite system into a clock $C$ and a set of evolving degrees of freedom, constituting a system $S$.~In order to avoid clock pathologies arising from a complicated dynamics \cite{Giddings:2005id,Marolf:1994nz,Hohn:2011us,Smith:2017pwx}, and as often the case in the literature on the Page-Wootters formalism \cite{pageEvolutionEvolutionDynamics1983,Page:1984qt,Hoehn:2019owq,HLSrelativistic,Hoehn:2021wet,giovannettiQuantumTime2015} (see \cite{Smith:2017pwx,Cafasso:2024zqa,castro-ruizTimeReferenceFrames2019} for some interesting exceptions), we shall henceforth assume that no interactions between $C$ and $S$ are present.~While we are therefore not covering the general case, we will be able to prove many explicit results which would be rather challenging in the presence of interactions, especially when they are strong enough to lead to chaos \cite{Hohn:2011us,Dittrich:2016hvj}.~Furthermore, interactions, unless fine-tuned, would typically ruin the periodicity of the dynamics.

The action will then take the form
\ba
\Fs = \Fs_C+\Fs_S = \int_\cm ds\,\left(L_C(q_C^a,\dot{q}_C^a)+L_S(q_S^a,\dot{q}_S^a)\right)\,\nn
\ea
and this also implies a similar form for the constraint
\ba
C_H=H_C+H_S\,.\label{noint}
\ea
Let us further assume that the kinematical phase space has the structure $\cp_{\rm kin} = \cp_C\times \cp_S$, where the system phase space is an arbitrary finite-dimensional symplectic manifold.~The clock phase space $\cp_C$, by contrast, is some two-dimensional phase space since we only need a single clock degree of freedom to provide a dynamical parametrisation of the one-dimensional orbits generated by $C_H$.~The clock and system Hamiltonians $H_C$ and $H_S$ are then functions on $\cp_C$ and $\cp_S$, respectively, and since their equations of motion decouple, we can solve them independently (except that on $\cc$ we have to match a clock dynamics with a given value $H_C=H_C^0$ with a system evolution with energy $H_S=-H_C^0$).~We shall assume them to be autonomous, i.e.\ independent of the gauge parameter $s$, as appropriate for a reparametrisation-invariant model.

An autonomous system on a two-dimensional phase space constitutes a completely integrable system.~Under the assumption that $\cp_C$ is boundary-free and that the flow of $H_C$ generated on it is complete, i.e.\ exists for all values of $s$ necessary to coordinatise $\cm$, Liouville's integrability theorem entails that every constant energy surface of $H_C$ is diffeomorphic to either $S^1$ or $\mathbb{R}$ \cite{Libermann1987}. 

Our subsequent exposition in the quantum theory can readily be generalised to encompass the situation that clock energy levels feature an energy-independent degeneracy.~The classical analog of this is that constant energy surfaces of $H_C$ in $\cp_C$ may be comprised of disconnected pieces and the number of such pieces does not depend on the energy.\footnote{Except possibly for a set of measure zero, such as the non-degenerate $p_t=0$ surface in $\cp_C$ of the otherwise twice-degenerate $H_C=p_t^2$.} Since $\dim\cp_C=2$, this implies that each connected piece will contain a single dynamical orbit.

\subsection{Classical periodic clocks: $\rm{U}(1)$-reference frames}

In this article, we will focus on periodic clocks, so we need to specify what we mean by that.~Owing to the reparametrisation-invariance, it is clear that we cannot define periodicity of the clock with respect to a given gauge parameter $s$; if the dynamics generated by $H_C$ on $\cp_C$ was periodic in the parameter $s$, we could find some other parametrisation $\tilde s(s)$ such that the clock dynamics is no longer periodic with respect to $\tilde s$.~Instead, we need a reparametrisation-invariant manner to say that a clock is periodic. In principle, one could achieve this by defining periodicity of $C$ relationally, i.e.\ here relative to some $S$ degrees of freedom.~However, this would be undesirable as such a notion of periodicity of $C$ would depend on the choice of not only $S$, but also of the periodicity-defining reference degrees of freedom within it.~Furthermore, this would be somewhat circular as, after all, we are interested in the dynamics of $S$ relative to $C$.~Accordingly, we need a way of characterising $C$ as a periodic clock that is both reparametrisation-invariant \emph{and} intrinsic, i.e.\ independent of $S$.~The above mentioned global structure of the clock dynamics provides such a characterisation:

\begin{definition}{\bf (Periodic clock.)}\label{def:perclock}
We shall say that clock $C$ is periodic if the dynamical orbits generated by the autonomous $H_C$ in (an open dense subset\footnote{Otherwise, even the harmonic oscillator would not constitute a periodic clock, as not every solution is diffeomorphic to $S^1$, namely the zero-energy one is not.} of) $\cp_C$ are diffeomorphic to $S^1$.~In other words, $H_C$ acts as a generator for the group $\rm{U}(1)\simeq\rm{SO}(2)$ in $\cp_C$. 

Given the integrability of $C$, we can always find so-called \emph{action-angle variables} $(\phi_C,H_C)$ on $\cp_C$ that are canonically conjugate $\{\phi_C,H_C\}=1$ (on an open dense subset) \cite{arnold1989mathematical}.~The angle variable $\phi_C$ takes value in $[0,\phi_{\rm max})$, for some (possibly energy-dependent) $\phi_{\rm max}$ on each dynamical clock orbit in (a dense subset of) $\cp_C$, which we thus call the \emph{clock period}.~The angle $\phi_C$ defines the reading of the clock and singles out a point on each dynamical orbit in $\cp_C$.~Such a clock function will also be called $\rm{U}(1)$-\emph{covariant} as it transforms uniformly along the orbit.~Hence, a periodic clock $C$ defines a dynamical $\rm{U}(1)$-reference frame.
\end{definition}

This is to be contrasted to the only other possibility consistent with Liouville's integrability theorem (on a two-dimensional phase space), which amounts to an aperiodic clock.
\begin{definition}{\bf (Aperiodic clock.)}\label{def:nonperclock}
We say that a clock $C$ is aperiodic if the orbits generated by the automonous $H_C$ in (an open dense subset\footnote{Otherwise, even the free particle would not constitute an aperiodic clock, as not every solution is diffeomorphic to $\mathbb{R}$, namely the zero-energy one is not.} of) $\mathcal{P}_C$ are diffeomorphic to $\mathbb{R}$.~The reading of the aperiodic clock is given by a phase-space function $Q$ that satisfies $\{Q,H_C\}=1$ (on a dense subset), so that $Q(s) = s+Q^0$.
\end{definition}

For a periodic clock, however, it is important to note that, depending on the shape of $H_S$, the constraint $C_H$ in Eq.~\eqref{noint} need \emph{not} be the generator of a $\rm{U}(1)$ action on $\cc\subset\cp_{\rm kin}$.~For instance, $C_H$ may also generate an action of the translation group $\mathbb{R}$, as we shall see in examples below.~It is also possible that $C_H$ will be a $\rm{U}(1)$ generator on $\cc$, but with a larger period than $\phi_{\rm max}$ which $H_C$ induces on $\cp_C$.~In those cases, a multitude of cycles of clock $C$ will fit into the constraint generated orbits and $C$ will take the same reading multiple times along it.~The relation between the evolving $S$ and the clock $C$ will thus \emph{a priori} be multivalued, posing a potential challenge to the relational dynamics.~In the next subsection, we will explain how to remedy this issue classically, while dealing with the quantum theory in Sec.~\ref{sec_covpovm}.~Remarkably, as we will later see both in the classical and quantum theory, all $S$ degrees of freedom consistent with the constraint turn out to be periodic too so that no multivaluedness will arise.

More generally, given a periodic clock according to this definition, how many of its cycles make up a complete evolution of the composite system depends on its action $\Fs$.~The total number of clock cycles is the number of cycles undergone by $C$ in $\cp_C$ as the gauge parameter $s$ runs once over $\cm$. Due to reparametrisation-invariance, this number of cycles is independent of which parametrisation one chooses.~If $\cm\simeq\mathbb{R}$ (and it does not take infinite parameter time for $C$ to complete one cyclic orbit in $\cp_C$) there will be a countably infinite number of clock cycles, while in the case $\cm\simeq S^1$ this number will typically be finite.~Since $\phi_{\rm max}$ may depend on the clock energy, so too may the number of clock cycles covering a complete evolution on $\cc$. Hence, through the constraint, this number may depend on the system $S$ (see also the discussion in the next subsection).

\begin{example}[\bf Harmonic oscillator]
An obvious example for a non-degenerate periodic clock is a harmonic oscillator
\ba
H_C&=&\f{p_t^2}{2 m_t}+\f{m_t\omega_t^2}{2}\,t^2\,.\label{HOHamiltonian}
\ea
Its phase or angle variable
\ba
\varphi(t,p_t)= \f{1}{\omega_t}\arctan\left(\f{-p_t}{m_t\,\omega_t\,t}\right)\label{HOphase}
\ea
is conjugate to the clock Hamiltonian, $\{\varphi,H_C\}=1$.~As $\arctan x\in(-\frac{\pi}{2},\frac{\pi}{2})$ for $x$ real, notice that we need two branches of $\arctan$ to cover one clock cycle.~More precisely, we may define the angle variable not from Eq.~\eqref{HOphase} but rather from
\ba
\phi_C(t,p_t)= \varphi(t,p_t)+\frac{\pi}{\omega_t}-\frac{\pi}{2\omega_t}\mathrm{sgn}\left(t\right) \,,\label{HOphase2}
\ea
where the sign function satisfies $\lim_{t\to0^\pm}\mathrm{sgn}(t) = \pm1$, and we obtain the limits $\lim_{t\to0^+}\phi_C(t,p_t>0) = 0$, $\lim_{t\to0^{\pm}}\phi_C(t,p_t<0) = \pi/\omega_t$, $\lim_{t\to0^{-}}\phi_C(t,p_t>0) = 2\pi/\omega_t$.~Notice, however, that $\phi_C(t=0,p_t)$ is undefined because $\varphi(t=0,p_t)$ is undefined and $\mathrm{sgn}(0) = 0$.~In this way, the phase variable in Eq.~\eqref{HOphase2} obeys $\{\phi_C,H_C\}=1$ where it is defined and differentiable, and it increases monotonically from $0$ to $2\pi/\omega_t$ in a clock cycle.~Thus, it has an \emph{energy-independent} period $\phi_{\rm max}=2\pi/\omega_t$.~Harmonic oscillators have been used extensively in the literature on relational dynamics, e.g.\ \cite{Rovelli:1990jm,Rovelli:1989jn,rovelliQuantumGravity2004,Bojowald:2010qw,Wendel:2020hqv}, however, not using the angle variable as a clock.~The advantage of the latter is that it is, in constrast to $t$, monotonic for each cycle, thereby avoiding turning points, and it runs over the same values for all (except the $H_C=0$) orbits.
\end{example}

\begin{example}[\bf Particle on a circle]
We can also consider a free particle on a circle with phase space $\cp_C=T^*S^1$, so that we identify $t+1\sim t$, and doubly degenerate Hamiltonian
\ba
H_C=\f{p_t^2}{2m_t}.\label{hcfree}
\ea
The angle variable conjugate to $H_C$ reads
\ba 
\phi_C(t,p_t)=\f{m_t t}{p_t}\label{compactclock}
\ea 
and has an \emph{energy-dependent} period $\phi_{\rm max}=\f{m_t}{p_t}$.~Since $p_t$ is a constant of motion, it is clear that the dynamical orbits are diffeomorphic to $S^1$.
\end{example}
To be clear about the scope of this article, it is also worthwhile to illustrate a periodic system that violates our definition.
\begin{nonexample}[\bf Particle in a box]
A free particle bouncing back and forth between the walls of a box with Hamiltonian
\ba 
H_C=\f{p_t^2}{2m_t}+V(t),\q\text{where}\q V(t)=\begin{cases} 0 &\mbox{if } 0<t<1 \\
+\infty & \mbox{otherwise}, \end{cases}\nn
\ea 
by contrast, does \emph{not} constitute an example for our definition of a periodic clock.~Its dynamical orbits in the phase space $\cp_C=T^*\mathbb{R}\simeq\mathbb{R}^2$ are not diffeomorphic to $S^1$ because the sign of $p_t$ changes during every bounce such that the orbits are discontinuous.~In particular, $H_C$ is not a $\rm{U}(1)$ generator.~Nevertheless, the clock function conjugate to $H_C$ would once more be given by Eq.~\eqref{compactclock} and monotonically and repeatedly run through the range $[-\frac{m_t}{p_t},\frac{m_t}{p_t}]$.~This clock function would thus be periodic with energy-dependent period $\phi_{\rm max}=2m_t/p_t$, however, it would neither be differentiable at $t=0$ nor $t=1$.
\end{nonexample}

In the sequel, we shall restrict to periodic clocks as $\rm{U}(1)$-reference frames so as to enable us to exploit the group structure.

\subsection{Using winding numbers to unravel periodic clocks}

We noted above that the periodicity of the clock leads to an apparent challenge for relational dynamics, namely a multivaluedness of evolving degrees of freedom at a specific clock reading.~It is, however, possible to `unwind' or `unravel' the periodic clock to become a monotonic one, using so-called winding numbers, which we now discuss.~The issue of defining relational observables with respect to the unwound clock is analysed next.

Given a periodic clock according to our definition, it is clear that the evolution of its angle variable reads
\ba
\phi_C(s) &=& \left(s + \phi_C^0\right) \,\,\text{mod} \,\,\phi_{\rm max}\nonumber\\
&=& s + \phi_C^0-\phi_{\rm max}\Big\lfloor \f{s + \phi_C^0}{\phi_{\rm max}}\Big\rfloor,
\label{solution}
\ea
where $\phi_C^0$ is its initial value and $\lfloor\cdot\rfloor$ denotes the floor function. In particular, 
\ba
n:=\Big\lfloor \f{s + \phi_C^0}{\phi_{\rm max}}\Big\rfloor\in\mathbb{Z}\nn
\ea
is the winding number of the clock at parameter time $s$, counted relative to an initial state with angle variable reading $\phi_C=\phi_C^0$.~A change in initial datum can thus induce a shift in the winding number.~Note that the difference between Eq.~\eqref{solution} and the flow of an aperiodic clock according to Definition \ref{def:nonperclock} is precisely the appearance of the $\text{mod } \phi_{\rm max}$ condition or, equivalently, of the floor function.~As an example, the validity of Eq.~\eqref{solution} is illustrated for the harmonic oscillator in Example \ref{ex:HOfloor} in Appendix \ref{app:examples} by using the classical oscillator solutions to compute the flow $\phi_C(s)$ of the angle variable defined in Eq.~\eqref{HOphase2}.

In the case that $\phi_{\rm max}$ depends on the clock energy, note that the initial value $\phi_C^0$ will only be accessible on those orbits with $\phi_C^0\leq\phi_{\rm max}$.~One could remedy this by rescaling the clock function $\tilde\phi_C\ce\f{2\pi}{\phi_{\rm max}}\,\phi_C$, so that $\tilde\phi_C\in[0,2\pi)$ independently of the orbit.~However, in this case the covariance condition would be affected, yielding $\{\tilde\phi_C,H_C\}=\f{2\pi}{\phi_{\rm max}}$ and thus a clock energy-dependent rate of change of the new angle variable along the orbit.~For our purposes, it will be more convenient to work with the $\rm{U}(1)$-covariant $\phi_C$ and it will not be a problem that its range may depend on the orbit.~In fact, later we shall see that the quantum analog of this covariant clock observable will feature an energy-independent range.

We are now in a position to `unwind' the clock and define a monotonic clock function $T$ for the periodic clock on $\cp_C$.~The price we pay is that this monotonic clock function is no longer purely kinematically defined, but depends on the solutions to the equations of motion:
\ba
T(s):=\phi_C(s)+n\,\phi_{\rm max}  = s + \phi_C^0,\label{globalclock}
\ea
which clearly is monotonic along the dynamical orbit.\footnote{For the example of a harmonic oscillator, a similar construction of a monotonic clock function was given in \cite{Wendel:2020hqv}, however, not for the angle variable as here, but for the position variable $t$ in Eq.~\eqref{HOHamiltonian}.~In that case one has to worry about clock energy-dependent turning points.~We refer to Sec.~\ref{Sec:discussion} for a further comparison with the work of \cite{Wendel:2020hqv}.} In particular, when $\cm\simeq\mathbb{R}$, this clock function will run monotonically over all of $\mathbb{R}$ along the clock's cyclic orbits.~Furthermore, if $\phi_{\rm max}$ depends on the clock energy and we keep the initial datum $\phi_C^0$ fixed this definition is dependent on the dynamical orbit of the clock.~However, we are free to leave the initial datum unrestricted and simply replace it with the phase space variable $\phi_C$ in order to obtain a monotonic clock function $T(s)=s+\phi_C$ defined everywhere on the dense subset of $\cp_C$ where $\{\phi_C,H_C\}=1$.~In fact, the dynamically defined $T(s)$ is a covariant function, i.e.\ canonically conjugate to $H_C$ on the same dense subset
\ba
\{T,H_C\}=1\,,\label{classcov}
\ea
and in turn also to the constraint $C_H$.~The unwound $T$ thus constitutes an ideal clock function and we henceforth drop the $0$ label from the initial data.

\subsection{Relational observables for periodic clocks} \label{sss_relobsperclock}

Given the ideal clock function $T$ on $\cp_C$, it is now easy to construct relational observables describing how system properties evolve with respect to it.~The relational observable encoding the value of some system observable $f_S:\cp_S\rightarrow\mathbb{R}$ when the covariant clock function $T$ reads $\tau$ can be conveniently constructed using the power series expansion in Eq.~\eqref{alpha}:
 \ba
F_{f_S,T}(\tau)&:=&\alpha_{C_H}^s\cdot f_S\Big|_{\alpha^s_{C_H}\cdot T=\tau}\nn\\
&\underset{(\ref{alpha})}{=}&\sum_{n=0}^{\infty}\,\f{s^n}{n!}\,\{f_S,C_H \}_n \Big|_{\alpha^s_{C_H}\cdot T=\tau}\nn\\
&=&\sum_{n=0}^{\infty}\,\f{s^n}{n!}\,\{f_S,H_S \}_n \Big|_{\alpha^s_{C_H}\cdot T=\tau}\label{relobs}\\
&=&\sum_{n=0}^{\infty}\,\f{(\tau-\phi_C)^n}{n!}\,\{f_S,H_S \}_n\,,\nn
\ea
where in the last line we have made use of ${\alpha^s_{C_H}\cdot T=s+\phi_C =\tau}$.~This expression requires $f_S$ to be analytic in at least a neighbourhood (\textit{cf}.\ Example~\ref{ex_bad} for an illustration of the importance of analyticity).

Before we explore its invariance properties, we note that Eq.~\eqref{relobs} is an adaptation of the sum representation of relational observables\footnote{There also exists an integral representation of relational observables (see e.g.~\cite{Marolf:1994wh,Giddings:2005id,Chataignier:2019kof,Chataignier:2020fys,ChataignierT}), which is (classically) equivalent to the sum representation used here.} developed in \cite{dittrichPartialCompleteObservables2007,Dittrich:2005kc,Dittrich:2006ee,Dittrich:2007jx} to unravelled periodic clock observables.~There is, however, a slight difference in the shape of the power series construction:~while the clock reading $\tau$ of the monotonic clock function $T$ features here as in \cite{dittrichPartialCompleteObservables2007,Dittrich:2005kc,Dittrich:2006ee,Dittrich:2007jx} (see also \cite{Hoehn:2019owq,HLSrelativistic}), in contrast to these references it is the \emph{periodic} angle variable $\phi_C\in[0,\phi_{\rm max})$\,---\,and not the unravelled $T\in\mathbb{R}$\,---\,that appears in the last line of Eq.~\eqref{relobs}.~This is due to the \emph{dynamical} definition of $T$ in Eq.~\eqref{globalclock}.~Notwithstanding, the observable $F_{f_S,T}(\tau)$ encodes the value of $f_S$ when $T$ reads $\tau$.~But it is not in general fully invariant.

\begin{Lemma}\label{lem_clobs}
For an arbitrary system phase space function $f_S:\cp_S\rightarrow\mathbb{R}$, the relational observables $F_{f_S,T}(\tau)$ satisfy the \emph{transient invariance property}
\ba
\alpha_{C_H}^{s}\cdot F_{f_S,T}(\tau)=\alpha_{C_H}^{z\phi_{\rm max}-\phi_C^0}\cdot F_{f_S,T}(\tau)\,,\label{transientinv}
\ea
with $z\phi_{\rm max}-\phi_C^0\leq s<(z+1)\phi_{\rm max}-\phi_C^0$, for $z\in\mathbb Z$, and
\ba 
\alpha_{C_H}^{z\phi_{\rm max}-\phi_C^0}\cdot F_{f_S,T}(\tau)=F_{f_S,T}(\tau+z\phi_{\rm max}).\q\;\,\label{eq:jump}
\ea
\end{Lemma}
\begin{proof}
\proofapp
\end{proof}
Note that, for a non-zero initial value $\phi_C(0)=\phi_C^0$ of the angle variable, in each period $\phi_C(s)\in[0,\phi_{\rm max})$ of the evolution \eqref{solution} of the angle variable the parameter $s$ runs over the interval $[z\phi_{\rm max}-\phi_C^0,(z+1)\phi_{\rm max}-\phi_C^0)$, for each $z\in\mathbb Z$.~Thus, Eq.~\eqref{transientinv} tells us that the relational observables $F_{f_S,T}(\tau)$ are generically only constant within each clock cycle and their value jumps as the clock completes a cycle according to Eq.~\eqref{eq:jump}.~We refer to Fig.~\ref{Fig:transientplot} in Example~\ref{ex_2} below for a visualisation of the transient invariance property in a simple example.

This has an important consequence:~since the relational observable $F_{f_S,T}(\tau)$ encodes the value of $f_S$ when $T$ reads $\tau$, it can only be constant along the \emph{entire} gauge orbit if $f_S$ takes the same value when $T$ reads $\tau$ as when it reads $\tau+z\,\phi_{\rm max}$. Hence, it must be periodic along the dynamical orbit too. In particular, we say that $f_S$ is $\theta$-periodic if $\alpha^s_{H_S}\cdot f_S = f_S$ for $s=z\theta$, $z\in\mathbb{Z}$ and $\theta\in\mathbb{R}$. Here, $\alpha_{H_S}^s$ denotes the dynamical flow generated by the system Hamiltonian $H_S$ on the system phase space $\mathcal{P}_S$.
\begin{corol}\label{cor_periodicobs}
The relational observable $F_{f_S,T}(\tau)$ is a \emph{transient} Dirac observable that is only invariant along the gauge orbit generated by $C_H$ \emph{per clock cycle}.~Its value will `jump' discontinuously as the clock completes a cycle, $\phi_C\to\phi_{\rm max}$.~The only relational observables that are \emph{global} Dirac observables, i.e.~invariant along the entire orbit generated by $C_H$, are those corresponding to system observables $f_S:\cp_S\rightarrow\mathbb{R}$ that are $\phi_{\rm max}$-periodic too.~Complete gauge invariance thereby enforces the clock's periodicity onto the evolving degrees of freedom.
\end{corol}

Infinitesimally, the transient invariance property of the lemma means that
\ba
\{F_{f_S,T}(\tau),C_H\}=0\nn
\ea
everywhere on the gauge orbit, except at the points where the clock completes a cycle; at these points this Poisson bracket is undefined.~This also follows directly from Eq.~\eqref{classcov}.~In fact, while $\dot T=1$ as clear from \eqref{globalclock}, $\{T,C_H\}=\{\phi_C,H_C\}=1$ except where a clock cycle is completed.

As a consequence of Lemma~\ref{lem_clobs}, it makes no difference to define the relational dynamics relative to the monotonic clock $T$ or relative to the non-monotonic clock $\phi_C$ as the system relational observables with respect to the monotonic clock $T$ are generically only invariant per clock cycle.~In fact, let $\tau_C\in[0,\phi_{\rm max})$ be the (possibly energy-dependent) evolution parameter analogous to $\tau$ that, however, runs over the values of $\phi_C$ along the evolution.~Thanks to Eq.~\eqref{globalclock}, we can split the monotonic evolution parameter $\tau$ into a continuous and a discrete part:
\ba
\tau=\tau_C +n\,\phi_{\rm max}\,,\q\q n\in\mathbb{Z}\,.\label{shift}
\ea
At an intuitive level, the winding number $n$ and $\tau_C$ are akin to counting the `calendar days' and parametrising the readings of a 24h clock, respectively. 

Inserting this relation into Eq.~\eqref{relobs}, we obtain
\ba
F_{f_S,T}(\tau)&=&\sum_{m=0}^{\infty}\,\f{(\tau_C+n\,\phi_{\rm max}-\phi_C)^m}{m!}\,\{f_S,H_S \}_m\nn\\
&=&F_{f_S,\phi_C}(\tau_C,n),\label{relobsrel}
\ea
where $F_{f_S,\phi_C}(\tau_C,n)$ is the relational observable encoding how $f_S$ evolves relative to the non-monotonic clock variable $\phi_C$ on the cycle of the dynamics determined by winding number $n$.~The relational evolution of $f_S$ relative to $T$ or relative to $\phi_C$ depends generally on the winding number and accordingly, the relational observables \eqref{relobsrel} capture transient information for a given cycle $n$.

\begin{example}[\bf{Two oscillators}]\label{ex_1}
Consider two harmonic oscillators with fixed total energy $E\in\mathbb{R}$, i.e.\ the constraint $C_H$ in Eq.~\eqref{noint} with the following clock and system Hamiltonians:
\ba
H_C&=&\f{p_t^2}{2 m_t}+\f{m_t\omega_t^2}{2}\,t^2\,,\nn\\
H_S&=&\f{p^2}{2 m}+\f{m\omega^2}{2}\,q^2-E\,.\nn
\ea
For the special case that $\omega_t=\omega$, this example has been studied extensively in the literature on relational dynamics \cite{Rovelli:1990jm,Rovelli:1989jn,rovelliQuantumGravity2004,Bojowald:2010qw}, however, not using angle variables as a clock.\footnote{Similar settings with two harmonic oscillators also occur in cosmology as e.g.~for a Friedmann universe with a conformally coupled scalar field \cite{Kiefer:1989va,Kiefer:1990ms,Kiefer:1993cqg}.~There, the total Hamiltonian constraint takes the form $H_C-H_S$, with $H_{C,S}$ the Hamiltonians of unit mass and frequency harmonic oscillators and $E=0$.~Similar conclusions as ours here apply also to such models.~See \cite{Chataignier:2024ley} for an alternative formalism.} Here we will expressly permit $\omega_t\neq\omega$ as this will lead to interesting repercussions in the quantum theory, especially when $\omega_t/\omega\notin\mathbb{Q}$.

We choose
\ba
T(s) = s+\phi_C(t,p_t)\,\label{HOclock}
\ea
as our monotonic clock function, where $\phi_C$
is the angle variable given in Eq.~\eqref{HOphase2}.~We can ask for the position $q$ and momentum $p$ of the second oscillator when the clock function $T$ reads $\tau$.~The corresponding relational observables can be computed according to Eq.~\eqref{relobs}
\ba
F_{q,T}(\tau)&=&q\,\cos\left((\phi_C-\tau)\omega\right)-\f{p}{m\omega}\,\sin\left((\phi_C-\tau)\omega\right)\,,\nn\\
F_{p,T}(\tau)&=&p\,\cos\left((\phi_C-\tau)\omega\right)+m\omega\,q\,\sin\left((\phi_C-\tau)\omega\right),\nn\\\label{horelobs}
\ea
and are canonically conjugate $\{F_{q,T}(\tau),F_{p,T}(\tau)\}=1$.~Using Eq.~\eqref{solution} with $\phi_{\rm max}=\frac{2\pi}{\omega_t}$ (as shown in Example \ref{ex:HOfloor} in Appendix \ref{app:examples}), it is easy to check that the relational observables \eqref{horelobs} satisfy the transient invariance property
\begin{align*}
\alpha_{C_H}^s\cdot F_{q,T}(\tau)&=q\,\cos\left((\phi_C-\tau-z\phi_{\rm max})\omega\right)\\
&\q-\f{p}{m\omega}\,\sin\left((\phi_C-\tau-z\phi_{\rm max})\omega\right)\\
&=F_{q,T}(\tau+z\phi_{\rm max})\,,\\
\alpha_{C_H}^s\cdot F_{p,T}(\tau)&=F_{p,T}(\tau+z\phi_{\rm max})\,,
\end{align*}
for $s\in[z\phi_{\rm max}-\phi_C^0,(z+1)\phi_{\rm max}-\phi_C^0)$, $ z\in\mathbb Z$.~Note that $F_{q,T}(\tau+z\phi_{\rm max})=F_{q,T}(\tau)$ and $F_{p,T}(\tau+z\phi_{\rm max})=F_{p,T}(\tau)$ for $\omega_t/\omega\in\mathbb Q$, $0<\omega_t/\omega\leq 1$.~In this case the system observables $q$ and $p$ are periodic by (a unit fraction of) $\phi_{\rm max}$ and, compatibly with Corollary~\ref{cor_periodicobs}, the relational observables \eqref{horelobs} are global Dirac observables.~This is compatible with the fact that, for commensurate frequencies $\omega_t/\omega\in\mathbb Q$, the system is fully integrable \cite{Dittrich:2016hvj,Dittrich:2015vfa}.~On the contrary, when the frequencies are incommensurate, $\omega_t/\omega\notin\mathbb{Q}$, the relational observables \eqref{horelobs} are only invariant along the gauge orbits per clock cycle and their values jump discontinuously as the clock completes a cycle.~In this case, there are no non-trivial $\phi_{\rm max}$-periodic $S$-observables.

As we will see in Sec.~\ref{ssec_ex2} and~\ref{sec_observables-periodic} (cf.~Examples~\ref{ex_3},~\ref{ex_4} and~\ref{ex_7},~\ref{ex_8}), a similar situation occurs also in the Dirac quantised theory.
\end{example}

\begin{example}[\bf{Oscillator clock and free particle}]\label{ex_2}
Let us choose $\cm\simeq\mathbb{R}$ and consider a harmonic oscillator clock with $H_C$ given by Eq.~\eqref{HOHamiltonian} and a free particle as the system with Hamiltonian
\ba
H_S&=&-\f{p^2}{2m}\,\nn
\ea
(the minus sign, so that we can solve Eq.~\eqref{noint}).~Again, we choose the monotonic clock function in Eq.~\eqref{HOclock}.~The relational observable encoding the position $q$ of the particle when the clock function $T$ reads $\tau\in\mathbb{R}$ is given by
\ba
F_{q,T}(\tau) &=& q-\f{p}{m}(\tau-\phi_C(t,p_t))\nn\\
&=&q-\f{p}{m} (\tau_C+\f{2\pi}{\omega_t} n-\phi_C(t,p_t))\label{hopartobs}\\
&=&F_{q,\phi_C}(\tau_C,n)\nn
\ea
where $\tau_C\in[0,\f{2\pi}{\omega_t})$ and $n$ denotes the winding number of the clock at parameter time $s$.~It is easy to check that
\ba
\alpha_{C_H}^s\cdot F_{q,T}(\tau)=q-\f{p}{m}(\tau-\phi_C^0+\f{2\pi}{\omega_t} n)\,,\label{ex:flowFq}
\ea
whose step function-like behaviour as a function of $s$ is plotted in Fig.~\ref{Fig:transientplot}.~The relational observable $F_{q,T}(\tau)$ in Eq.~\eqref{hopartobs} is thus a transient Dirac observable whose value remains constant within each clock cycle and jumps discontinuously as a cycle is completed.
\begin{figure}[t!]
  \centering
\includegraphics[width=0.3\textwidth]{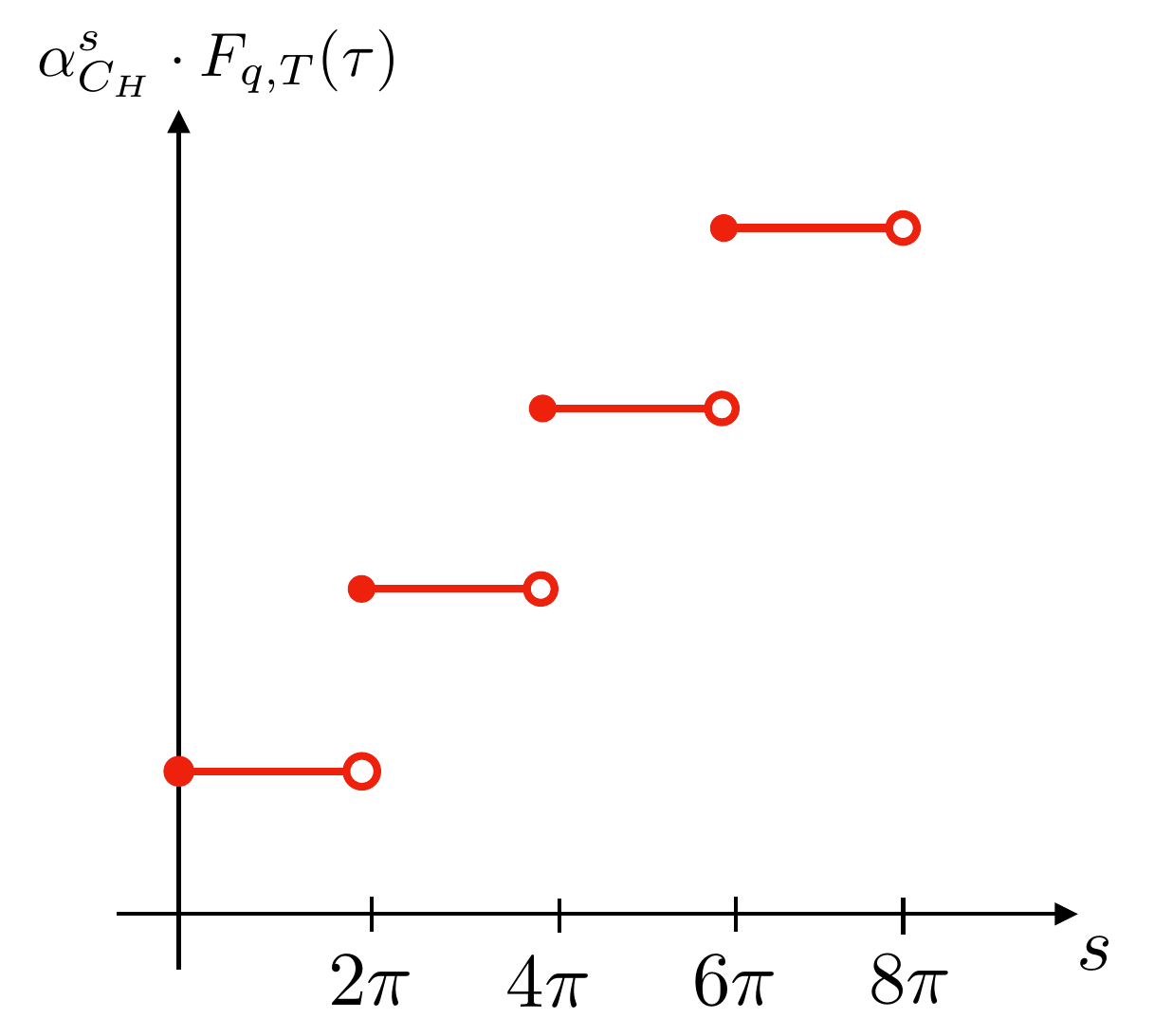}
  \caption{Plot of the flow $\alpha_{C_H}^s\cdot F_{q,T}(\tau)$ given in Eq.~\eqref{ex:flowFq}.~Here, $\phi_C^0=0$, and $\omega_t=1$ so that $\phi_{\rm max}=2\pi$.~The relational observable $F_{q,T}(\tau)$ in Eq.~\eqref{hopartobs} thus remains invariant within each clock cycle and its value jumps discontinuously as the clock completes a cycle (transient Dirac observable).}
  \label{Fig:transientplot}
\end{figure}
It is also clear that the relational observable $F_{p,T}(\tau)=p$ conjugate to $F_{q,T}(\tau)$ is a strong Dirac observable, $\{F_{p,T}(\tau),C_H\}=0$ for all $\tau$.~The periodicity requirement of Corollary~\ref{cor_periodicobs} is in fact trivially met for the constant system observable $p$.~For a less trivial example of a Dirac observable which remains invariant along the entire orbit generated by $C_H$, let us consider the $S$ observable $f_S=\cos(m\omega_tq/p)$.~This is $2\pi/\omega_t$-periodic along the flow generated by $C_H$ and the corresponding relational observable $F_{f_S,T}(\tau)$ satisfies
\ba
\alpha_{C_H}^s\cdot F_{f_S,T}(\tau)=\cos(\omega_t(\f{mq}{p}-\tau+\phi_C^0))=F_{f_S,T}(\tau)\q\forall\,s\,.\nn
\ea
\end{example}

Lastly, let us remark the importance of working with analytic functions $f_S$ when using the power series expansion Eq.~\eqref{relobs} for the relational observables $F_{f_S,T}(\tau)$.~When $f_S$ is not analytic, the power series expansion might in fact fail to give the correct result.~We refer to Example~\ref{ex_bad} in Appendix~\ref{app:examples} for a concrete illustration.~Since we will use this power series later to quantise the relational dynamics with respect to periodic clocks, we will henceforth restrict to systems $S$ which feature a Poisson subalgebra $\ca_S$ of analytic functions that also separates the points in $\cp_S$ and so can be used to coordinatise $\cp_S$.

\subsection{The unravelled (monotonic) clock as a carrier of unphysical information} \label{sUnravClass}
As we have discussed, the relational dynamics relative to the unravelled (monotonic) clock $T$ is the same as that relative to the non-monotonic, periodic clock $\phi_C$, due to the transient invariance property of the system relational observables (a consequence of Lemma~\ref{lem_clobs}).~This implies that the extra information carried by the monotonic clock (the winding number) is irrelevant to the relational dynamics.~The reason for this is that $T$ depends explicitly on the parameter time $s$ [cf.~Eq.~\eqref{globalclock}]; i.e., it is not solely a function of the kinematical phase space variables, and thus it depends on extrinsic information that is not available to the $\rm{U}(1)$-reference frames.~The periodic clock can only keep track of $s$ along one cycle and cannot resolve the winding number information.

Intuitively, one might expect that the winding number could be resolved relative to the clock of an aperiodic reference frame, such as an ideal clock or a free particle.~However, as Lemma~\ref{lem_monoPNP} and Example~\ref{ex:monoPNP} below show, the relational observable that describes the value of the unravelled (monotonic) clock $T$ of a periodic system relative to an aperiodic clock is still a transient Dirac observable; i.e., it is only an invariant per clock cycle.~Thus one must explicitly include the physical mechanism by which the aperiodic clock counts the clock cycles, (which constitutes a modification of the constraint) if the counting is to feature in the relational dynamics.~This is complementary to Lemma~\ref{lem_clobs}, which concerned the observables relative to the unravelled (monotonic) clock $T$ of a periodic system.~In this way, the winding number of a periodic system is not resolved by invariants, and it remains an extrinsic (unphysical) element that is not captured by the relational dynamics. Only an external clock reference frame can access the winding number, namely one that can measure the external evolution parameter $s$. Reparametrisation invariance ensures that any such information is wiped out in a purely relational description. In gravitational scenarios, such an external clock would also be a fictitious one, while in laboratory situations one can imagine a 
\emph{physical} external laboratory clock frame that can measure the winding numbers, however, that  one chooses not to access, e.g.\ to simulate a relational dynamics as in \cite{moreva2014time}. More generally, gauge invariance in the context of internal reference frames as the clocks here can be understood as `external reference frame independence', see \cite[Sec.~II]{Hoehn:2023ehz} for a discussion.

Concretely, the fact that a periodic system's winding number is relationally irrelevant is reflected in the fact that the relational observables, both in Lemmas~\ref{lem_clobs} and~\ref{lem_monoPNP}, depend on $\phi_C$, which is insensitive to the winding number information.~Indeed, the flow $\phi_C(s)$ given in Eq.~\eqref{solution} is not differentiable with respect to all values of $s$ due to the appearance of the floor function (which can be explicitly seen for the harmonic oscillator phase variable in Example~\ref{ex:HOfloor} in Appendix~\ref{app:examples}).

\begin{Lemma}\label{lem_monoPNP}
The relational observable $F_{T,Q}(\tau)$ that encodes the value of the unravelled (monotonic) clock $T$ of a periodic system [cf.~Eq.~\eqref{globalclock}] relative to the value $\tau$ of the clock $Q$ of an aperiodic system [cf.~Definition~\ref{def:nonperclock}] satisfies the \emph{transient invariance property}
\ba
\alpha_{C_H}^{s}\cdot F_{T,Q}(\tau)=\alpha_{C_H}^{z\phi_{\rm max}-\phi_C^0}\cdot F_{T,Q}(\tau)\,,
\ea
with $z\phi_{\rm max}-\phi_C^0\leq s<(z+1)\phi_{\rm max}-\phi_C^0$, for $z\in\mathbb Z$, and
\ba 
\alpha_{C_H}^{z\phi_{\rm max}-\phi_C^0}\cdot F_{T,Q}(\tau)=F_{T,Q}(\tau-z\phi_{\rm max}).\q\;\,
\ea
\end{Lemma}
\begin{proof}
\proofapp
\end{proof}

\begin{example}[\bf Harmonic oscillator, free particle, and ideal clock]\label{ex:monoPNP}
Let us consider the case in which $\cm\simeq\mathbb{R}$ and the Hamiltonian constraint reads
\ba
C_H = H_O+H_P+H_I \; ,
\ea
where $H_O$ is the oscillator Hamiltonian
\ba\label{monoPNP-HO}
H_O = \frac{p_1^2}{2m_1}+\frac{m_1\omega^2q_1^2}{2} \,,
\ea
while $H_P$ and $H_I$ are the Hamiltonians of the free particle and ideal clock, respectively:
\begin{align}
H_P &= \frac{p_2^2}{2m_2} \,,\label{monoPNP-HP}\\
H_I &= -p_3 \ .\label{monoPNP-HI}
\end{align}
The (monotonic) clocks for each system are
\begin{align}
T_O(s) &= s + \phi_C(q_1,p_1) \,, \label{TO-mono}\\
T_P(s) &= \frac{m_2 q_2(s)}{p_2(s)} = s+\frac{m_2q_2}{p_2} \,,\label{TPclock}\\
T_I(s) &= -q_3(s) = s-q_3 \label{TIclock}\,,
\end{align}
where $\phi_C(q_1,p_1)$ was defined in Eq.~\eqref{HOphase2} (with the correspondence $t\leftrightarrow q_1$, $p_t\leftrightarrow p_1$).~Each clock is canonically conjugate to $C_H$ in the phase-space regions where they are differentiable.~Notice that, contrary to $T_P(s)$ and $T_I(s)$, $T_O(s)$ cannot be written as a phase-space function without an explicit dependence on $s$; i.e.,
\begin{equation}\label{TO-mono2}
\begin{aligned}
T_O(s) &\neq \phi_C(q_1(s),p_1(s))\\
&= s+\phi_C(q_1,p_1)-\frac{2\pi}{\omega}\left\lfloor\frac{s+\phi_C(q_1,p_1)}{2\pi/\omega}\right\rfloor \,.
\end{aligned}
\end{equation}
The explicit $s$ dependence in Eq.~\eqref{TO-mono} is related to the extrinsic (winding number) information in $T_O(s)$, precisely what is subtracted by the floor function in the second line of Eq.~\eqref{TO-mono2}.~We can consider the relational observables
\begin{align}
    F_{T_O,T_P}(\tau) &= \tau-\frac{m_2q_2}{p_2}+\phi_C(q_1,p_1) \,,\label{FTOTP}\\
    F_{T_O,T_I}(\tau) &= \tau+q_3+\phi_C(q_1,p_1) \,,\label{FTOTI}
\end{align}
which encode the evolution of $T_O$ relative to the clocks of the aperiodic systems.~Using Eq.~\eqref{solution} with $\phi_{\rm max} = 2\pi/\omega$, we find
\begin{align}
\alpha_{C_H}^s\cdot F_{T_O,T_P}(\tau) &= \tau-\frac{m_2q_2}{p_2}+\phi_C^0-\frac{2\pi}{\omega}\left\lfloor\frac{s+\phi_C^0}{2\pi/\omega}\right\rfloor \notag\,,\\
\alpha_{C_H}^s\cdot F_{T_O,T_I}(\tau) &= \tau+q_3+\phi_C^0-\frac{2\pi}{\omega}\left\lfloor\frac{s+\phi_C^0}{2\pi/\omega}\right\rfloor \notag\,.
\end{align}
With $n\phi_{\rm max}\leq s+\phi_C^0< (n+1)\phi_{\rm max}$ ($n\in\mathbb{Z}$), the above observables obey the transient property:
\begin{equation}\label{FTOtransient}
\begin{aligned}
\alpha_{C_H}^s\cdot F_{T_O,T_P}(\tau) &= F_{T_O,T_P}(\tau-n\phi_{\rm max}) \,,\\
\alpha_{C_H}^s\cdot F_{T_O,T_I}(\tau) &= F_{T_O,T_I}(\tau-n\phi_{\rm max}) \,.
\end{aligned}
\end{equation}
\end{example}

\subsection{Reduced phase space}\label{sec_redps}
Usually, in the presence of globally monotonic clock functions, it is also possible to construct a reduced phase space through gauge fixing \cite{hoehnHowSwitchRelational2018,Hoehn:2018whn,Hoehn:2019owq,HLSrelativistic,ChataignierT} (see also \cite{Henneaux:1992ig}).~For example, when such a clock function features in the expression of relational observables, one can gauge-fix the latter by fixing the clock function to some arbitrary value and solving the constraint $C_H$ for the variable conjugate to the clock, thereby entirely removing the clock degrees of freedom from among the dynamical variables.~However, in the present periodic clock case, despite having the monotonic clock function $T(s)$ in Eq.~\eqref{globalclock} at our disposal, the fact that the angle variable $\phi_C$ -- not $T$ -- appears in the expression Eq.~\eqref{relobs} of the relational observables poses some challenges.

For systems which do not feature a Poisson subalgebra of point separating, analytic, and $\phi_{\rm max}$-periodic functions $f_S$ on $\cp_S$, the relational observables $F_{f_S,T}(\tau)$ in Eq.~\eqref{relobs} are in fact not invariant along the entire gauge orbit but only per clock cycle.~Therefore, they cannot be used to parametrise the space of gauge orbits $\cc/\!\!\sim$, where $\sim$ denotes equivalence under gauge transformations\footnote{Also, when the analyticity assumption on $S$ observables is violated, the space of orbits will generally fail to be a phase space with symplectic structure.~For example, for the Example~\ref{ex_bad} discussed in Appendix~\ref{app:examples} it has been shown in \cite{Dittrich:2015vfa} that its corresponding $\cc/\!\!\sim$ fails to be a manifold and thus a phase space.}.~This in turn reflects into the fact that one can not define a global gauge-fixing of the relational observables in Eq.~\eqref{relobs}.~Fixing $\phi_C$ to some value $\phi_C=\phi_C^*\in[0,\phi_{\rm max})$ will in fact define multiple points on the orbit generated by $C_H$, namely
\ba\label{multipleptsGF}
s=\phi_C^*-\phi_C^0+n\phi_{\rm max}\;,
\ea
as for Eq.~\eqref{solution}.~As such, it can not be used to gauge fix $F_{f_S,T}(\tau)$ which will still exhibit a dependence on the clock winding number $n$, namely $F_{f_S,T(n)}(\tau)=F_{f_S,\phi_C^*}(\tau_C,n)$ as for Eq.~\eqref{relobsrel} with $T(n)=\phi_C^*+n\phi_{\rm max}$ (cf.~Eq.~\eqref{globalclock}).~This is for instance the case for the relational observables in Eq.~\eqref{horelobs} of Example~\ref{ex_1} for incommensurate frequencies and for the relational observable in Eq.~\eqref{hopartobs} of Example~\ref{ex_2} as schematically illustrated for the latter in Fig.~\ref{fig:noglobGF}.
\begin{figure}[t!]
  \centering
\includegraphics[width=0.3\textwidth]{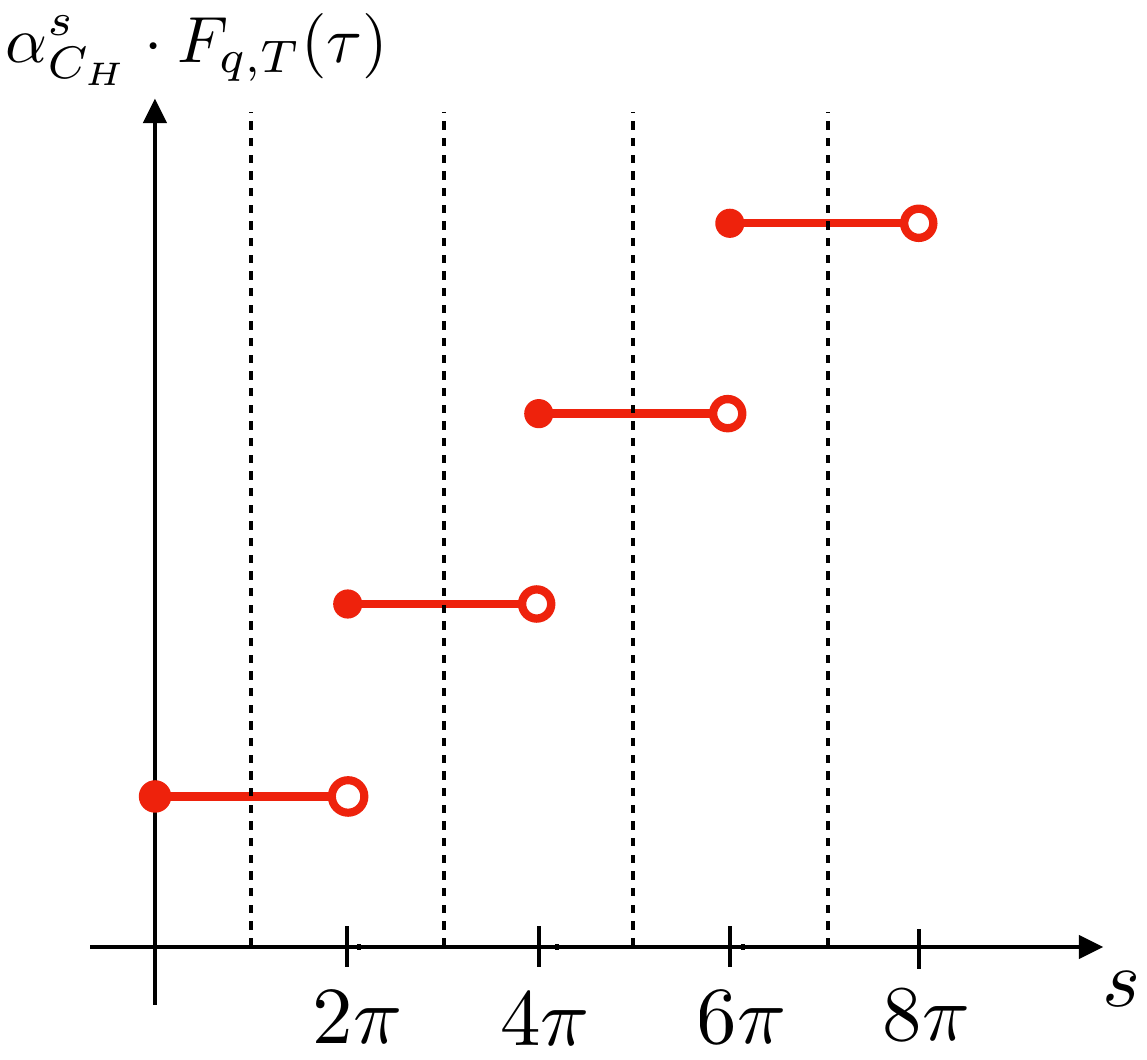}
  \caption{Fixing $\phi_C=\phi_C^*\in[0,\phi_{\rm max})$ defines multiple points on the dynamical orbit of the transient relational observable $F_{q,T}(\tau)$ given in Eq.~\eqref{hopartobs} of Example~\ref{ex_2}.~No global gauge fixing can be defined and different values $F_{q,\phi_C^*}(\tau_C,n)$ of $F_{q,T}(\tau)$ are thus singled out in different clock cycles.~Here, $\phi_C^0=0$, $\omega_t=1$ (hence $\phi_{\rm max}=2\pi$), and $\phi_C^*=\pi$ so that the different points correspond to $s=(2n+1)\pi$ (dashed lines).}
  \label{fig:noglobGF}
\end{figure}

When the system under consideration admits instead a point separating Poisson subalgebra of analytic functions on $\cp_S$ which are also $\phi_{\rm max}$-periodic, no issue arises when constructing the physical phase space as the space of gauge orbits $\cc/\!\!\sim$.~In this case, in fact, Eq.~\eqref{relobs} will provide sufficiently many well-behaved relational observables that, being constant on each orbit, can be used as coordinates on $\cc/\!\!\sim$.~The latter also inherits a symplectic structure through the Poisson brackets of the Dirac observables on $\cc$ \cite{Henneaux:1992ig}.~This is for instance the case of commensurate frequencies for the two oscillators in Example~\ref{ex_1} where we have seen that $F_{q,T}(\tau)$ and $F_{p,T}(\tau)$ in Eq.~\eqref{horelobs} provide us with a pair of canonically conjugate Dirac observables that can be used as canonical coordinates for the two-dimensional space $\cc/\!\!\sim$.

Moreover, the Dirac observables $F_{f_S,T}(\tau)$ encoding the value of $\phi_{\rm max}$-periodic, analytic $S$ observables $f_S$ when the clock $T$ takes the value $\tau$ do not depend on the winding number.~In fact, since for such observables $F_{f_S,T}(\tau+z\phi_{\rm max})=F_{f_S,T}(\tau)$ for any $z\in\mathbb Z$, from Eq.~\eqref{relobsrel} it is clear that
\ba
F_{f_S,T}(\tau)=F_{f_S,T}(\tau-n\,\phi_{\rm max})=F_{f_S,\phi_C}(\tau_C)\,,\nn
\ea
with $\tau_C$ the continuous part of the evolution parameter $\tau$ as given in Eq.~\eqref{shift}.~As such, these Dirac observables can be gauge-fixed by fixing the clock function to some arbitrary value $\phi_C=\phi_C^*\in[0,\phi_{\rm max})$.~This leads to the gauge-fixed relational observables $F_{f_S,\phi_C^*}(\tau_C)$ which no longer feature any clock degrees of freedom among the dynamical variables and take the same value at all points \eqref{multipleptsGF} identified by $\phi_C=\phi_C^*$ along the orbit generated by $C_H$.~Thus, if the system under consideration features a Poisson subalgebra of point separating, analytic, $\phi_{\rm max}$-periodic functions on $\cp_S$, doing such a reduction we obtain a gauge-fixed reduced phase space that is isomorphic to $\cc/\!\!\sim$ and features standard Hamiltonian equations of motion for the gauge-fixed relational Dirac observables in the time parameter $\tau_C$ with symplectic structure obtained through the Dirac bracket \cite{diracLecturesQuantumMechanics1964,Henneaux:1992ig,Vanrietvelde:2018pgb,Vanrietvelde:2018dit,hoehnHowSwitchRelational2018,Hoehn:2018whn,dittrichPartialCompleteObservables2007,Dittrich:2005kc,Thiemann:2004wk}.

This closes our analysis of classical relational dynamics with periodic clocks.~The above mentioned subtleties with periodic clocks have consequences also at the quantum level, both for the Dirac-quantised theory obtained by directly quantising the relational observables in Eq.~\eqref{relobs} and for the reduced phase space quantisation.~In the remainder of this work we shall focus on the former.

\section{Covariant POVMs for periodic clocks}\label{sec_covpovm}

As a first step into the construction of the quantum theory, we shall model quantum periodic clocks via $\rm{U}(1)$-covariant positive operator-valued measures (POVMs).~The construction of covariant POVMs for periodic clocks has been partially discussed in \cite{braunsteinGeneralizedUncertaintyRelations1996,Hoehn:2019owq} (see also the construction of covariant phase observables in \cite{buschOperationalQuantumPhysics}).~ Recently, such POVMs have been employed for harmonic oscillator clocks in the context of exploring constraints on time travel  \cite{Alonso-Serrano:2023gir} and in several discussions using finite-dimensional quantum clocks that are also periodic, e.g.\ see \cite{Smith:2017pwx,Cafasso:2024zqa,Favalli:2020gmx}.~Here we expound on their properties; our exposition applies to both finite- and infinite-dimensional periodic clocks.

We denote the clock Hilbert space by $\ch_C$.~Being the quantisation of a $\rm{U}(1)$-generator, the spectrum of the clock Hamiltonian $\hat H_C$ will have to be discrete and its eigenvectors pairwise commensurable in order to give rise to a projective unitary representation of $\rm{U}(1)$ (i.e.\ a representation up to phase) in the quantum theory~\cite{Hoehn:2019owq}:
\ba
U_C(t_{\rm max})=e^{-i\varphi}\,I_C\,,\q\q\varphi\in [0,2\pi)\,,\label{projrep}
\ea
where\footnote{We use units where $\hbar=1$ in this manuscript.} $U_C(t)\ce\exp(-i\,t\,\hat H_C)$, $I_C$ is the identity on $\ch_C$, and $t_{\rm max}$ is the period of the clock.~Indeed, this equation entails
\begin{align}
e^{-i\varepsilon_j \,t_{\rm max}} = e^{-i\varphi},\quad  \forall\,\varepsilon_j \in \spec(\hat H_C),
\label{point}
\end{align}
which, in turn leads to
\ba \label{eClockEnergies}
\varepsilon_j= \omega_{t} \left( n_j + \frac{\varphi}{2\pi} \right), \quad \forall j .
\ea
where $n_j\in\mathbb{Z}$ and $\omega_{t}\ce 2\pi/t_{\rm max}$. Hence, the spectrum of $\hat H_C$ must be both discrete and rational.~Note that the global phase $\varphi$ is unique only up to multiples of $2\pi$.~We are therefore free to choose $\varphi$ such that $n_{j^*}=0$ occurs in the label set for some $j=j^*$.~In particular, since $\spec(\hat H_C)$ will typically be bounded below, we can interpret $\varepsilon_{j^*}=\varphi/t_{\rm max}$ as the zero-point energy of the clock.

Let us now construct a POVM for the angle observable $\phi_C$ which is $\rm{U}(1)$-covariant.~Assuming $\spec(\hat H_C)$ to be non-degenerate,\footnote{This can be readily generalised to the case where each eigenvalue of $\spec(\hat H_C)$ has the same degree of degeneracy by repeating the formalism described here per degeneracy sector, as in~\cite{HLSrelativistic,DVEHK2}.} we can define clock states for $\phi\in[0,t_{\rm max})$
\begin{align}
\ket{\phi} = \sum_{\varepsilon_j \in \spec(\hat H_C)}  e^{ig(\varepsilon_j)} e^{-i\varepsilon_j \phi} \ket{\varepsilon_j }, \label{discreteClockState}
\end{align}
where $g(\varepsilon_j)$ is an arbitrary real function encoding a freedom in the choice of clock states.~The states \eqref{discreteClockState} clearly satisfy 
\ba
U_C(\phi')\ket{\phi} = \ket{\phi+\phi'}\,,\label{statecov}
\ea
and, owing to Eq.~\eqref{projrep}, are periodic up to the phase $\varphi$.~In later sections, we will use $\ket{\tau}\equiv\ket{\phi=\tau}$ to represent the specific clock reading with respect to which relational quantities are defined. Note that the clock states are nothing but generalised (possibly distributional) coherent states for $\rm{U}(1)$.

Defining the effect densities as
\ba\label{eq:effectdensities}
E_{\phi_C}(d\phi)\ce\f{1}{t_{\rm max}}d\phi\,\ket{\phi}\!\bra{\phi},
\ea
we can define the effect operators as
\ba\label{eq:effectoperators}
E_{\phi_C}(X)\ce\int_{X\subset [0,t_{\rm max})} E_{\phi_C}(d\phi),
\ea
where $X$ denotes some Borel subset on the clock readings. These effect operators are bounded and self-adjoint \cite{buschOperationalQuantumPhysics}.
Using Eq.~\eqref{discreteClockState} and the identity
\ba
\f{1}{t_{\rm max}}\int_0^{t_{\rm max}}d\phi\,e^{-i\,\phi(\varepsilon_i-\varepsilon_j)}=\delta_{\varepsilon_i,\varepsilon_j},\label{id2}
\ea
which follows from Eq.~\eqref{point} and where the symbol on the r.h.s.\ is a Kronecker-delta, it is straightforward to verify that
\ba
E_{\phi_C}([0,t_{\rm max}))=\f{1}{t_{\rm max}}\int_0^{t_{\rm max}}d\phi\,\ket{\phi}\!\bra{\phi}=I_C.\label{resolid}
\ea
Hence, the effect operators, which are positive semi-definite, define a resolution of the identity and thanks to Eq.~\eqref{statecov} we also have
\ba
E_{\phi_C}(X+\phi)=U_C(\phi)\,E_{\phi_C}(X)\,U_C^\dag(\phi).
\ea
Furthermore, on account of Eq.~\eqref{point}
\ba
\ket{\phi}\!\bra{\phi}=\ket{\phi+z t_{\rm max}}\!\bra{\phi+z t_{\rm max}}\,,\q\q\forall\,z\in\mathbb{Z}\,,\label{id}
\ea
so that
\ba
E_{\phi_C}(X+z t_{\rm max})=E_{\phi_C}(X)\,,\q\q\forall\,z\in\mathbb{Z}.
\ea
We have therefore constructed a $\rm{U}(1)$-covariant and $t_{\rm max}$-periodic clock POVM. 

In the sequel, we shall model the periodic clock in the quantum theory through this POVM and its associated $n^{\rm th}$-moment operators:
\ba
\hat \phi_C^{(n)} = \f{1}{t_{\rm max}}\int_0^{t_{\rm max}}\,d\phi\,\phi^n\,\ket{\phi}\!\bra{\phi}\,.\label{nthmoment}
\ea
In particular, the zeroth moment is just the identity, ${\hat\phi_C^{(0)}=I_C}$. In Sec.~\ref{sec_observables-periodic}, the $n^{\rm th}$-moment operators will assume the role of the $n^{\rm th}$ power of the angle observable $\phi_C$ -- appearing in the power series representation of relational observables Eq.~\eqref{relobs} -- in the quantum theory. The following result will be responsible for the failure of the quantisation of these relational observables to produce quantum Dirac observables.
\begin{Lemma}\label{lem_povm}
The $n^{\rm th}$-moment operators of the covariant clock POVM are not conjugate to the clock Hamiltonian for $n>0$
\ba\label{eq:nthmomHCcomm}
[\hat\phi_C^{(n)},\hat H_C] = i\,n\,\hat\phi_C^{(n-1)}-i\,(t_{\rm max})^{n-1}\,\ket{0}\!\bra{0}\,.
\ea
For $n=0$, we clearly have $[\hat\phi_C^{(0)},\hat H_C]=0$. 
\end{Lemma}
\begin{proof}
\proofapp
\end{proof}
The $n^{\rm th}$-moment operators are therefore only conjugate to the clock Hamiltonian on the subspace
\begin{align}
\mathcal{D} \ce \{ \ket{\psi} \in \mathcal{H}_C \ \big| \ \braket{\phi=0 | \psi} =0 \} \subset \mathcal{H}_C, \label{D}
\end{align}
which is dense in the clock Hilbert space $\mathcal{H}_C$ when its dimensionality is infinite~\cite{garrisonCanonicallyConjugatePairs1970,lahtiCovariantPhaseObservables1999}.~This is the quantum analogue of the situation encountered in the classical theory where the angle variable $\phi_C$, and consequently the dynamically defined monotonic clock $T(s)$, were canonically conjugate to $H_C$ on a dense subset of $\cp_C$ (cf.~Eq.~\eqref{classcov} and the surrounding discussion). Note that, despite this correspondence, the covariant POVM should not be understood as a direct quantisation of $\phi_C$, but rather is defined by the covariance condition. The quantisation of phase variables is ambiguous~\cite{susskind1964quantum}, while here the form of the phase states in Eq.~\eqref{discreteClockState} follows uniquely from the covariance condition in Eq.~\eqref{statecov}, and the demand that they form a POVM via Eq.~\eqref{eq:effectdensities}. Given that the effect operators are bounded and self-adjoint, the resulting POVM is a well-defined measure on the entire clock Hilbert space (including the complement of~$\mathcal{D}$), whose first moment corresponds to $\phi_C$.

The overlap of two clock states is
\begin{align}
\braket{\phi|\phi'} = \sum_{\varepsilon_j \in \spec(\hat H_C)}   e^{i\varepsilon_j (\phi-\phi')}, \label{id3}
\end{align}
such that the clock states are only orthogonal -- i.e.\ perfectly distinguishable -- if the $n_j$ in Eq.~\eqref{eClockEnergies} are such that $\lbrace n_j \rbrace= \mathbb{Z}$~\cite{braunsteinGeneralizedUncertaintyRelations1996}; this corresponds to the case of an \emph{ideal periodic clock}.~Furthermore, we note that the clock states are generically \emph{not} eigenstates of the first moment $\hat\phi_C^{(1)}$ \cite{Hoehn:2019owq},
\begin{equation}\label{eq:hatphiketphi}
\hat\phi_C^{(1)} \ket{\phi}= \f{t_{\rm max}}{2} \ket{\phi} +   i\sum_{\substack{\varepsilon_j,\varepsilon_k \in \spec(\hat H_C) \\ j\neq k}}  \frac{e^{ig(\varepsilon_j) }}{\varepsilon_j- \varepsilon_k} e^{-i\varepsilon_k\phi}\ket{\varepsilon_j}.
\end{equation}
Indeed, there is also no reason to expect them to be eigenstates.~For example, we could have $\hat H_C=\hat L$, where $\hat L$ is an angular momentum operator.~In this case, the spectrum of $\hat H_C$ can be upper and lower bounded and $\ch_C$ finite-dimensional.~In that case $\hat\phi_C^{(1)}$ will be hermitian and feature a finite spectrum, while the clock states $\ket{\phi}$ constitute a continuous one-parameter family of states in $\ch_C$.~It is only true for ideal periodic clocks, defined above, that the clock states are eigenstates of the moment operators.~Indeed, this can be easily verified using the r.h.s.\ of Eq.~\eqref{nthmoment} and the orthogonality of the ideal periodic clock states.

We can use the covariance property given in Eq.~\eqref{statecov} together with the periodicity displayed in Eq.~\eqref{id} to determine the evolution of the clock operator $\hat{\phi}_C$, understood as the first moment of the clock POVM, with respect to $\hat{H}_C$.~In line with the fact that $\hat{\phi}_C$ is conjugate to $\hat{H}_C$ only on a dense subspace of $\mathcal{H}_C$ (cf.~Lemma \ref{lem_povm}), the following Lemma establishes that this evolution is  analogous to the classical case given by Eq.~\eqref{solution}.

\begin{Lemma}\label{lem_quantumPhiEvol}
Let $\ket{\psi_{1,2}}$ be states in the clock Hilbert space such that $\braket{\phi|\psi_{1,2}} = \psi_{1,2}(\phi)$ are integrable functions, and let $\hat{\phi}_C(s):=U^{\dagger}_C(s)\hat{\phi}_CU_C(s)$ with $\braket{\hat{\phi}_C(s)}_{21}:=\braket{\psi_2|\hat{\phi}_C(s)|\psi_1}$. Then, the clock operator $\hat{\phi}_C$ obeys the following evolution law:
\begin{equation}\begin{aligned}
&\braket{\hat{\phi}_C(s)}_{21}\notag\\
&= \frac{1}{t_{\rm max}}\int_0^{t_{\rm max}}d\phi \ \left(s+\phi-t_{\rm max}\left\lfloor\frac{s+\phi}{t_{\rm max}}\right\rfloor\right)\psi_2^*(\phi)\psi_1(\phi)\,,
\end{aligned}
\end{equation}
which we take to be the quantum analogue of the classical evolution given by Eq.~\eqref{solution}, with $t_{\rm max}$ being the counterpart to the classical period $\phi_{\rm max}$.
\end{Lemma}
\begin{proof}
\proofapp
\end{proof}

Before we dig into the details of the Dirac quantised theory, let us emphasise that the formalism of periodic, covariant POVMs encompasses both infinite-dimensional and finite-dimensional quantum clocks.~We close this section by illustrating the above results for the smallest non-trivial clock, that is a qubit clock.~Examples of clocks with infinitely many energy levels will be provided in the coming sections.

\begin{example}[\textbf{Two-level clock}]\label{Ex:2levelclock1}
    Consider the clock $C$ to be a two-level system like a spin-1/2 particle with Hamiltonian $\hat{H}_C=\frac{\omega}{2}\hat{\sigma}_C^x$.~The energy eigenstates are $\ket{\varepsilon_{\pm}}=\ket{\pm}$, with $\ket\pm$ the $\pm1$-eigenstates of the Pauli X operator $\hat{\sigma}_C^x$, and their eigenvalues are $\varepsilon_{\pm}=\pm\frac{\omega}{2}$.~The evolution operator $U_C(t)=\exp(-it\hat{H}_C)$ is $\frac{2\pi}{\omega}$-periodic and the clock states read 
\begin{equation}\label{eq:qubitclockstates}
        \ket{\phi}=e^{ig(\frac{\omega}{2})}e^{-i\frac{\omega}{2}\phi}\ket{+}+e^{ig(-\frac{\omega}{2})}e^{i\frac{\omega}{2}\phi}\ket{-}\,,
    \end{equation}
    for $\phi\in[0,\frac{2\pi}{\omega})$ and $g(\varepsilon_{\pm})=g(\pm\frac{\omega}{2})$ an arbitrary real function.\footnote{For $g(\pm\omega)=2\pi z$, $z\in\mathbb Z$, the clock states \eqref{eq:qubitclockstates} reproduce (up to a $1/\sqrt{2}$ normalisation) those recently used in \cite{Cafasso:2024zqa} for implementing a finite-dimensional analogue of the quantum time-dilation mechanism in the presence of gravitation-like interactions.}~These are nothing but a system of generalised coherent states $\{\ket{\phi}=U_C(\phi)\ket{\phi=0}\}_{\phi\in[0,2\pi/\omega)}$ for $\rm{U}(1)$, and are not all orthogonal/perfectly distinguishable.~That is, the clock is not an ideal reference frame.~It is straightforward to check that the properties \eqref{statecov},~\eqref{resolid},~and \eqref{id} are satisfied with $t_{\rm max}=\frac{2\pi}{\omega}$.~Thus, the clock states \eqref{eq:qubitclockstates} allow to construct a $\rm{U}(1)$-covariant and $\frac{2\pi}{\omega}$-periodic POVM as in Eqs.~\eqref{eq:effectdensities}-\eqref{eq:effectoperators}.~A little algebra shows that the associated $n$-th moments satisfy the property \eqref{eq:nthmomHCcomm} for $n>0$.~In particular, the qubit clock operator $\hat{\phi}_C$, identified with the first moment, reads
    \begin{equation}\label{eq:qubitclockoperator}
    \begin{aligned}
        \hat{\phi}_C=\frac{\pi}{\omega}I_C+\frac{i}{\omega}\Bigl(&e^{i(g(\varepsilon_+)-g(\varepsilon_-)) }\ket{+}\!\bra{-}\\
        &-e^{-i(g(\varepsilon_+)-g(\varepsilon_-))}\ket{-}\!\bra{+}\Bigr)\;,
    \end{aligned}
\end{equation}
and we have
\begin{equation}\label{eq:phiCHCcomm}
\begin{aligned}
    [\hat{\phi}_C,\hat{H}_C]&=i(I_C-\ket{\phi=0}\!\bra{\phi=0})\\
    &=i(I_C-\ket{\phi=2\pi z/\omega}\!\bra{\phi=2\pi z/\omega})\,,\;z\in\mathbb Z\,.
\end{aligned}
\end{equation}
The clock operator and the clock Hamiltonian then form a conjugate pair only on the one-dimensional subspace
\begin{equation}\label{eq:2levelDsubspace}
\begin{aligned}
\mathcal D&=\emph{span}\{e^{ig(\varepsilon_+)}\ket{+}-e^{ig(\varepsilon_-)}\ket{-}\}\\
&=\emph{span}\{\ket{\phi=\pi/\omega}\}\;,
\end{aligned}
\end{equation}
orthogonal to $\ket{\phi=2\pi z/\omega}$ (cf.~Eq.~\eqref{D}).~Moreover, using Eq.~\eqref{eq:qubitclockoperator}, it is easy to check that the clock states are not eigenstates of the clock operator as in Eq.~\eqref{eq:hatphiketphi}.
\end{example}

\section{Dirac quantisation and the physical Hilbert space}\label{sec_Dirac}

Our aim is now to quantise the kinematical phase space $\cp_{\rm kin}=\cp_C\times\cp_S$ and subsequently to impose the constraint Eq.~\eqref{noint} in the quantum theory.~To this end, we replace $\cp_{\rm kin}$ with a kinematical Hilbert space $\ch_{\rm kin}=\ch_C\otimes\ch_S$ and assume both clock and system Hamiltonian are promoted to self-adjoint operators $\hat H_C$ and $\hat H_S$ on $\ch_{\rm kin}$. The resulting formalism can, however, be equally well applied when no classical analogue of the quantum theory exists, such as for intrinsic angular momentum (spin).

\subsection{General procedure}\label{ssec_diracgen}

Since $\hat H_C$ has discrete spectrum (\textit{cf}.\ Sec~\ref{sec_covpovm}), the expansion of an arbitrary kinematical state  in the $\hat H_C$ and $\hat{H}_S$ eigenbases $\ket{\varepsilon_i}_C$, $\ket{E,\sigma_E}_S$, where $\sigma_E$ is a possible degeneracy label, takes the form
\ba
\ket{\psi_{\rm kin}} &=& \sum_{\varepsilon_i\in\spec(\hat H_C)}\,\,\,\,{{\intsum}_{E\in\spec(\hat H_S)}}\,\sum_{\sigma_E}\,\psi_{\rm kin}(\varepsilon_i,E,\sigma_E)\,\nn\\
&&\q\q\q\q\q\q\q\q\times\ket{\varepsilon_i}_C\otimes\ket{E,\sigma_E}_S\,.\label{kinstate}
\ea

In the Dirac constraint quantisation procedure, physical states are required to satisfy the quantisation of the constraint Eq.~\eqref{noint} in the form
\begin{align}
\hat{C}_H \ket{\psi_{\rm phys}} =  \big( \hat{H}_C \otimes I_S + I_C \otimes \hat{H}_S\big) \ket{\psi_{\rm phys}} = 0.
\label{WDW}
\end{align}
Assuming this equation has a non-trivial solution, and denoting
\ba
\sigma_{S\vert C}:=\spec(\hat H_S)\cap\spec(-\hat H_C)\,,\label{spec}
\ea
we have to distinguish between the following two cases: 
\begin{enumerate}[label=(\alph*)]
\item Zero lies in the point (discrete) spectrum of $\hat C_H$. This requires the part of $\spec(\hat H_S)$ containing $\sigma_{S\vert C}$ to be discrete;
\item Zero lies in the continuous spectrum of $\hat C_H$.~This requires the part of $\spec(\hat H_S)$ containing $\sigma_{S\vert C}$ to be continuous.
\end{enumerate}
Cases (a) and (b) will hinge on the relation between the groups and representations generated by the clock Hamiltonian $\hat H_C$ and the constraint $\hat C_H$.~Generally, these will differ and this will become crucial in the quantisation of the relational observables.~We will provide examples for both cases below in Sec.~\ref{ssec_ex2}.

For both cases we can construct a `projector' map from $\ch_{\rm kin}$ to solutions to Eq.~\eqref{WDW}\footnote{This improper projector implements a coherent averaging over the group $G$ generated by the constraint \cite{marolfRAQ1995,Giulini:1998kf,Giulini:1998rk,Hartle:1997dc,Marolf:2000iq,thiemannModernCanonicalQuantum2008}.~For example, when the spectrum of $\hat H_S$ is purely continuous and $\hat C_H$ thus generates a non-compact one-parameter group, this improper projector coincides with
\ba
\Pi_{\rm phys}=\delta(\hat C_H)\ce\f{1}{2\pi}\int_\mathbb{R}ds\,e^{is\hat C_H},\nn
\ea
as can be easily checked by using the spectral decomposition
\ba
\hat C_H=\!\!\!\!\sum_{\varepsilon_i\in\spec(\hat H_C)}\int_{\spec(\hat{H}_{S})} dE\sum_{\sigma_E}(\varepsilon_i+E)\ket{\varepsilon_i}_C\!\bra{\varepsilon_i}\otimes\ket{E,\sigma_E}_S\!\bra{E,\sigma_E}.\nn
\ea
More generally, $\Pi_{\rm phys}=\delta_G(\hat C_H)$, where $\delta_G(\cdot)$ is the delta function over the group $G$. }
\ba
\Pi_{\rm phys}=\sum_{E\in\sigma_{S\vert C}}\sum_{\sigma_{E}}\ket{-E}_C\!\bra{-E}\otimes\ket{E,\sigma_{E}}_S\!\bra{E,\sigma_{E}}.\nn\\\label{impproj}
\ea
Indeed, kinematical states as in Eq.~\eqref{kinstate} map in both cases (a) and (b) to states
\ba
\ket{\psi_{\rm phys}} &=& \Pi_{\rm phys}\,\ket{\psi_{\rm kin}}\nn\\
&=&\!\!\!\!\sum_{E\in\sigma_{S\vert C}}\sum_{\sigma_{E}}\,\psi_{\rm kin}(-E,E,\sigma_{E})\,\ket{-E}_C\ket{E,\sigma_{E}}_S\,,\nn\\\label{crap}
\ea
which formally solve the constraint Eq.~\eqref{WDW}. 

In fact, the situation is somewhat more subtle:~in case (a), $\Pi_{\rm phys}^2=\Pi_{\rm phys}$ and so it is a proper orthogonal projector on the subspace of $\ch_{\rm kin}$ corresponding to solutions to Eq.~\eqref{WDW}.~In particular, physical states are normalizable in $\ch_{\rm kin}$.~By contrast, in case (b), $\Pi_{\rm phys}$ is an improper projector because the proportionality factor in $\Pi^2_{\rm phys}\sim\Pi_{\rm phys}$ will diverge.~Indeed, physical states are improper eigenstates of $\hat C_H$ and thereby not normalizable in $\ch_{\rm kin}$.~As such, they are not elements of $\ch_{\rm kin}$ (but rather distributions on kinematical states \cite{marolfRAQ1995,Giulini:1998kf,Giulini:1998rk,Hartle:1997dc,Marolf:2000iq, thiemannModernCanonicalQuantum2008}).~In this case, we need a new inner product to normalise physical states.

To accommodate both cases, we can define a physical (gauge-invariant) inner product on the space of solutions to (\ref{WDW}), which in case (a) coincides with the restriction of the kinematical inner product $\la\cdot\ket{\cdot}_{\rm kin}$ to the space of solutions:

\ba \label{PIP}
\langle\psi_{\rm phys}|\phi_{\rm phys}\rangle_{\rm phys}&:=&\la\psi_{\rm kin}|\Pi_{\rm phys}\ket{\phi_{\rm kin}}_{\rm kin}\\
&=&  \sum_{E\in\sigma_{S\vert C}}\sum_{\sigma_{E}}
\psi^*_{\rm kin}(-E,E,\sigma_{E})\nn\\
&&\q\q\q\q\q\q\times\phi_{\rm kin}(-E,E,\sigma_{E})\,.\nn
\ea
On the right hand side, the state $\ket{\phi_{\rm kin}}$ ($\ket{\psi_{\rm kin}}$) is any member of the equivalence class of kinematical states that project under Eq.~\eqref{impproj} to the same physical state $\ket{\phi_{\rm phys}}$ ($\ket{\psi_{\rm phys}}$).~Since $\Pi_{\rm phys}$ is symmetric, this expression depends only on the equivalence classes of kinematical states, but not on their representatives.~As such it can be used to define an inner product on the space of solutions to Eq.~\eqref{WDW} also in case (b) (which may require modding out spurious zero-norm solutions).~Cauchy completion then yields the so-called physical Hilbert space $\ch_{\rm phys}$  \cite{marolfRAQ1995,Giulini:1998kf,Giulini:1998rk,Hartle:1997dc,Marolf:2000iq, thiemannModernCanonicalQuantum2008} (which, in case (a) is a proper subspace of $\ch_{\rm kin}$).

From Eq.~\eqref{WDW}, one can see that the dimensionality of $\ch_{\rm phys}$ is determined by the number of solutions to the equation $\varepsilon_{j}+E=0$, with $\varepsilon_j\in\spec(\hat H_C)$ and $E\in\spec(\hat H_S)$, and the degree of degeneracy of $E$ in each solution.~Considering two different solutions, denoted $(\varepsilon_j ,E)$ and $(\varepsilon_k ,E')$ and using~Eq.~\eqref{eClockEnergies}, one finds
\begin{equation} \label{eEnergyDiffCond}
    n_j - n_k = \frac{E' - E}{\omega_{t}}\q,\q n_j,n_k\in\mathbb Z\,.
\end{equation}
Noting that the left-hand side is an integer, we see that if there do not exist two eigenvalues of $\hat{H}_{S}$ which differ by an integer multiple of the clock frequency $\omega_{t}$, then $\ch_\mathrm{phys}$ can be at most one-dimensional (\textit{i.e}.\ $\ch_{\rm phys}\simeq\mathbb{C}$), and the relational theory is thus trivial.~This will be illustrated in an example~\ref{ex_4} below.

It will be convenient to introduce the quantum weak equality $\approx$ of operators, i.e.\ equality of operators $\hat O_1,\hat O_2$ on the `quantum constraint surface' $\ch_{\rm phys}$:
\ba
&&\hat O_1\approx\hat O_2\q\Leftrightarrow\q\nn\\
&&\q\q\q\q\q \hat O_1\,\ket{\psi_{\rm phys}}=\hat O_2\,\ket{\psi_{\rm phys}},\q\forall\,\ket{\psi_{\rm phys}}\in\ch_{\rm phys}.\nn
\ea

\subsection{Examples}\label{ssec_ex2}

It is useful to illustrate the difference between cases (a) and (b) and their consequences for the relation between the groups and representations generated by clock Hamiltonian and constraint. 

\begin{example}[\bf{Commensurate oscillators}]\label{ex_3}
We quantise Example~\ref{ex_1} from Sec.~\ref{sss_relobsperclock} in the case where $\omega_t/\omega=q\in\mathbb{Q}\,\leq 1$, which as discussed there corresponds to the case in which the relational observables describing the position and momentum of $S$ relative to the clock provide us with canonically conjugate global Dirac observables.~This is an example of case (a) in Sec.~\ref{ssec_diracgen}.~We will consider the case of incommensurate frequencies separately in the Example~\ref{ex_4} below.

The constraint is given by 
\ba
\hat C_H&=&\left(\f{\hat p_t^2}{2m_t}+\f{m_t\omega_t^2}{2}\hat t^2\right)\otimes I_S+I_C\otimes\left(\f{\hat p^2}{2m}+\f{m\omega^2}{2}\hat q^2\right)\nn\\
&&\q\q\q\q\q\q\q\q\q\q-E\,I_C\otimes I_S.\label{WDW2}
\ea
Let $m_1,m_2\in\mathbb{N}$ be the smallest numbers such that $q=m_1/m_2$.~As can be readily seen by using the energy eigenbasis for the two oscillators, say $\{\ket{n_1}_C\}_{n_1\in\mathbb N}$ and $\{\ket{n_2}_S\}_{n_2\in\mathbb N}$, solving this constraint in the sense of Eq.~\eqref{WDW} requires $\omega_t$,~$\omega$,~and $E$ to be arranged such that 
\ba
\tilde E=m_1\,n_1+m_2\,n_2,\nn
\ea
where 
\ba
\tilde E:=\f{m_1\,E}{\omega_t}-\f{m_1+m_2}{2}\nn
\ea
 has solutions for some $n_1,n_2\in\mathbb{N}$.~A necessary but not sufficient condition is $\tilde E\in\mathbb{N}$.~Determining the numbers of solutions is in general a complicated problem.~However, for sufficiently large $\tilde E$ there will exist many solutions.~For example, in the special case that $m_1=m_2$ there are $\tilde E+1$ solutions.

Physical states are of the form
\ba
\ket{\psi_{\rm phys}}=\sum_{\stackrel{n_1,n_2:}{m_1n_1+m_2n_2=\tilde E}}\,\alpha_{n_1,n_2}\ket{n_1}_C\otimes\ket{n_2}_S\nn
\ea
for some $\alpha_{n_1,n_2}\in\mathbb{C}$ and the physical inner product is 
\ba
\braket{\psi'_{\rm phys}|\psi_{\rm phys}}_{\rm phys}=\sum_{\stackrel{n_1,n_2:}{m_1n_1+m_2n_2=\tilde E}}\,\alpha'^*_{n_1,n_2}\,\alpha_{n_1,n_2},\nn
\ea
which is simply the kinematical inner product on $\ch_{\rm kin}=\ch_C\otimes\ch_S$ restricted to the subspace of solutions to Eq.~\eqref{WDW2}.~This defines the physical subspace $\ch_{\rm phys}\subset\ch_{\rm kin}$.

Let us now consider the group actions.~$\hat H_C$ and $\hat H_S$ generate a $t_{\rm max}$- and $t_{{\rm max},S}$-periodic projective unitary representation of $\rm{U}(1)$.~The ratio between their periods is $t_{\rm max}/t_{{\rm max},S}=\omega/\omega_t=1/q=m_2/m_1$.~In other words, for every $m_1$ cycles of the clock $C$, the system $S$ undergoes $m_2$ cycles.~Consequently, the constraint in Eq.~\eqref{WDW2} generates a $(m_1 t_{\rm max})$-periodic projective unitary representation of $\rm{U}(1)$ on $\ch_{\rm kin}$.~In particular, in the special case that $\omega_t\equiv\omega$, the group representations generated by $\hat H_C,\hat H_S$ and $\hat C_H$ match, but in general they do not.
\end{example}

\begin{example}[\bf{Incommensurate oscillators}]\label{ex_4}
We consider again the constraint in Eq.~\eqref{WDW2} as in Example~\ref{ex_3}, however, now with the important twist that we assume the two oscillators to have incommensurate frequencies, $\omega_t/\omega\notin\mathbb{Q}$, which is again an example of case (a) in Sec.~\ref{ssec_diracgen}.~However, noting that Eq.~\eqref{eEnergyDiffCond} is unchanged by the inclusion of a constant term in the constraint, and cannot be satisfied for oscillators of incommensurate frequencies, we immediately see that this twist has a dramatic effect: the space of solutions $\ch_{\rm phys}$ is at most one-dimensional.

Hence, when the frequencies are incommensurate and the constraint can be solved, there is up to phase and normalisation a \emph{unique} physical state
\ba
\ket{\psi_{\rm phys}}=\ket{n_1}_C\otimes\Big|\f{E}{\omega}-\f{1}{2}-\left(n_1+\f{1}{2}\right)\f{\omega_t}{\omega}\Big\rangle_S\nn
\ea
and the physical Hilbert space is thus $\ch_{\rm phys}\simeq\mathbb{C}$.~As we shall discuss later in Sec.~\ref{sec_observables-periodic}, it is clear that this has drastic consequences also for the existence of quantum Dirac observables (cf.~Example~\ref{ex_8}).

The above observation becomes more transparent when looking at the group actions generated by the Hamiltonians and the constraint.~Firstly, the clock and system Hamiltonian $\hat H_C$ and $\hat H_S$ generate a $t_{\rm max}$- and a $t_{{\rm max},S}$-periodic projective unitary representation of $\rm{U}(1)$ on $\ch_{\rm kin}=\ch_C\otimes\ch_S$, respectively.~However, since $t_{\rm max}/t_{{\rm max},S}=\omega/\omega_t\notin\mathbb{Q}$ it is clear that if $C$ and $S$ start a cycle simultaneously, they will never complete a cycle for finite `time' simultaneously again; the representations of $\rm{U}(1)$ which they generate are incommensurate.~From this we can already guess that the constraint $\hat C_H$, in fact, does \emph{not} generate a $\rm{U}(1)$ representation.

Indeed, $\hat C_H$ has a discrete spectrum $\omega_t(n_1+\f{1}{2})+\omega(n_2+\f{1}{2})-E$, $n_1,n_2\in\mathbb{N}$ on $\ch_{\rm kin}$.~Let us denote the eigenvalues by $C_i$.~We have 
\ba
\Delta C_{ij}:=C_i-C_j=\omega_t(n_{1,i}-n_{1,j})+\omega(n_{2,i}-n_{2,j}),\nn
\ea
and since $\omega_t/\omega\notin\mathbb{Q}$, the spectrum is non-degenerate.~Now define the unitaries $U_{CS}(s):=\exp(-i\,s\,\hat C_H)$.~It is clear that, in order for $U_{CS}(s)=e^{i\,\beta}I_C\otimes I_S$ for some $\beta\in[0,2\pi)$ and $s\in\mathbb{R}$, we need to have for all eigenvalues $\Delta C_{ij}\,s=2\pi\,z_{ij}$, where $z_{ij}$ is an $i,j$-dependent integer.~But this requires $\Delta C_{ij}/\Delta C_{lk}\in\mathbb{Q}$ for all eigenvalues $i,j,k,l$.~However, given that $\omega_t/\omega\notin\mathbb{Q}$, this is impossible.~Hence, $U_{CS}(s)$ does not define a projective unitary representation of $\rm{U}(1)$ on $\ch_{\rm kin}$.

In fact, given that the spectrum of the constraint is discrete and non-degenerate, we can expand an arbitrary kinematical state in the constraint eigenstates $\ket{C_i}$ as $\ket{\psi_{\rm kin}}=\sum_i\,c_i\,\ket{C_i}$.~On account of what we have just observed, there does not exist a finite $s$ such that $U_{CS}(s)\,\ket{\psi_{\rm kin}}=e^{i\beta}\ket{\psi_{\rm kin}}$ and accordingly the group generated by $\hat C_H$ is noncompact.~The constraint which is the sum of two $\rm{U}(1)$-generators yields, instead, a unitary representation of the translation group $\mathbb{R}$ on $\ch_{\rm kin}$. 

Nevertheless, the recurrence theorems~\cite{PhysRev.107.337,percival,schulman} tell us that in aperiodic intervals in $s$, $U_{CS}(s)\,\ket{\psi_{\rm kin}}$ may come arbitrary close to $\ket{\psi_{\rm kin}}$ in the sense that the difference vector becomes arbitrarily close to the zero-vector.
\end{example}

\begin{example}[\bf{Oscillator clock and free particle}]\label{ex_5}
As an illustration of case (b) in Sec.~\ref{ssec_diracgen}, we quantise Example~\ref{ex_2} of the harmonic oscillator clock and the free particle.~The quantum constraint is given by Eq.~\eqref{WDW} where the Hamiltonians are given by 
\ba
\hat H_C=\f{\hat p_t^2}{2m_t}+\f{m_t\omega_t^2}{2}\hat t^2\q,\q\hat H_S=-\f{\hat p^2}{2m}\;.\nn
\ea
Since the system energies are now doubly degenerate, the basis of the physical Hilbert space is isomorphic to two copies of a harmonic oscillator eigenbasis
\ba
\ket{n}_C\otimes\Big| p=\pm\sqrt{2m\,\omega_t(n+\f{1}{2})}\Big\rangle_S
\ea
 correspondonding to left and right moving modes of the free particle.~Indeed, since the expression in the square root is never zero, the basis elements of the left and right moving modes are orthogonal and an arbitrary physical state can be written as
\ba
\ket{\psi_{\rm phys}} = \sum_{\sigma=+,-}\,\sum_{n=0}^\infty\,\psi_n^\sigma\, \ket{n}_C\otimes\Big| \sigma\sqrt{2m\,\omega_t(n+\f{1}{2})}\Big\rangle_S.\nn
\ea
The physical inner product reads
\ba
\braket{\phi_{\rm phys}|\psi_{\rm phys}}_{\rm phys} = \sum_{\sigma=+,-}\sum_{n=0}^\infty\,\left(\phi^\sigma_n\right)^*\,\psi_n^\sigma.\nn
\ea
Altogether, we thus have a decomposition of the physical Hilbert space into left and right mover sectors ${\ch_{\rm phys}\simeq L^2(\mathbb{R})_+\oplus L^2(\mathbb{R})_-}$.

It is clear that $\hat H_C$ generates a $t_{\rm max}$-periodic projective unitary representation of $\rm{U}(1)$, while the constraint $\hat C_H$, which has continuous spectrum, generates a unitary representation of the translation group $\mathbb{R}$ on the kinematical Hilbert space $\ch_{\rm kin}$.~We thus again have a mismatch between the groups generated by clock Hamiltonian and constraint.
\end{example}

\begin{example}[\bf{Two-level clock and free particle}]\label{Ex:2levelclock2}
    Lastly, as an example with a finite-dimensional clock, we consider the two-level clock of Example~\ref{Ex:2levelclock1} and a free particle.~The clock and system Hamiltonians entering the constraint \eqref{WDW} are then
    \begin{equation}\label{eq:qubitHSandfreeHS}
    \hat{H}_C=\frac{\omega}{2}\hat{\sigma}_C^x\qquad,\qquad\hat{H_S}=-\frac{\hat{p}^2}{2m}\,.
    \end{equation}
    Similarly to Example~\ref{ex_5}, physical states can be written as 
\begin{equation}\label{eq:qubitpsiphys}
    \ket{\psi_{\rm phys}}=\sum_{\sigma=+,-}\psi_{+}^{\sigma}\ket{\varepsilon_+}_C\otimes\ket{\sigma\sqrt{m\omega}}_S\,,
    \end{equation}
    with $\ket{\varepsilon_{\pm}}=\ket{\pm}$ the eigenstates of $\hat{\sigma}_C^x$ and inner product
    $$
    \langle\phi_{\rm phys}|\psi_{\rm phys}\rangle_{\rm phys}=\sum_{\sigma=+,-}(\phi_{+}^{\sigma})^*\psi_{+}^{\sigma}\,.
    $$
    That is, $\mathcal{H}_{\rm phys}\simeq\mathbb{C}^2$ corresponds to a qubit itself.

    As discussed in Example~\ref{Ex:2levelclock1}, $\hat H_C$ generates a $2\pi/\omega$-periodic unitary representation of $\rm{U}(1)$.~The constraint $\hat C_H$, which has continuous spectrum, generates instead a unitary representation of the translation group $\mathbb{R}$ on $\ch_{\rm kin}=\ch_{\rm C}\otimes\ch_{\rm S}$.~Again, the groups generated by $\hat{H}_C$ and $\hat{C}_H$ do not match.
\end{example}

\section{Quantum relational observables relative to periodic clocks}\label{sec_observables-periodic}

Let us now move our discussion to the relational observables in the quantum theory.~As we will see in this section, similar challenges to those discussed in Sec.~\ref{sss_relobsperclock} for the classical theory arise also for relational quantum dynamics with periodic clocks.~Similarly to the classical result of Lemma~\ref{lem_clobs}, relational observables relative to periodic clocks turn out not to constitute Dirac observables in the quantum theory, unless the system observable is periodic too.

This failure to yield gauge-invariant observables is rooted in a mismatch between the $\rm{U}(1)$-group generated by the clock Hamiltonian $\hat H_C$ and the gauge group $G$ generated by the constraint $\hat C_H$, which we have already seen in the examples of the previous section.~If such a mismatch arises, it comes from the fact that the group generated by the constraint on $\ch_{\rm kin}$ does not act freely on the clock subsystem and there is a non-trivial isotropy subgroup.\footnote{In the terminology of \cite{delaHamette:2021oex}, the clock states \eqref{discreteClockState} identify a incomplete reference frame.} In other words, we have an action of $X=G/H$, rather than $G$ on the clock subsystem, where $H$ is the clock states' isotropy group.~According to Eq.~\eqref{id}, $H$ is isomorphic to some subset $\mathbb{Z}_G\subset\mathbb{Z}$, as all transformations by integer multiples of the clock period result in the same clock state up to phase, but the pertinent integer is not permitted to map out of the parameter range of $G$.~For example, when the constraint generates the translation group (cf.~Example~\ref{ex_5}), then the clock subsystem carries a (projective unitary) representation of the coset $\rm{U}(1)\simeq \mathbb{R}/\mathbb{Z}$, i.e.\ the translation group modded out by its normal subgroup the integers.~Similarly, when $G=\rm{U}(1)$ itself, then the clock's isotropy group can be any of its cyclic finite subgroups,~i.e.\ $H=\mathbb{Z}_n$, depending on how many clock cycles fit into one cycle of the gauge flow.

It has been shown for general groups $G$ with \emph{compact} isotropy group $H$ of the quantum reference frame coherent state system that this quantum frame can only gauge-invariantly resolve properties of its complement (the system $S$) that are $H$-invariant themselves \cite{delaHamette:2021oex}. The periodic clock discussion of relational observables below can be viewed as an explicit illustration of this observation for the compact case $H=\mathbb{Z}_n$ and an extension in the noncompact case $H=\mathbb{Z}$. Indeed, the periodic clock can only resolve a periodic time evolution of $S$.

\subsection{Quantisation of classical relational observables: partial and complete $G$-twirl} \label{sec_quantclassobs}

We begin with a direct quantisation of the classical relational Dirac observables.~The power series Eq.~\eqref{relobs} shows that to this end we do not need to define winding numbers in the quantum theory (the winding numbers are hidden in the evolution parameter $\tau$).~Since it is the (non-monotonic) phase observable $\phi_C$, rather than the monotonic clock function $T$ that appears on the r.h.s.\ of Eq.~\eqref{relobs}, we only have to replace the $n^{\rm th}$ power of $\phi_C$ by the $n^{\rm th}$-moment $\hat \phi_C^{(n)}$ of the covariant and periodic clock POVM of Sec.~\ref{sec_covpovm}. 
% There is, however, one subtlety: recalling from the end of Sec.~\ref{sec_covpovm} that $\phi\in[0,t_{\rm max})$, we will need to rescale the classical evolution parameter $\tau$.~Since it runs through the values of the unravelled clock function $T(s)=s+\phi_C$ which is proportional to the angle variable $\phi_C\in[0,2\pi)$, we should rescale it as follows: $\tau\rightarrow \tilde\tau=\f{t_{\rm max}}{2\pi}\tau$.~For notational simplicity, however, we will continue to use the symbol $\tau$ in all of below, but emphasise that whenever we write $\tau$ in the quantum theory, we really mean $\tilde\tau$, i.e.\ a parameter that differs by the factor $t_{\rm max}/2\pi$ from its classical counterpart. 
The quantisation of the relational observables in Eq.~\eqref{relobs} thus reads
\begin{align}
\!\!\hat{F}_{f_S,T}(\tau)  &\ce  \frac{1}{t_{\rm max}}\int_0^{t_{\rm max}}d\phi\, \ket{\phi}\!\bra{\phi}\nn\\
&\q\q\q\q\otimes\sum_{n=0}^{\infty}\, \frac{i^n}{n!} \left( \phi-\tau\right)^n  \bigl[ \hat{f}_S , \hat{H}_S \bigr]_n\nn \\
&= \frac{1}{t_{\rm max}} \int_0^{t_{\rm max}} d\phi\, \ket{\phi}\!\bra{\phi} \otimes e^{i(\tau -\phi) \hat{H}_S}\hat{f}_S\, e^{-i(\tau -\phi) \hat{H}_S} \nn \\
&= U_{CS}^\dag(\tau)\,\left[\f{1}{t_{\rm max}}\int_{0}^{t_{\rm max}} d\phi\, U_{CS}(\phi) \left( \ket{\tau}\!\bra{\tau} \otimes \hat{f}_S \right)  \right.\nn\\
&\left.\q\q\q\q\q\q\q\q\q\q\times U_{CS}^\dagger(\phi)\right] U_{CS}(\tau)\nn\\
&=:U_{CS}^\dag(\tau)\, \mathcal{G}_{[0,t_{\rm max})}\left( \ket{\tau}\!\bra{\tau} \otimes \hat{f}_S\right) U_{CS}(\tau)\,,
\label{qrelobs}
\end{align}
where $U_{CS}(\tau):=\exp(-i\,\tau\,\hat C_H)$ furnish the unitary representation of the group generated by the constraint on the kinematical Hilbert space and labeled by $\tau$ which parametrises the temporal manifold $\cm$.~In particular, if $\cm\simeq\mathbb{R}$, then $\tau\in\mathbb{R}$ in contrast to $\phi\in[0,t_{\rm max})$.~In the third line, we made use of Eq.~\eqref{statecov}.
$\mathcal{G}_{[0,t_{\rm max})}(\cdot)$ denotes the $G$-twirl over an interval $[0,t_{\rm max})$ of the group generated by the quantum constraint $\hat C_H$.~In other words, it is the $\rm{U}(1)$-twirl corresponding to the group generated by the clock Hamiltonian, which as already emphasised need not coincide with the group generated by the constraint (for concrete examples, see Sec.~\ref{ssec_ex2}).~As such, it will generally be a \emph{partial $G$-twirl} over the group generated by the constraint.

The partial $G$-twirl is independent of the clock cycle. Indeed, using Eq.~\eqref{id} and Eq.~\eqref{shift}, one finds
\ba
&&\mathcal{G}_{[0,t_{\rm max})}\left( \ket{\tau}\!\bra{\tau} \otimes \hat{f}_S\right) \nn\\
&&\q\q\q= \mathcal{G}_{[0,t_{\rm max})}\left( \ket{\tau_C+t_{\rm max}\,n}\!\bra{\tau_C+t_{\rm max}\,n} \otimes \hat{f}_S\right)\nn\\
&&\q\q\q= \mathcal{G}_{[0,t_{\rm max})}\left( \ket{\tau_C}\!\bra{\tau_C} \otimes \hat{f}_S\right)\,,\nn
\ea
where now $\tau_C\in[0,t_{\rm max})$.~The partial $G$-twirl is thus $n$-independent, so that
\ba
\hat{F}_{f_S,T}(\tau) &=&U_{CS}^\dag(\tau_C+t_{\rm max}\,n)\, \mathcal{G}_{[0,t_{\rm max})}\left( \ket{\tau_C}\!\bra{\tau_C} \otimes \hat{f}_S\right) \nn\\
&&\q\q\q\q\times U_{CS}(\tau_C+t_{\rm max}\,n)\nn\\
&=&U_{CS}^\dag(t_{\rm max}\,n)\, \hat F_{f_S,\phi_C}(\tau_C,n=0)\,  U_{CS}(t_{\rm max}\,n)\,,\nn
\ea
where $\hat F_{f_S,\phi_C}(\tau_C,n=0) $ is the quantisation of the relational observable relative to the non-monotonic clock $\phi_C$ on its $n=0$ cycle.

The crucial question is now:~do the $\hat F_{f_S,T}(\tau)$ in Eq.~\eqref{qrelobs} constitute a $\tau$-parameter family of quantum Dirac observables, i.e.\ are they gauge-invariant, satisfying $[\hat F_{f_S,T}(\tau),\hat C_H]\,\ket{\psi_{\rm phys}}=0$ $\forall\,\ket{\psi_{\rm phys}}\in\ch_{\rm phys}$? Classically, we have seen that the relational observables $F_{f_S,T}(\tau)$ are generically transient observables (cf.~Lemma~\ref{lem_clobs}) so that they are Dirac observables invariant along the entire gauge orbits only when $f_S$ is as periodic as the non-monotonic phase observable $\phi_C$ (cf.~Corollary~\ref{cor_periodicobs}).~We recall that the transient invariance property is infinitesimally related to the fact that $\phi_C$ is canonically conjugate to the clock Hamiltonian $\{\phi_C,H_C\}=1$ (and thereby to the constraint $C_H$) except where it completes its cycles.~Similarly in the quantum theory we have seen from Lemma~\ref{lem_povm} that conjugacy between the $n^{\rm th}$-moment operators and the clock Hamiltonian only holds on the subset $\cd$ of $\ch_C$ quoted in Eq.~\eqref{D}.~At the quantum level, the question then becomes:~does there at least exist a non-trivial subset of $\ch_{\rm phys}$ such that the evaluation of the commutator of the moment operators with the clock Hamiltonian takes a canonical form?~The following answers this question in the negative.

\begin{Lemma}\label{lem_ruin}
Suppose $\hat H_C$ has discrete, non-degenerate spectrum and $\hat\phi_C^{(n)}$ is the $n^{\rm th}$-moment operator of the covariant and periodic clock POVM in Eq.~\eqref{nthmoment}.~Then
\ba
[\hat\phi_C^{(n)},\hat H_C]\,\ket{\psi_{\rm phys}} = i\,n\,\hat\phi_C^{(n-1)}\,\ket{\psi_{\rm phys}}\label{wantthat}
\ea
only holds for the zero vector $\ket{\psi_{\rm phys}}\equiv0$.
\end{Lemma}
\begin{proof}
\proofapp
\end{proof}
This is already suggestive that the power series quantisation in Eq.~\eqref{qrelobs} does not in general yield quantum Dirac observables.~Indeed, we have the following results which establish a quantum analogue of Lemma~\ref{lem_clobs} and its Corollary~\ref{cor_periodicobs} for classical relational observables.

\begin{theorem}
\label{lem_noDirac}
The commutator between the quantisation of relational observables relative to periodic clocks in Eq.~\eqref{qrelobs} and the constraint evaluates to
\ba
&&[\hat F_{f_S,T}(\tau),\hat C_H]=-\f{i}{t_{\rm max}}\ket{0}\!\bra{0}\nn\\
&&\q\otimes U_S^\dag(\tau)\left[U_S(t_{\rm max})\hat f_SU_S^\dag(t_{\rm max})-\hat f_S\right]U_S(\tau).\label{comm:qrelobs}
\ea

Furthermore, $\hat F_{f_S,T}(\tau)$ is a \emph{weak} quantum Dirac observable, i.e.\ $[\hat F_{f_S,T}(\tau),\hat C_H]\approx0$, where $\approx$ is the weak equality, if and only if $\hat f_S$ is \emph{weakly} $t_{\rm max}$-periodic, i.e.\ if and only if 
\ba
{I_C\otimes U_S(t_{\rm max})\hat f_SU_S^\dag(t_{\rm max})\approx I_C\otimes \hat f_S}.\label{condobs}
\ea
In all other cases, $\hat F_{f_S,T}(\tau)$ is neither a weak nor a strong quantum Dirac observable.
\end{theorem}
\begin{proof}\proofapp
\end{proof}

When the relational observable \emph{is} a Dirac observable, note that we can write Eq.~\eqref{qrelobs} in the simplified form
\begin{eqnarray}\label{QRO:fSwp}
    \hat{F}_{f_S,T}(\tau)&\approx
    &\mathcal{G}_{[0,t_{\rm max})}\left( \ket{\tau}\!\bra{\tau} \otimes \hat{f}_S\right) \\
    &=&\f{1}{t_{\rm max}}\int_{0}^{t_{\rm max}} d\phi\, U_{CS}(\phi) \left( \ket{\tau}\!\bra{\tau} \otimes \hat{f}_S \right)U_{CS}^\dag(\phi)\,,\nn
\end{eqnarray}
i.e., as the $\rm{U}(1)$-twirl corresponding to a single clock cycle.

Next, let us explore how the relational observables `evolve' under the constraint flow, encompassing the case when they are not invariant.

\begin{Lemma}\label{lem_clobsQ}
Let $\ket{\psi_{1}}$ be a physical state and $\ket{\psi_2}$ a kinematical state such that $\bra{\phi}\otimes\braket{q|\psi_{1,2}} = \psi_{1,2}(\phi,q)$ are integrable functions of $\phi$ for any choice of basis $\ket{q}$ in the system Hilbert space. Given the quantum relational observables defined in Eq.~\eqref{qrelobs}, let $\alpha^s_{C_H}\cdot\hat{F}_{f_S,T}(\tau):= U_{CS}^{\dagger}(s)\hat{F}_{f_S,T}(\tau)U_{CS}(s)$ and $\braket{\alpha^s_{C_H}\cdot\hat{F}_{f_S,T}(\tau)}_{21}:=\braket{\psi_2|\alpha^s_{C_H}\cdot\hat{F}_{f_S,T}(\tau)|\psi_1}$.~(Note that the latter expression invokes the physical inner product Eq.~\eqref{PIP}.) Then, the quantum relational observables obey the following property:
\begin{align}
&\braket{\alpha^s_{C_H}\cdot\hat{F}_{f_S,T}(\tau)}_{21}\notag\\
&\ = \bra{\psi_2}\frac{1}{ t_{\rm max}}\int_0^{ t_{\rm max}}d\phi\,\ket{\phi}\!\bra{\phi}\notag\\
&\ \ \ \otimes\hat{f}_S\left(\left(\tau+ t_{\rm max}\left\lfloor\frac{s+\phi}{ t_{\rm max}}\right\rfloor\right)-\phi\right)\ket{\psi_1}\,,\label{quant:transient}
\end{align}
where $\hat{f}_S(t) = U_S^\dagger(t)\hat{f}_SU_S(t)$ with $U_S(t)=\exp(-it\hat{H}_S)$.~We take Eq.~\eqref{quant:transient} to be the quantum version of the transient invariance property of classical relational observables established in Lemma \ref{lem_clobs}.
\end{Lemma}
\begin{proof}
\proofapp
\end{proof}

The proof of Theorem~\ref{lem_noDirac}, which is found in Appendix~\ref{app:proof}, demonstrates that the failure of $\hat F_{f_S,T}(\tau)$ to yield quantum Dirac observables when Eq.~\eqref{condobs} is \emph{not} fulfilled is rooted precisely in the term proportional to $\ket{0}\!\bra{0}$ in Lemma~\ref{lem_povm} which ruins the canonical conjugation relations on the \emph{full} clock Hilbert space in the quantum theory.~This in turn reflects in Eq.~\eqref{quant:transient} according to which $\hat F_{f_S,T}(\tau)$ is typically not invariant under the action of the group generated by the constraint, even when restricted to $\ch_{\rm phys}$ where $U_{CS}^{\dagger}(s)\hat{F}_{f_S,T}(\tau)U_{CS}(s)\approx U_{CS}^{\dagger}(s)\hat{F}_{f_S,T}\not\approx\hat{F}_{f_S,T}$ due to Eq.~\eqref{comm:qrelobs}.~Only for weakly $t_{\rm max}$-periodic $\hat f_S$,~i.e.~when Eq.~\eqref{condobs} is satisfied, the RHS of Eq.~\eqref{quant:transient} reduces to the second line of Eq.~\eqref{qrelobs} and $\hat F_{f_S,T}(\tau)$ is weakly invariant under the group generated by the constraint (weak quantum Dirac observable).

Note that Eq.~\eqref{condobs} can be satisfied in a variety of cases.~For example, when $\hat f_S$ is a constant of motion, $[\hat f_S,\hat H_S]=0$, or when  the constraint generates a $t_{\rm max}$-periodic projective representation of $\rm{U}(1)$ too (in which case $U_S(t_{\rm max})=e^{i\theta}I_S$ for some $\theta\in[0,2\pi)$), the condition \eqref{condobs} is even satisfied strongly, in which case $\hat F_{f_S,T}(\tau)$ is a strong quantum Dirac observable.~As we will see later, Eq.~\eqref{condobs} can be fulfilled under weaker conditions and in fact will be satisfied by all system observables compatible with solutions to the constraint.~Generically, however, this set of system observables satisfying Eq.~\eqref{condobs} will be a small subset of system observables.~In other words, generic system observables will typically \emph{not} lead to relational observables that are quantum Dirac observables.

As commented above, if the group representation generated by the constraint matches the group action generated by the clock Hamiltonian (i.e.\ both yield $t_{\rm max}$-periodic projective unitary representations of $\rm{U}(1)$), then the partial $G$-twirl is actually the full $G$-twirl over the group generated by the constraint.~Theorem~\ref{lem_noDirac} shows that the quantisation of the relational observables \emph{does} yield quantum relational Dirac observables in this case for arbitrary $\hat f_S$.~We can then also interpret the failure to produce quantum Dirac observables for arbitrary $\hat f_S\in\cl(\ch_S)$ in generic cases as being related to the fact that there is a mismatch between the groups generated by the clock Hamiltonian and the constraint.~The partial $G$-twirl averages over the former group and is therefore not sufficient to yield invariance under the action of the latter group.~This raises the question, whether we should always be using the full $G$-twirl over the group generated by the constraint in order to define quantum relational Dirac observables.

The following two results answer this question. In short, when $G=\rm{U}(1)$, the full $G$-twirl also averages over the clock's isotropy group $H=\mathbb{Z}_G$, which enforces its periodicity on the system observables and thus provides a valid relational Dirac observable. However, the same observable can also be obtained via the partial $G$-twirl above. By contrast, when $G=(\mathbb{R},+)$ is the translation group, the full $G$-twirl is ill-defined.

\begin{Lemma}\label{lem_fullG}
The $G$-twirl over the full group generated by the constraint $\hat{C}_H$ yields
\begin{eqnarray}
   \f{1}{|H|} \cg_G\left( \ket{\tau}\!\bra{\tau} \otimes \hat{f}_S\right)&=&\f{t_{\rm max}}{N_G} \,\mathcal{G}_{[0,t_{\rm max})}\left( \ket{\tau}\!\bra{\tau} \otimes \hat{f}_S^{H}\right)\nn\\
    &=&\f{t_{\rm max}}{N_G} \,U_{CS}(\tau)\,\hat F_{f_S^{H},T}(\tau)\,U_{CS}^\dag(\tau)\,,\nn
\end{eqnarray}
where
\begin{equation}
    \hat{f}_S^{H} \coloneqq\cg_H\left(\hat{f}_S\right)=\f{1}{|H|}\sum_{z\in\mathbb{Z}_G}\,U_S(zt_{\rm max})\,\hat f_S\,U_S^\dag(zt_{\rm max})\label{isotropyproj}
\end{equation}
is the averaging over the isotropy group $H=\mathbb{Z}_G$ of clock $C$ with respect to the full group $G$ and $|H|$ its cardinality. In the case that $G=\rm{U}(1)$, this assumes that an integer multiple of clock cycles fits into one cycle of $G$.\footnote{If this was not the case, the below expressions would acquire an additional correction term corresponding to the incomplete clock cycle, which would only aggravate the situation for the existence of quantum Dirac observables.} (The normalisation constant is $N_G=\tilde{t}_{\rm max}$ for $G=\rm{U}(1)$, where $\tilde{t}_{\rm max}$ is the analog of $t_{\rm max}$ in Eq.~\eqref{point}, but for the constraint $\hat{C}_H$, and $N_G=2\pi$ for $G=(\mathbb{R},+)$ \cite{Hoehn:2019owq}.)
\end{Lemma} 

\begin{proof}
     \proofapp
\end{proof}

Specifically, we have
\begin{equation}\label{periodicity}
    U_S(zt_{\rm max})\,\hat f_S^{H}\,U^\dag_S(zt_{\rm max})=\hat f_S^{H},
\end{equation}
so the $H$-averaged system observable \emph{is} periodic.\footnote{In the finite case, $H=\mathbb{Z}_n$, this follows from the group law which is addition modulo $n$ and applies because also the projective unitary representation of the gauge group $G=\rm{U}(1)$ is periodic in $t$.} We thus note that when $H=\mathbb{Z}$, i.e.\ when $G$ is the translation group, the full $G$-twirl is generally not well-defined, whereas the partial one can still yield sensible results;\footnote{E.g., suppose $\hat{f}_S$ is gauge-invariant already, i.e.\ $[\hat{f}_S,H_S]=0$, and $H=\mathbb{Z}$. Then $\cg_\mathbb{R}\left(\ket{\tau}\!\bra{\tau}\otimes \hat{f}_S\right)=|\mathbb{Z}|\f{t_{\rm max}}{2\pi}I_C\otimes\hat{f}_S$, and so the full $G$-twirl counts the infinitely many times $|\mathbb{Z}|$ that $C$ reads $\tau$ along the orbit generated by $\hat C_H$ and multiplies the gauge-invariant observable $f_S$ by that number. Note that this is \emph{distinct} from the  full $G$-twirl in the presence of monotonic clocks, which under the same circumstances yields the desired result $\cg_{\mathbb{R}}\left(\ket{\tau}\!\bra{\tau}\otimes \hat{f}_S\right)=I_C\otimes\hat{f}_S$, where $\tau$ corresponds to the reading of a monotonic clock \cite{Hoehn:2019owq}; this is in line with that clock reading $\tau$ only once along the orbit. Similarly, the partial $G$-twirl yields the correct result in the present context, namely $\cg_{[0,t_{\rm max})}\left(\ket{\tau}\!\bra{\tau}\otimes \hat{f}_S\right)=I_C\otimes\hat{f}_S$.} we will also see this in Example~\ref{ex_9} below. In the compact case, i.e.\ $G=\rm{U}(1)$ and $H=\mathbb{Z}_n$, it is easy to see that the isotropy group $G$-twirl is, in fact, a projector on the system algebra, namely the projector into its $t_{\rm max}$-periodic subalgebra. Furthermore, given that $\hat f_S^{H}$ satisfies Eq.~\eqref{condobs} (even strongly), it is clear that the full $G$-twirl yields a Dirac observable -- when it is defined.
\begin{corol} \label{corFullGwithHfinite}
    For $H$ finite, the full $G$-twirl yields a strong Dirac observable,
    \begin{eqnarray} \label{GtwirlCommutator}
        \left[\cg_G\left( \ket{\tau}\!\bra{\tau} \otimes \hat{f}_S\right),\hat C_H\right]=0,
    \end{eqnarray}
    and we have
    \begin{equation}
        \cg_G\left( \ket{\tau}\!\bra{\tau} \otimes \hat{f}_S\right)=\cg_G\left( \ket{\tau}\!\bra{\tau} \otimes \hat{f}_S^{H}\right)\approx|H|\f{t_{\rm max}}{\tilde{t}_{\rm max}}\,\hat{F}_{f_S^{H},T}(\tau)\,.
    \end{equation}
\end{corol}
Furthermore, if $\hat{f}_S$ is weakly $t_{\rm max}$-periodic, then ${I_C\otimes \hat f_S^{H}\approx I_C\otimes \hat f_S}$, and thus $\hat{F}_{f_S^{H},T}(\tau)\approx\hat{F}_{f_S,T}(\tau)$, and hence, the relational Dirac observable obtained via the full $G$-twirl coincides, up to a prefactor, with one obtained via the partial $G$-twirl. We have thus not gained anything new and it is sufficient to restrict to the partial $G$-twirl in what follows -- provided we insert periodic $S$ observables. This is consistent with the observations in \cite{delaHamette:2021oex}, but provides a more explicit illustration of them. ~In other words, since the full $G$-twirl corresponds to averaging over multiple winding numbers, we see that unless one explicitly considers a physical counting system for these winding numbers, the result is either ill-defined (when $G$ is the translation group), or equivalent to the theory without counting winding numbers (when $G=\rm{U}(1)$).

\subsection{Examples}\label{Sec:quantexamples}

Let us illustrate the content of Theorem~\ref{lem_noDirac} for the examples considered in the previous section.
\begin{example}[\bf{Commensurate oscillators}]\label{ex_7}
We continue the discussion of Example~\ref{ex_3}, i.e.\ two oscillators with fixed energy and commensurate frequencies $\omega_t/\omega=m_1/m_2$, $m_1,m_2\in\mathbb{N}$, subject to the constraint in Eq.~\eqref{WDW2}.~Relational observables have been studied before in this model, however, only for the special case that $\omega_t=\omega$ and not with covariant and periodic clock POVMs  \cite{Rovelli:1990jm,Rovelli:1989jn,rovelliQuantumGravity2004,Bojowald:2010qw}.~We choose the first oscillator (the $C$ tensor factor) as our clock and ask what the position $\hat q$ and momentum $\hat p$ of the second is, when the phase observable of the first reads $\tau$.~The corresponding relational observables computed according to Eq.~\eqref{qrelobs} read
\ba
\hat F_{q,T}(\tau)&=&\f{1}{t_{\rm max}}\int_0^{t_{\rm max}}d\phi\,\ket{\phi}\!\bra{\phi}\otimes[\hat q\,\cos\left((\phi-\tau)\,\omega\right)\nn\\
&&\q\q\q\q-\f{\hat p}{m\omega }\sin\left((\phi-\tau)\,\omega\right)],\label{HOrelob}\\
\hat F_{p,T}(\tau)&=&\f{1}{t_{\rm max}}\int_0^{t_{\rm max}}d\phi\,\ket{\phi}\!\bra{\phi}\otimes\left[\hat p\,\cos\left((\phi-\tau)\,\omega\right)\right.\nn\\
&&\left.\q\q\q\q +m\omega\hat q\,\sin\left((\phi-\tau)\,\omega\right)\right].\nn
\ea
and are the direct quantisation of the classical relational observables in Eq.~\eqref{horelobs}.~The commutator of these relational observables with the constraint can be easily evaluated and in the case of the first one yields
\ba
[\hat F_{q,\phi_C}(\tau)&,&\hat C_H]=-\f{i}{t_{\rm max}}\bigl[\ket{0}\!\bra{0}\nn\\
&\otimes&\bigl(\hat q\left(\cos\left((t_{\rm max}-\tau)\,\omega\right)-\cos(\tau\omega)\right)\nn\\
&-&\f{\hat p}{\omega m}\left(\sin\left((t_{\rm max}-\tau)\,\omega\right)+\sin(\tau\omega)\right)\bigr)\bigr].\label{HOrelobcom}
\ea
Recall from Example~\ref{ex_3} that $m_1\,t_{\rm max}=m_2\,t_{{\rm max},S}$.~Hence, $t_{\rm max}\,\omega=2\pi\f{m_2}{m_1}$ and the commutator in Eq.~\eqref{HOrelobcom} only vanishes provided that $m_2/m_1\in\mathbb{N}$.~Hence, in that case, $\hat F_{q,T}(\tau)$ is a strong quantum relational Dirac observable.~The system $S$ undergoes exactly $m_2/m_1\in\mathbb{N}$ cycles for every cycle of the clock $C$, i.e.\ the period of the clock is equal to or an integer multiple of the period of the system $S$.~In particular, from Example~\ref{ex_3} we already know that the constraint generates a $t_{\rm max}$-periodic projective unitary representation of $\rm{U}(1)$ in this case and so the condition in Eq.~\eqref{condobs} of Theorem~\ref{lem_noDirac} is satisfied even strongly.~By contrast, when $m_2/m_1\notin\mathbb{N}$, the expression in Eq.~\eqref{HOrelob} is \emph{not} a quantum Dirac observable simply because $\hat q$ is not periodic.~The argumentation for $\hat F_{p,T}(\tau)$ is identical.

This is the analogue of the situation discussed in Example~\ref{ex_1} for the classical model and, for $E$ sufficiently large, there are typically enough states in $\ch_{\rm phys}$ (\textit{cf}.\ Example~\ref{ex_3}) to obtain a semiclassical limit that matches the classical relational dynamics \cite{Rovelli:1990jm,Rovelli:1989jn,rovelliQuantumGravity2004,Bojowald:2010qw}.
\end{example}

\begin{example}[\bf{Incommensurate oscillators}]\label{ex_8} Let us return to the example of two harmonic oscillators with incommensurate frequencies studied in Example~\ref{ex_4}.~Thus, the constraint is once more given by Eq.~\eqref{WDW2}, but subject to the condition that $\omega_t/\omega\notin\mathbb{Q}$.~The quantisation of the classical relational observables in Eq.~\eqref{horelobs} is again given by Eq.~\eqref{HOrelob}.~However, in this case it is impossible for the commutator in Eq.~\eqref{HOrelobcom} to vanish because we have $t_{\rm max}/t_{{\rm max},S}=\omega/\omega_t$ and this implies $t_{\rm max}\omega=2\pi\f{\omega}{\omega_t}$, which is an irrational multiple of $2\pi$. Hence, in this case, neither $\hat F_{q,T}(\tau)$ nor $\hat F_{p,T}(\tau)$ are Dirac observables.

One might wonder whether there exist \emph{any} quantum relational Dirac observables in this model.~However, we have already seen in Example~\ref{ex_4} that in this case, the physical Hilbert space is at most one-dimensional, $\ch_{\rm phys}\simeq\mathbb{C}$.~Up to a multiplicative real constant, there is a \emph{unique} quantum Dirac observable on $\ch_{\rm phys}$: the identity $I_C\otimes I_S$ which can be obtained as a relational observable by inserting $f_S=I_S$ into Eq.~\eqref{qrelobs}.~For instance, $I_C\otimes\hat H_S\approx -\omega_t\left(n_1+\f{1}{2}\right)\, I_C\otimes I_S$, where $\approx$ denotes a weak equality, i.e.\ equality on $\ch_{\rm phys}$.

It is therefore clear that in this model there do not exist any quantum relational Dirac observables which describe how some system degrees of freedom evolve relative to the clock.~This is the analogous situation to the classical theory of Sec.~\ref{sss_relobsperclock} where also a drastic difference between commensurate and incommensurate frequencies did occur.~In the latter case, as discussed in Example~\ref{ex_1}, there were no non-trivial periodic $S$-observables yielding global classical Dirac observables.
\end{example}

\begin{example}[\bf{Oscillator clock and free particle}]\label{ex_9}
Next, we consider the example of the harmonic oscillator clock and the free particle from Sec.~\ref{ssec_ex2} in the quantum theory.~The constraint is
\ba
\hat C_H= \f{\hat p_t^2}{2 m_t}+\f{m_t\omega^2_t}{2}\,\hat t^2\otimes I_S-I_C\otimes\f{\hat p^2}{2m}.\nn
\ea
As already emphasised in Example~\ref{ex_5}, this is an example for case (b) in Sec.~\ref{ssec_diracgen} where the clock Hamiltonian generates a $t_{\rm max}$-periodic projective unitary representation of $\rm{U}(1)$, while the constraint $\hat C_H$ generates a unitary representation of the translation group $\mathbb{R}$.~We clearly have a mismatch of the two group actions and thus suspect that the quantisation of the classical relational observable $F_{q,T}(\tau)$ in Eq.~\eqref{hopartobs} of Example~\ref{ex_2}, which was a transient observable in the classical theory, does not yield a Dirac observable also in the quantum theory.~Indeed, computing the relational observable encoding the position $\hat q$ of the particle when the harmonic oscillator clock reads $\tau$ via Eq.~\eqref{qrelobs} gives
\ba\label{freenoobs}
\hat F_{q,T}(\tau)=I_C\otimes\left(\hat q-\f{\tau\hat p}{m}\right)+\hat\phi_C^{(1)}\otimes\f{\hat p}{m},
\ea
which coincides with the direct quantisation of the classical expression in Eq.~\eqref{hopartobs}.~Its commutator with the constraint does not vanish
\ba\label{eq:FqTCHcomm}
[\hat F_{q,T}(\tau),\hat C_H]=-\f{i}{m}\,\ket{0}\!\bra{0}\otimes\hat p,
\ea
being ruined by the correction term to the canonical commutation relations in Lemma~\ref{lem_povm}.~Thus, similarly to the classical situation of Example~\ref{ex_2}, $\hat F_{q,T}(\tau)$ is neither a strong nor weak Dirac observable.

We may thus wonder what happens when, instead of the partial $G$-twirl, we employ the full version of Lemma~\ref{lem_fullG}.~A formal calculation shows that we obtain an ill-defined expression
\begin{equation}
    \cg\left( \ket{\tau}\!\bra{\tau} \otimes\hat q\right)\!=\!|\mathbb{Z}|\f{t_{\rm max}}{2\pi}\left(I_C\otimes\left(\cg_\mathbb{Z}(\hat q)\!-\!\f{\tau\hat p}{m}\right)\!+\!\hat\phi_C^{(1)}\otimes\f{\hat p}{m}\right)\nn
\end{equation}
with divergent prefactor $|H|=|\mathbb{Z}|$, where
\begin{equation}
    \cg_\mathbb{Z}(\hat q)=\f{1}{|\mathbb{Z}|}\left(\hat q-t_{\rm max}\f{\hat p}{m}\sum_{z\in\mathbb{Z}}z \right)
\end{equation}
is the ill-defined isotropy group average in this case (also the bi-infinite sum is not well-defined). This divergence is rooted in the fact that $\ket{\tau}\!\bra{\tau}=\ket{\tau_C}\!\bra{\tau_C}$ for some $\tau_C\in[0,t_{\rm max})$ and that there exist infinitely many clock cycles in which the periodic clock reads $\tau_C$. The question what the position $\hat q$ of the free particle is when the harmonic oscillator clock reads $\tau_C$ thus has infinitely many answers and the full $G$-twirl sums over all of them.

On the contrary, the momentum $\hat p$ of the particle, which is a constant of the motion as in the classical theory, via Eq.~\eqref{qrelobs} yields $\hat F_{p,T}(\tau)=I_C\otimes\hat p$.~Condition~\eqref{condobs} in Lemma~\ref{lem_ruin} is thus satisfied strongly and $\hat F_{p,T}(\tau)$ is a strong quantum Dirac observable. 

As we saw in Example~\ref{ex_5}, the physical Hilbert space of this example is infinite-dimensional, there is thus an infinite-dimensional algebra of Dirac observables as well. These can be obtained by projecting the system algebra $\mathcal{A}_S=\mathcal{B}(\mathcal{H}_S)$ onto its $t_{\rm max}$-periodic subalgebra and inserting the resulting operators into Eq.~\eqref{qrelobs} (unlike for $\hat q$, this will yield valid operators in some cases); this will become clearer when discussing Page-Wootters reduction in the next section. For example, using Eq.~\eqref{physsobs} therein, one can project $\hat q$ onto the observable $\hat{q}^{\rm phys} \coloneqq \int_\mathbb{R}ds\, q \ket{\pi(q)}\!\bra{q}$, where 
\ba
    \ket{\pi(q)} \ce \f{1}{\sqrt{2\pi}} \sum_{\substack{\mu = \pm 1 \\ n\in\mathbb{N}_{0}}} & e^{-i \mu p_{n} q} \ket{p=p_{n}} 
\ea
%\ba
%    \ket{\pi(q)} \ce \f{1}{\sqrt{2\pi}} \sum_{\substack{\mu = \pm 1 \\ n\in\mathbb{N}_{0}}} & e^{-i \mu \sqrt{2m \omega_{t}(n+\frac{1}{2})} q} \nn\\
%& \ket{p=\mu \sqrt{2m \omega_{t}(n+\frac{1}{2})}} 
%\ea
with $p_{n}\coloneqq \sqrt{2m \omega_{t}(n+\frac{1}{2})}$. This is the projection of $\ket{q}$ onto those eigenstates of $\hat{H}_{S}$ which are compatible with the constraint. While $\hat{q}^{\rm phys}$ is not periodic, specifically
\ba
U_S(t_{\rm max})\hat{q}^{\rm phys}U_S^\dag(t_{\rm max})=-\hat{q}^{\rm phys}U_S^\dag(t_{\rm max})\neq\hat{q}^{\rm phys}, \nn
\ea
it is weakly periodic, i.e.\ $I_C\otimes U_S(t_{\rm max})\hat{q}^{\rm phys}U_S^\dag(t_{\rm max})\approx I_C\otimes \hat{q}^{\rm phys}$, and thus $[\hat F_{q^{\rm phys},T}(\tau),\hat C_H]\approx 0$.

\end{example}
%\sum_{\mu = \pm 1}\sum_{n\in\mathbb{N}_{0}} e^{-i \mu \sqrt{2m \omega_{t}(n+\frac{1}{2}} q} 
%    \nn\\
%&     \ket{p=\mu \sqrt{2m \omega_{t}(n+\frac{1}{2}}}
\begin{example}[\bf{Two-level clock and free particle}]\label{Ex:2levelclock3}
Lastly, we continue with the example of the two-level clock and the free particle from Sec.~\ref{ssec_ex2} and whose Hamiltonian operators are given in Eq.~\eqref{eq:qubitHSandfreeHS}.~This also provides us with an application of our formalism to the case when no classical analogue exists.~As already emphasised in Example~\ref{Ex:2levelclock2}, there is again a mismatch between the $U(1)$-representation generated by the clock Hamiltonian and the representation of the translation group generated by the constraint.~The relational observable $\hat{F}_{q,T}(\tau)$ encoding the position $\hat q$ of the particle when the qubit clock reads $\tau$ is again given by Eq.~\eqref{freenoobs} with the clock operator $\hat{\phi}_C=\hat{\phi}_C^{(1)}$ now given by Eq.~\eqref{eq:qubitclockoperator}.~As discussed in Example~\ref{Ex:2levelclock1} (cf.~Eq.~\eqref{eq:phiCHCcomm}), $\hat{\phi}_C$ and $\hat{H}_C$ fail to be a Heisenberg pair except on the subspace given in Eq.~\eqref{eq:2levelDsubspace}.~This reflects into the commutator $[\hat{F}_{q,T}(\tau),\hat{C}_H]$ which reads as in Eq.~\eqref{eq:FqTCHcomm} and does not vanish on $\ch_{\rm phys}$, as it can be easily checked from the expression of physical states given in Eq.~\eqref{eq:qubitpsiphys}.~The relational observable $\hat{F}_{q,T}(\tau)$ is thus neither a strong nor a weak Dirac observable.~By contrast, the relational observable $\hat{F}_{p,T}(\tau)$ encoding the momentum $\hat p$ of the particle (a constant of the motion) when the qubit clock reads $\tau$ yields $\hat{F}_{p,T}(\tau)=I_C\otimes\hat{p}$, as in the previous example, and is a strong Dirac observable. 

Given that $\mathcal{H}_{\rm phys}\simeq\mathbb{C}^2$ (cf.\ Example~\ref{Ex:2levelclock2}), it is clear, however, that the algebra $\mathcal{A}_{\rm phys}=\mathcal{L}(\mathcal{H}_{\rm phys})$ of Dirac observables is four-dimensional.~Again, these can be obtained by projecting the system algebra into its $t_{\rm max}$-periodic subalgebra and inserting the resulting operators into the partial $G$-twirl in Eq.~\eqref{qrelobs}.~This will become clear in the following section. 
\end{example}

\section{Quantum reduction for periodic clocks}\label{sec_reduction}

Having seen that relational observables relative to periodic clocks only promote to quantum Dirac observables if the system observable is periodic itself and are otherwise not gauge-invariant, one might wonder whether it is possible to construct a useful relational quantum dynamics via quantum reductions to the ``clock perspective" as in the Page-Wootters formalism \cite{pageEvolutionEvolutionDynamics1983,woottersTimeReplacedQuantum1984,Page:1984qt,delaHamette:2021oex,Gambini:2006ph, Gambini:2006yj,Gambini:2008ke, giovannettiQuantumTime2015,Smith:2017pwx,Hoehn:2021wet,Cafasso:2024zqa,Smith:2019imm,Dolby:2004ak, castro-ruizTimeReferenceFrames2019,Boette:2018uix,Diaz:2018uny,Diaz:2019xie,leonPauliObjection2017,Nikolova:2017huj,baumann2019generalized,Favalli:2020gmx,Foti:2020erm,Hoehn:2019owq,HLSrelativistic,Hausmann:2023jpn,diaz2023parallel} or in quantum deparametrisations \cite{hoehnHowSwitchRelational2018,Hoehn:2018whn}.~For monotonic clocks, it was shown in \cite{Hoehn:2019owq,HLSrelativistic} (see also the related results in \cite{Chataignier:2020fys,ChataignierT,delaHamette:2021oex}) that the Page-Wootters formalism and quantum deparametrisations are fully equivalent to the relational quantum dynamics encoded in relational Dirac observables on the physical Hilbert space.~The relational Dirac observables in Dirac quantisation provide a manifestly gauge-invariant \emph{clock-neutral picture} of the relational quantum dynamics, the Page-Wootters formalism was shown to yield a \emph{relational Schr\"odinger picture}, while quantum deparametrisations produce a \emph{relational Heisenberg picture}.~In particular, the relational Schr\"odinger and Heisenberg pictures constitute quantum analogs of gauge-fixed formulations of the clock-neutral picture.~This equivalence of three faces of the same dynamics was termed \emph{the trinity of relational quantum dynamics}.

Given that this equivalence holds through a quantum analog of gauge-fixing for monotonic clocks and we have already seen that `projectors' $\ket{\tau}\!\bra{\tau}$ onto clock readings do not provide good gauge-fixing conditions for periodic clocks, one might worry about the status of the trinity for periodic clocks.~As we will now exhibit, the trinity in fact survives entirely:~relational observables in Dirac quantisation, the Page-Wootters formalism and quantum deparametrisations continue to be equivalent for periodic clocks. However, this equivalence is more subtle than in the case of monotonic clocks \cite{Hoehn:2019owq,HLSrelativistic} and is rooted in the fact that solving the constraints induces the clock periodicity also on the states and observables of the system.~Specifically, relational observables associated with system observables that are compatible with solutions to the constraints will also be quantum Dirac observables and thus gauge-invariant.~The challenge then is, however, that the set of periodic states and observables for the system may be `small'.~Heuristically, one may paraphrase this as meaning that \emph{in a universe evolving relative to a periodic clock, only observables and states that are periodic for the evolving degrees of freedom will be physical}.~In Sec.~\ref{sec_clockchanges}, we will compare observables described with respect to periodic and aperiodic clocks.

\subsection{The physical system Hilbert space} \label{sec_HSphys}

Before we discuss the Page-Wootters formalism and quantum deparametrisations in the context of periodic clocks, it is worthwhile to construct the \emph{physical system Hilbert space} $\ch_S^{\rm phys}$, i.e.\ the Hilbert space for the system $S$ that is compatible with solutions to the constraint.~In other words, we construct the space spanned by the system energy eigenstates subject to the spectrum condition in Eq.~\eqref{spec}.~This Hilbert space will be the image of the physical Hilbert space $\ch_{\rm phys}$ under the Page-Wootters and quantum symmetry reduction maps defined below.

The two cases in Sec.~\ref{sec_Dirac}, namely (a), $\sigma_{S\vert C}$, defined in Eq.~\eqref{spec}, lies in the point spectrum of $\hat H_S$, and (b), $\sigma_{S\vert C}$ lies in the continuous spectrum of $\hat H_S$, need to be distinguished due to the following subtlety.~As discussed in the previous section, in case (a), the physical Hilbert space $\ch_{\rm phys}$ is a subspace of the kinematical Hilbert space $\ch_{\rm kin}$, while this is \emph{not} the case for (b).~Similarly now, in case (a), the physical system Hilbert space $\ch_S^{\rm phys}$ will be a subspace of the kinematical system Hilbert space $\ch_S$, while in case (b) this will \emph{not} be true.~To see this, note that we can define the `projector' 
\ba
\Pi_{\sigma_{S\vert C}}:=\sum_{E\in\sigma_{S\vert C}}\sum_{\sigma_E}\,\ket{E,\sigma_E}_S\!\bra{E,\sigma_E}\label{physsproj}
\ea
from the kinematical system Hilbert space $\ch_S$ to  what will become $\ch_S^{\rm phys}$.~We write projector in quotation marks because, while in case (a) it satisfies $\Pi_{\sigma_{S\vert C}}^2=\Pi_{\sigma_{S\vert C}}$ and is thus the orthogonal projector onto $\ch_S^{\rm phys}$, it is an improper projector in case (b) because then we have a discrete sum of Dirac delta function normalised energy eigenstates so that $\Pi_{\sigma_{S\vert C}}^2$ yields a divergence.~This is analogous to the situation with $\Pi_{\rm phys}$ in Sec.~\ref{sec_Dirac}.

Let
\ba
\ket{\psi_S}=\,\,\,\,{{\intsum}_{E}}\,\sum_{\sigma_E}\,\psi_S(E,\sigma_E)\,\ket{E,\sigma_E}_S\label{physsys}
\ea
be an arbitrary state in the kinematical system Hilbert space $\ch_S$.~The corresponding physical system state is
\ba
\ket{\psi_S^{\rm phys}}&\ce&\Pi_{\sigma_{S\vert C}}\,\ket{\psi_S}\nn\\
&=& \sum_{E\in\sigma_{S\vert C}}\sum_{\sigma_E}\,\psi(E,\sigma_E)\,\ket{E,\sigma_E}_S,\label{physsysstate}
\ea
which in case (b), being a sum of improper energy eigenstates, is not normalizable in $\ch_S$ and thus not contained in it.~Rather, it should be understood in a rigged Hilbert space sense as being a distribution over the kinematical system states.~Since distributions can be integrated against states, it is natural to define the inner product for physical system states in both case (a) and (b) as\footnote{In fact, a more precise version would replace the r.h.s.\ by $\braket{\phi_S^{\rm phys}|\psi_S}_S$, taking into account the distributional nature of physical system states in case (b). However, in line with much of the physics literature on the subject, we use this more sloppy version in what follows. Given the symmetry of $\Pi_{\sigma_{S|C}}$, all below statements regarding expectation values and overlaps in the physical system inner product also hold in the more precise version. In the same vein, in case (b), one should also rather think of the l.h.s.\ in Eq.~\eqref{physsysstate} as $\bra{\psi_S^{\rm phys}}$.}
\ba
\braket{\phi_S^{\rm phys}|\psi_S^{\rm phys}}_{\ch_S^{\rm phys}}\ce\braket{\phi_S|\,\Pi_{\sigma_{S\vert C}}\,|\psi_S}_S=\braket{\phi_S|\psi_S^{\rm phys}}_S,\nn\\\label{sysPIP}
\ea
where $\ket{\psi_S}$ is any member of the equivalence class of kinematical system states in $\ch_S$ that map under $\Pi_{\sigma_{S\vert C}}$ to the physical system state $\ket{\psi_S^{\rm phys}}$, and similarly for $\ket{\phi_S}$.~Here, $\braket{\cdot|\cdot}_S$ is the inner product on $\ch_S$ and since $\Pi_{\sigma_{S\vert C}}$ is symmetric, the inner product indeed only depends on the equivalence class of kinematical system states.~In case (a) it is clear that the physical system inner product coincides with the standard inner product on $\ch_S$, but restricted to the subspace $\ch_S^{\rm phys}\subset\ch_S$ of states in Eq.~\eqref{physsysstate}.~In case (b), the situation is more subtle.~In order to turn the image $\Pi_{\sigma_{S\vert C}}(\ch_S)$ into a Hilbert space, it may be necessary to divide out spurious zero-norm physical system states and it will be necessary to Cauchy complete in the norm defined by Eq.~\eqref{sysPIP}.~The result of this will also be denoted by $\ch_S^{\rm phys}$ and in both cases (a) and (b) we shall refer to it as the \emph{physical system Hilbert space}.~We stress: in case (b) it is \emph{not} a subspace of $\ch_S$.

The system energy eigenstates compatible with the constraints form an orthonormal basis for $\ch_S^{\rm phys}$.~In other words, in case (b), the physical system inner product replaces the Dirac delta function normalisation of $\braket{\cdot|\cdot}_S$ of energy eigenstates with the Kronecker delta orthonormalisation.

\begin{Lemma}\label{lem_sysphysIP}
Let $E,E'\in\sigma_{S\vert C}$. Then 
\ba
\braket{E,\sigma_E|E',\sigma_{E'}}_{\ch_S^{\rm phys}} = \delta_{E,E'}\delta_{\sigma_E,\sigma_{E'}}.\nn
\ea
\end{Lemma}
\begin{proof}\proofapp\end{proof}

Any state in $\ch_S^{\rm phys}$ can be obtained from a state in $\ch_S$ via `projection' with $\Pi_{\sigma_{S\vert C}}$ (the wave function $\psi(E,\sigma_E)$ for $E\in\sigma_{S\vert C}$ can be extended to other eigenvalues in $\spec(\hat H_S)$ in many ways while maintaining square integrability/summability).

Below, we will also need to discuss observables on $\ch_S^{\rm phys}$.~Let $\hat f_S\in\cl(\ch_S)$ be an arbitrary operator on the system Hilbert space $\ch_S$.~Then we define the associated \emph{physical system observable} $\hat f_S^{\rm phys}\in\cl(\ch_S^{\rm phys})$ by
\ba\label{physsobs}
\hat f_S^{\rm phys}\ce\begin{cases}
 \hat f_S\,\restriction{\mathcal{H}_S^{\rm phys}}& \text{if }[\hat f_S,\hat H_S]=0, \\
 \Pi_{\sigma_{S\vert C}}\, \hat f_S\,\restriction{\mathcal{H}_S^{\rm phys}}   & \text{otherwise}.
\end{cases}
\ea
We emphasise that the domain of $\hat f_S^{\rm phys}$ is $\ch_S^{\rm phys}$, and Eq.~\eqref{physsobs} thus includes a modification of the domain of $\hat f_S$ (denoted by the symbol $\restriction$).~We have distinguished the case that $\hat f_S$ is a constant of motion as in that case it commutes with $\Pi_{\sigma_{S\vert C}}$, and therefore if we did project it with the latter, it would produce a divergence in case (b) when acting on physical system states.~The reason we have not conjugated $\hat f_S$ in the second case is again case (b) as it would then too generate a divergence when acting on $\ch_S^{\rm phys}$.~In case (a) this does not make a difference on $\ch_S^{\rm phys}$.

It will be convenient to define the \emph{system weak equality} $\approx_S$ indicating equality of two system operators $\hat O_1,\hat O_2$ on  the physical system Hilbert space
\ba
\hat O_1\approx_S\hat O_2 \q\Leftrightarrow\q \hat O_1\,\ket{\psi_S^{\rm phys}}=\hat O_2\,\ket{\psi_S^{\rm phys}},\nn\\
\q\forall\,\ket{\psi_S^{\rm phys}}\in\ch_S^{\rm phys}.
\ea
In particular, we note that owing to Eqs.~\eqref{projrep} and~\eqref{WDW} we have
\ba
I_C\otimes U_{S}^\dag( z t_{\rm max})\,\ket{\psi_{\rm phys}}= e^{i\varphi\,z}I_C\otimes I_S\,\ket{\psi_{\rm phys}}\,,\label{eq:speccon}
\ea
for some $\varphi\in[0,2\pi)$.~Therefore, solving the constraint induces the clock periodicity on the physical system Hilbert space, i.e.\
\ba
U_S(z t_{\rm max})\approx_S e^{iz\varphi}\,I_S,\q\forall z\in\mathbb{Z}_G\,,\label{periodicb}
\ea
with $\mathbb{Z}_G$ the set of integers counting the clock cycles which fit into one period of $G$.~Importantly, every physical system observable, Eq.~\eqref{physsobs}, is therefore weakly $t_{\rm max}$-periodic
\ba
U_S(zt_{\rm max})\,\hat f_S^{\rm phys}\,U_S^\dag(zt_{\rm max})&\approx_S&\hat f_S^{\rm phys},\label{periodicc}\\
&&\forall\,\hat f_S^{\rm phys}\in\cl(\ch_S^{\rm phys})\nn
\ea
and thus satisfies the condition in Eq.~\eqref{condobs} of Theorem~\ref{lem_noDirac}.

In summary, solving the constraint Eq.~\eqref{WDW} can drastically change the permissible set of system states and observables.~Through the constraint the clock Hamiltonian induces the clock periodicity on the physical system Hilbert space $\ch_S^{\rm phys}$ and observables on it.~Depending on the system $S$, the set of $t_{\rm max}$-periodic states and observables may be `small' compared to the original set of kinematical states and observables of $S$, as indicated by some of the examples.~However, as long as the constraint can be solved in the quantum theory, these will exist, but may be trivial as in the extreme example of the incommensurate oscillators.~In case (b), $\ch_S^{\rm phys}$ is not a subspace of the original system Hilbert space $\ch_S$.~The situation that the physical system Hilbert space is no longer contained in the kinematical system Hilbert space does not arise for monotonic clocks with continuous spectrum Hamiltonians \cite{Hoehn:2019owq,HLSrelativistic}.

\subsection{The relational Schr\"{o}dinger picture (the Page-Wootters formalism)}
\subsubsection{State reductions and embeddings}

Examples of periodic clocks in the Page-Wootters formalism can be found in~\cite{Smith:2017pwx,Favalli:2020gmx,Cafasso:2024zqa}. Here we give a general formulation. As in \cite{Hoehn:2019owq,HLSrelativistic}, we define the Page-Wootters reduction map $\calr_{\mathbf S}(\tau):\ch_{\rm phys}\rightarrow\ch^{\rm phys}_S$ to the ``perspective of clock $C$"  through the conditioning on the clock reading $\tau$:
\ba
\calr_{\mathbf S}(\tau)\ce\bra{\tau}\otimes I_S.\label{PWred}
\ea 
Suppose $\tau=\tau_C+z t_{\rm max}$ for $\tau_C\in[0,t_{\rm max})$ and some $z\in\mathbb{Z}_G$. Owing to Eq.~\eqref{projrep}, for periodic clocks we have in addition the property
\ba
\calr_{\mathbf S}(\tau) = e^{iz\varphi}\calr_{\mathbf S}(\tau_C)\,,\q z\in\mathbb{Z}.\label{redperiod}
\ea

The standard argument of the Page-Wootters formalism~\cite{pageEvolutionEvolutionDynamics1983, woottersTimeReplacedQuantum1984,Page:1984qt,delaHamette:2021oex,giovannettiQuantumTime2015,Smith:2017pwx,Hoehn:2021wet,Cafasso:2024zqa,Smith:2019imm,Dolby:2004ak, castro-ruizTimeReferenceFrames2019,Boette:2018uix,Hoehn:2019owq,HLSrelativistic,Hausmann:2023jpn,diaz2023parallel} applies, whereby one can use the resolution of the identity given in Eq.~\eqref{resolid} to write the physical state as a so-called ``history state'':
\ba \label{eHistory}
    \ket{\psi_{\rm phys}}=\f{1}{t_{\rm max}}\int_{0}^{t_{\rm max}} d \phi \,\ket{\phi}\otimes\ket{\psi^{\rm phys}_S(\phi)}
\ea
and show that the reduced system states
\ba
\ket{\psi^{\rm phys}_S(\tau)}\ce\calr_{\mathbf S}(\tau)\,\ket{\psi_{\rm phys}}\label{schrodexp}
\ea
satisfy the Schr\"odinger equation on $\ch_S^{\rm phys}$
\ba
i\f{d}{d\tau}\ket{\psi^{\rm phys}_S(\tau)} = \hat H_S\,\ket{\psi^{\rm phys}_S(\tau)},\label{schrod}
\ea
where we recall that the domain of $\hat H_S$ has been restricted to $\ch_S^{\rm phys}$, via Eq~\eqref{physsobs}.~Indeed, using Eq.~\eqref{crap}, one finds $\calr_{\mathbf S}(\tau)\,\ket{\psi_{\rm phys}}=U_S(\tau)\,\ket{\psi^{\rm phys}_S}$, where $\ket{\psi_S^{\rm phys}}$ is given by Eq.~\eqref{physsysstate} with wave function
\ba
\psi_S(E,\sigma_E)=e^{-ig(-E)}\psi_{\rm kin}(-E,E,\sigma_E).\q\label{Heisstate}
\ea

Note, however, that on account of Eqs.~\eqref{redperiod} and~\eqref{periodicb}, the relational Schr\"odinger state dynamics is now periodic up to phase
\ba \label{ePeriodicPWstate}
\ket{\psi_S(\tau)}=e^{iz\varphi}\ket{\psi_S(\tau_C)}.
\ea
We will shortly interpret this relational Schr\"odinger state dynamics in the light of our observation that relational observables associated with periodic clocks break gauge-invariance, unless the encoded $S$-observable is periodic too.

Due to this phase, the inverse of this map will be clock cycle dependent. Every unravelled clock reading $\tau$ will lie in an interval $[zt_{\rm max},(z+1) t_{\rm max})$, for some $z\in\mathbb{Z}_G$.
For $\tau\in[z t_{\rm max},(z+1)t_{\rm max})$, the inverse $\calr_{\mathbf S}^{-1}(\tau):\ch_S^{\rm phys}\rightarrow\ch_{\rm phys}$ of this reduction map reads
\ba\label{PWinvred1}
\calr_{\mathbf S}^{-1}(\tau)&\ce&\f{1}{t_{\rm max}}\int_{0}^{t_{\rm max}} d\phi\,\ket{\phi}\otimes U_S(\phi-\tau).
\ea
Owing to Eq.~\eqref{periodicb}, it \emph{only} holds on the physical system Hilbert space that\footnote{Note that one could also define the inverse reduction map as
\ba
\calr_{\mathbf S}^{-1}(\tau)&\ce&\f{1}{t_{\rm max}}\int_{z t_{\rm max}}^{(z+1)t_{\rm max}} d\phi\,\ket{\phi}\otimes U_S(\phi-\tau)\,,\nn
\ea
i.e.\ with a shift in integration range.~In this case one would have the identity $\calr_{\mathbf S}^{-1}(\tau) = e^{-iz\varphi}\,\calr_{\mathbf S}^{-1}(\tau_C)$ also outside of $\ch_S^{\rm phys}$.~While this yields identical results for mapping states from $\ch_S^{\rm phys}$ back into the physical Hilbert space, it does lead to differences for the observable embedding map below.~In that case, one would obtain a periodic embedding map and could thus only reconstruct $\hat F_{f_S,T}(\tau_C)$, i.e.\ for clock readings $\tau_C\in[0,t_{\rm max})$.}
\ba
\calr_{\mathbf S}^{-1}(\tau) = e^{-iz\varphi}\,\calr_{\mathbf S}^{-1}(\tau_C),\q\q\text{on }\ch_S^{\rm phys}.\label{id4}
\ea
For later purpose note that this can also be written as
\ba
\calr_{\mathbf S}^{-1}(\tau)=U_{CS}^\dag(\tau)\f{1}{t_{\rm max}}\int_0^{t_{\rm max}}d\phi\,U_{CS}(\phi)\left(\ket{\tau}\otimes I_S\right),\nn\\\label{PWinvred}
\ea
which, unless the constraint $\hat C_H$ generates a $t_{\rm max}$-periodic projective unitary representation of $\rm{U}(1)$ too, amounts to a \emph{partial coherent group average} of the operator $\ket{\tau}\otimes I_S$ over the group generated by the constraint. In analogy to the quantisation of the relational observables in Eq.~\eqref{qrelobs}, this is the coherent average over the group generated by the clock Hamiltonian, namely one cycle of the clock. In particular, unless $\hat C_H$ generates a $t_{\rm max}$-periodic projective unitary representation of $\rm{U}(1)$ too, we have
\ba
\calr_{\mathbf S}^{-1}(\tau)\neq\Pi_{\rm phys}\left(\ket{\tau}\otimes I_S\right).\label{issue}
\ea
Nevertheless, invertibility holds as follows:
\begin{Lemma}\label{lem_5}
The reduction maps satisfy for all admissible unravelled clock readings $\tau$
\ba
\calr_{\mathbf S}^{-1}(\tau)\cdot\calr_{\mathbf S}(\tau)&\approx& I_{\rm phys}\,,\nn\\
\calr_{\mathbf S}(\tau)\cdot\calr_{\mathbf S}^{-1}(\tau)&\approx_S& I_S^{\rm phys},\nn
\ea
where $I_{\rm phys}$ and $I_S^{\rm phys}$ are the identities on $\ch_{\rm phys}$ and $\ch_S^{\rm phys}$, respectively.
\end{Lemma}\begin{proof}\proofapp\end{proof}
We emphasise that the reduction map in Eq.~\eqref{PWred}, which amounts to conditioning states on a particular clock reading, is \emph{only} invertible when acting on the physical Hilbert space $\ch_{\rm phys}$.~This is a consequence of the constraint induced gauge redundancy in the description of $\ch_{\rm phys}$ in terms of kinematical degrees of freedom; the conditioning, as a (partial) gauge fixing, only removes redundant information in the clock factor.~It is clear that this conditioning is \emph{not} invertible for kinematical states in $\ch_{\rm kin}$ where no such redundancy arises.

Hence, we have an (up to phase) periodic relational dynamics of the system $S$ relative to the clock $C$ obeying the Schr\"odinger equation and this reduced relational dynamics is consistent with gauge-invariance of the state:~using the inverse reduction map $\calr_{\mathbf S}^{-1}(\tau)$ at clock reading $\tau$ we can also reconstruct the gauge-invariant physical state $\ket{\psi_{\rm phys}}$ from $\ket{\psi_S(\tau)}$.

\subsubsection{Reduction and embedding of observables}

Let us now consider observables.~As in the case of monotonic clocks \cite{Hoehn:2019owq,HLSrelativistic}, we define an embedding map of observables on $\ch_S^{\rm phys}$ into observables on the physical Hilbert space $\ch_{\rm phys}$ by using the reduction map and its inverse
\ba
\mathcal{E}_{\mathbf S}^\tau\left(\hat f_S^{\rm phys}\right)&\ce&\calr_{\mathbf S}^{-1}(\tau)\,\hat f_S^{\rm phys}\,\calr_{\mathbf S}(\tau).\label{embed}
\ea
Thanks to Eq.~\eqref{id4}, this embedding is weakly periodic
\ba
\mathcal{E}_{\mathbf S}^\tau\left(\hat f_S^{\rm phys}\right)\approx\mathcal{E}_{\mathbf S}^{\tau_C}\left(\hat f_S^{\rm phys}\right)\,,\label{id5}
\ea
in line with the induced periodicity of the physical system observables in Eq.~\eqref{periodicc}.~The following result tells us that the embedded physical system observables coincide weakly with the quantisation of the relational observables in Eq.~\eqref{qrelobs} associated with them -- and these are quantum Dirac observables.
\begin{theorem}\label{thm_relobs}
Let $\hat f_S^{\rm phys}\in\cl(\ch_S^{\rm phys})$ be a physical system observable. Its embedding coincides weakly with the quantisation of the relational observables in Eq.~\eqref{qrelobs},
\ba
\mathcal{E}_{\mathbf S}^\tau\left(\hat f_S^{\rm phys}\right)\approx\hat F_{f_S^{\rm phys},T}(\tau),
\ea
which in this case \emph{are} weak quantum Dirac observables, i.e.\ $[\hat F_{f_S^{\rm phys},T}(\tau),\hat C_H]\approx0$.

Conversely, the reduction of a relational observable associated with a physical system observable $\hat f_S^{\rm phys}$ coincides with that observable on the physical system Hilbert space $\ch_S^{\rm phys}$,\footnote{Note the difference of this relation compared to the case of monotonic clocks where \cite{Hoehn:2019owq,HLSrelativistic}
\ba
\calr_{\mathbf S}(\tau)\,\hat F_{f_S,T}(\tau)\,\calr_{\mathbf S}^{-1}(\tau) = \Pi_{\sigma_{S\vert C}}\,\hat f_S\,\Pi_{\sigma_{S\vert C}}\nn
\ea
for arbitrary $\hat f_S\in\cl(\ch_S)$.}
\ba
\calr_{\mathbf S}(\tau)\,\hat F_{f_S^{\rm phys},T}(\tau)\,\calr_{\mathbf S}^{-1}(\tau)\approx_S \hat f_S^{\rm phys}.
\ea
\end{theorem}
\begin{proof}
\proofapp
\end{proof}

In other words, the embedding of physical system observables is consistent with the quantisation of the relational observables.~Note that Theorem~\ref{thm_relobs} and Eq.~\eqref{id5} imply that the quantisation of relational observables associated with physical system observables is weakly $t_{\rm max}$-periodic too
\ba
\hat F_{f_S^{\rm phys},T}(\tau)\approx\hat F_{f_S^{\rm phys},T}(\tau_C).\nn
\ea

\subsubsection{Preservation of expectation values and inner products}

As in the case of monotonic clocks  \cite{Hoehn:2019owq,HLSrelativistic}, the expectation values of physical observables are preserved by the reduction and embedding maps:
\begin{theorem}\label{thm_expobs}
Let $\hat f_S^{\rm phys}\in\cl(\ch_S^{\rm phys})$ be a physical system observable. The expectation value of the corresponding relational observable evaluated in the physical inner product on $\ch_{\rm phys}$, given in Eq.~\eqref{PIP}, coincides with the expectation value of $\hat{f}_S^{\rm phys}$ evaluated in the inner product on $\ch_S^{\rm phys}$, given in Eq.~\eqref{sysPIP}, i.e.
\ba
&&\braket{\phi_{\rm phys}|\,\hat F_{f_S^{\rm phys},T}(\tau)\,|\psi_{\rm phys}}_{\rm phys}\nn\\
&&\q\q\q\q\q\q=\braket{\phi_S^{\rm phys}(\tau)|\,\hat f_S^{\rm phys}\,|\psi_S^{\rm phys}(\tau)}_{\ch_S^{\rm phys}}\nn\\
&&\q\q\q\q\q\q=\braket{\phi_S(\tau)|\,\hat f_S^{\rm phys}\,|\psi_S^{\rm phys}(\tau)}_S,\nn
\ea
where 
\begin{itemize}
\item[(i)] physical states and physical system states are related by Page-Wootters reduction, ${\ket{\psi_S^{\rm phys}(\tau)}:=\calr_{\mathbf S}(\tau)\,\ket{\psi_{\rm phys}}}$ and similarly for $\ket{\phi_S^{\rm phys}(\tau)}$, and 
\item[(ii)] $\ket{\phi_S(\tau)}:=U_S(\tau)\ket{\phi_S}$ is any kinematical system state $\ket{\phi_S}\in\ch_S$ such that $\Pi_{\sigma_{S\vert C}}\,\ket{\phi_S(\tau)} = \calr_{\mathbf S}(\tau)\ket{\phi_{\rm phys}}=\ket{\phi_S^{\rm phys}(\tau)}\in\ch_S^{\rm phys}$.
\end{itemize}
\end{theorem}
\begin{proof}
\proofapp
\end{proof}

This has a useful corollary, showing the equivalence of the inner products on $\ch_{\rm phys}$ and $\ch_S^{\rm phys}$:

\begin{corol} \label{corolInProdsPW}
Setting $\hat f_S^{\rm phys}=I_S$ in Theorem~\ref{thm_expobs}, we find that the Page-Wootters reduction map $\calr_{\mathbf S}(\tau):\ch_{\rm phys}\rightarrow\ch_S^{\rm phys}$ preserves the physical inner product, i.e. 
\ba
\braket{\phi_{\rm phys}|\psi_{\rm phys}}_{\rm phys}=\braket{\phi_S^{\rm phys}(\tau)|\psi_S^{\rm phys}(\tau)}_{\ch_S^{\rm phys}},\nn
\ea
where ${\ket{\psi_S^{\rm phys}(\tau)}:=\calr_{\mathbf S}(\tau)\,\ket{\psi_{\rm phys}}}$ and similarly for $\ket{\phi_S^{\rm phys}(\tau)}$.
\end{corol}

The reduction and embedding are thus formally unitary and the physical Hilbert space $\ch_{\rm phys}$ and the physical system Hilbert space $\ch_S^{\rm phys}$ are isometric under the reduction map, as are the algebras generated by observables on them. In particular, every quantum Dirac observable can be written as a relational observable corresponding to some physical system observable $\hat f_S^{\rm phys}$, and, vice versa, every physical system observable is the reduction of a relational Dirac observable.~In this sense, the equivalence between relational Dirac observables and the Page-Wootters formulation, established in \cite{Hoehn:2019owq,HLSrelativistic} for monotonic clocks (and in \cite{delaHamette:2021oex} for quantum reference frames for general symmetry groups), extends to periodic clocks as well.~As mentioned earlier, however, it may well be that the physical Hilbert spaces $\ch_{\rm phys}$ and $\ch_S^{\rm phys}$ will turn out to be `too small' in order to support non-trivial observables on them (recall the Examples~\ref{ex_4},~\ref{ex_8} of incommensurate oscillators where the physical Hilbert space and the algebra generated by Dirac observables are one-dimensional).

\subsubsection{Correct conditional probabilities for periodic clocks} \label{sCorrectCondProbPW}

Altogether, this means that for periodic clocks we should define the conditional probabilities (or rather probability \emph{densities}) of the Page-Wootters formalism in terms of expectation values of physical conditioning operators (or densities) in physical inner products:
\ba\label{eProbDens}
P(f_S|\tau)&\ce&\f{\braket{\psi_{\rm phys}|\,\hat F_{\ket{f_S^{\rm phys}}\!\bra{f_S^{\rm phys}},T}(\tau)\,|\psi_{\rm phys}}_{\rm phys}}{\braket{\psi_{\rm phys}\,|\psi_{\rm phys}}_{\rm phys}}\nn\\
&=&\f{\braket{\psi_S^{\rm phys}(\tau)|\,\ket{f_S^{\rm phys}}\!\bra{f_S^{\rm phys}}\,|\psi_S^{\rm phys}(\tau)}_{\ch_S^{\rm phys}}}{\braket{\psi_S^{\rm phys}(\tau)|\psi_S^{\rm phys}(\tau)}_{\ch_S^{\rm phys}}},\nn
\ea
where $\ket{f_S^{\rm phys}}$ is the eigenstate of $\hat f_S^{\rm phys}$ corresponding to the outcome $f_S$. In this manner, the conditional probabilities are manifestly gauge-invariant.

The conditional probabilities of the Page-Wootters formalism are often defined using the \emph{conditional} inner product $\braket{\phi_{\rm phys}|\ket{\tau}\!\bra{\tau}\otimes I_S |\psi_{\rm phys}}_{\rm kin}$, which is just the expectation value of the `projector' onto clock reading $\tau$ in physical states, but evaluated in the kinematical inner product.~This conditional inner product can be interpreted as a gauge-fixed inner product.~For monotonic clocks it is equivalent to the physical inner product Eq.~\eqref{PIP} \cite{Hoehn:2019owq,HLSrelativistic}, but for periodic clocks, a subtlety arises, as illustrated by the following lemma.

\begin{Lemma}\label{lem_6}
Let $G$ be the group generated by the constraint $\hat C_H$. If $G=\rm{U}(1)$, then the conditional inner product equals the physical inner product, i.e.
\ba
\braket{\phi_{\rm phys}|\ket{\tau}\!\bra{\tau}\otimes I_S |\psi_{\rm phys}}_{\rm kin}=\braket{\phi_{\rm phys}|\psi_{\rm phys}}_{\rm phys}.\nn
\ea
However, if $G=\mathbb{R}$, then the conditional inner product $\braket{\phi_{\rm phys}|\ket{\tau}\!\bra{\tau}\otimes I_S |\psi_{\rm phys}}_{\rm kin}$ diverges.
\end{Lemma}
\begin{proof}
\proofapp 
\end{proof}

The interpretation of Lemma~\ref{lem_6} is clear: the group average used to obtain the physical states is oblivious to the clock cycle and simply counts all cycles that fit into the group $G$, during which the periodic clock runs through the same value multiple times. When $G$ generates a compact group, this ``overcounting'' is finite, and compensated by the normalisation constant of the group average, giving the correct inner product.  When $G$ generates a non-compact group, however, this ``overcounting'' occurs infinitely many times, and cannot be accounted for by a finite normalisation constant, leading to the failure of the conditional inner product.

To illustrate the importance of choosing the correct inner product, we can contrast the definition of the conditional probability densities $P(f_S|\tau)$ above with the following na\"ive (but standard) definition with respect to the kinematical inner product:
\ba
    \tilde{P}(f_S|\tau) &\ce& \f{\braket{\psi_{\rm phys}|\, (\ket{\tau}\!\bra{\tau}\otimes \ket{f_S^{\rm phys}}\!\bra{f_S^{\rm phys}}) \,|\psi_{\rm phys}}_{\rm kin}}{\braket{\psi_{\rm phys}|\, (\ket{\tau}\!\bra{\tau}\otimes I_{S}) \,|\psi_{\rm phys}}_{\rm kin}} . \nn
\ea
As we show in Appendix~\ref{sFailureCondProbPW}, choosing a physical state $\ket{\phi_{\rm phys}}$ such that $\ket{f_S^{\rm phys}}=\calr_{\mathbf S}(\tau)\ket{\phi_{\rm phys}}$, this na\"ive conditional probability can be written
\ba 
    \tilde{P}(f_S|\tau) &\ce& \f{|\braket{\phi_{\rm phys}|\, (\ket{\tau}\!\bra{\tau}\otimes I_{S}) \,|\psi_{\rm phys}}_{\rm kin} |^{2}
    }{\braket{\psi_{\rm phys}|\, (\ket{\tau}\!\bra{\tau}\otimes I_{S}) \,|\psi_{\rm phys}}_{\rm kin}} .
\ea
Applying Lemma~\ref{lem_6} to both the numerator and denominator, one can then show that in the case where $G=\rm{U}(1)$, $\tilde{P}(f_S|\tau)$ coincides with $P(f_S|\tau)$, but when $G=\mathbb{R}$ is the translation group, $\tilde{P}(f_S|\tau)$ diverges (see Appendix~\ref{sFailureCondProbPW}). In other words, an incorrect choice of inner product in the definition of the conditional probability densities will lead them to diverge in the case of a perodic clock describing an aperiodic system.

\subsection{The relational Heisenberg picture (quantum deparametrisation)}

In Sec.~\ref{sec_redps} we discussed the construction of classical reduced phase spaces through gauge-fixing which encounters global challenges when the clock is periodic (see also \cite{Vanrietvelde:2018pgb,Vanrietvelde:2018dit,hoehnHowSwitchRelational2018,Hoehn:2018whn,Hoehn:2019owq,HLSrelativistic}).~This symmetry reduction leads to a deparametrisation of the relational dynamics.~We shall now extend the quantum version of this deparametrisation from monotonic clocks \cite{hoehnHowSwitchRelational2018,Hoehn:2018whn,Hoehn:2019owq,HLSrelativistic} to periodic clocks, resulting in a relational Heisenberg picture. 

As in \cite{Hoehn:2019owq,HLSrelativistic}, this reduced quantum theory will be unitarily equivalent to the relational Schr\"odinger picture of the Page-Wootters formalism of the previous subsection and can be interpreted as the description of the dynamics of the system $S$ relative to the temporal reference frame defined by the periodic clock $C$.~Quantum deparametrisation consists of two steps: (i) transform the constraint in such a way that it only acts on the chosen reference system (here the periodic clock $C$), fixing its now redundant degrees of freedom, while retaining the evolving system degrees of freedom as the unconstrained and independent ones (\emph{constraint trivialisation}); (ii) condition on a classical gauge-fixing condition of the unraveled clock $T=\tau$ to remove the redundant clock degrees of freedom.~This symmetry induced redundancy only arises on solutions to the constraint and as such the quantum deparametrisation procedure will only be invertible on solutions to the constraint.

\subsubsection{Constraint trivialisation: kinematically disentangling the clock}

The constraint trivialisation map is defined by \cite{Vanrietvelde:2018pgb,Vanrietvelde:2018dit,hoehnHowSwitchRelational2018,Hoehn:2018whn,Hoehn:2019owq,HLSrelativistic}
\ba
\ct_{C}&\ce&\sum_{n=0}^\infty\,\f{i^n}{n!}\,\hat\phi^{(n)}\otimes\left(\hat H_S+\varepsilon_*I_S\right)^n\nn\\
&=&\f{1}{t_{\rm max}}\int_0^{t_{\rm max}} d\phi\,\ket{\phi}\!\bra{\phi}\otimes e^{i\phi(\hat H_S+\varepsilon_*I_S)},\label{trivialisation}
\ea
where
$
\varepsilon_*
$
is any energy eigenvalue of the periodic clock such that $-\varepsilon_*\in\sigma_{S\vert C}$ and $\hat\phi^{(n)}$ is the $n^{\rm th}$-moment operator of the covariant and periodic clock POVM $E_{\phi_C}$, see Sec.~\ref{sec_covpovm}.~The constraint trivialisation defines an isometry $\ct_{C}:\ch_{\rm phys}\rightarrow\ct_C(\ch_{\rm phys})$ from the physical Hilbert space into a new, `trivialised' physical Hilbert space $\ct_C(\ch_{\rm phys})$ in which only the clock, i.e.\ reference degrees of freedom are constrained.~Including $\varepsilon_*$ in the map will be necessary in order to render it invertible on solutions to the constraint.
\begin{Lemma}\label{lem_triv1}
On solutions to the constraint in Eq.~\eqref{WDW}, the inverse $\ct_C(\ch_{\rm phys})\rightarrow\ch_{\rm phys}$ of the constraint trivialisation map is given by
\ba
\ct_{C}^{(-1)}=\f{1}{t_{\rm max}}\int_0^{t_{\rm max}} d\phi\,\ket{\phi}\!\bra{\phi}\otimes e^{-i\phi(\hat H_S+\varepsilon_*I_S)},\label{trivinv}
\ea
so that
\ba
\ct_{C}^{(-1)}\cdot\ct_{C}\approx I_{\rm phys}.\nn
\ea
\end{Lemma}
\begin{proof}
\proofapp
\end{proof}
We emphasise that this inverse relation only holds for physical states.~This has an immediate consequence for the inverse direction:
\begin{corol}\label{cor_triv1}
Let $I_{\ct_C(\ch_{\rm phys})}$ and $\overset{*}{\approx}$ be the identity and weak equality on the trivialised physical Hilbert space $\ct_C(\ch_{\rm phys})$, respectively. Then
\ba
\ct_C\cdot\ct_C^{(-1)}\overset{*}{\approx}I_{\ct_C(\ch_{\rm phys})}.\nn
\ea
\end{corol}

We are now ready to understand the key property of the trivialisation map: it transforms the constraint and physical state in such a way that only the clock degrees of freedom are constrained.

\begin{Lemma}\label{lem_8}
The map $\ct_C$ (weakly) trivialises the constraint in Eq.~\eqref{WDW} to the clock degrees of freedom on $\ct_C(\ch_{\rm phys})$, i.e.
\ba
\ct_C\,\hat C_H\,\ct_C^{(-1)}\overset{*}{\approx}\left(\hat H_C-\varepsilon_*\right)\otimes I_S.\nn
\ea
Furthermore, it transforms physical states in Eq.~\eqref{crap} into a product form (relative to the tensor factorization of $\ch_{\rm kin}$):
\ba
&&\ct_C\,\ket{\psi_{\rm phys}} =e^{ig(\varepsilon_*)} \ket{\varepsilon_*}_C\otimes\ket{\psi_S^{\rm phys}}\nn
\ea
with $\ket{\psi_S^{\rm phys}}\in\ch_S^{\rm phys}$ given by Eqs.~\eqref{physsysstate} and~\eqref{Heisstate}.
\end{Lemma}
\begin{proof}
\proofapp
\end{proof}
With respect to the kinematical tensor product decomposition, we may thus also call the trivialisation map a disentangling operation.~For a deeper discussion of this kinematical disentangling and the role of entanglement in relational models, see \cite{Hoehn:2019owq,HLSrelativistic,Hoehn:2021wet,Hoehn:2023ehz}.

\subsubsection{State reductions and embeddings}

Upon kinematically disentangling the clock from the evolving system and fixing the former to one of its energy eigenstates, the clock tensor factor has become entirely redundant and no longer carries any information about the original physical state.~We are thus free to remove it without losing information, just as in the Page-Wootters case.~Noting that $\braket{\tau|\varepsilon_*}_C=e^{-i(g(\varepsilon_*)-\varepsilon_*\tau)}$, we define the quantum deparametrisation map to the system physical Hilbert space, $\calr_{\mathbf H}:\ch_{\rm phys}\rightarrow\ch_{S}^{\rm phys}$, by
\ba
\calr_{\mathbf H}\ce \left(e^{-i\varepsilon_*\tau}{}_C\bra{\tau}\otimes I_S\right)\ct_C. \label{Hredmap}
\ea
Consistent with a relational Heisenberg picture, the image of this deparametrisation map is independent of $\tau$, i.e.
\ba
\calr_{\mathbf H}\,\ket{\psi_{\rm phys}} = \ket{\psi^{\rm phys}_S}.\label{Hredstate} 
\ea
In other words, $\calr_{\mathbf H}$ is independent of the clock reading $\tau$ on its entire domain of definition, $\ch_{\rm phys}$; for this reason we do not write $\calr_{\mathbf H}$ as a function of $\tau$. This contrasts with the up-to-a-phase periodicity of the Page-Wootters reduction map in Eq.~\eqref{redperiod}.
As such, the inverse reduction (or parametrisation) map, $\ch_S^{\rm phys}\rightarrow\ch_{\rm phys}$, is independent of $\tau$:
\ba
\calr_{\mathbf H}^{-1}\ce e^{ig(\varepsilon_*)}\,\ct_C^{(-1)}\left(\ket{\varepsilon_*}_C\otimes I_S\right).\label{QRinv}
\ea
The following lemma, extending results in \cite{Hoehn:2019owq,HLSrelativistic} to the $\rm{U}(1)$ case, establishes invertibility for physical states and shows that the Page-Wootters reduction and quantum deparametrisation maps are (weakly) unitarily equivalent.
\begin{Lemma}\label{lem_10}
The quantum deparametrisation map weakly equals the Page-Wootters reduction map and a system time evolution,
\ba
\calr_{\mathbf H}\approx U_S^\dag(\tau)\cdot \calr_{\mathbf S}(\tau),
\ea
while their inverses satisfy the strong relation for all $\tau$
\ba
\calr^{-1}_{\mathbf H}=\calr_{\mathbf S}^{-1}(\tau)\cdot U_S(\tau).\label{qrinv2}
\ea
In particular,
\ba
\calr_{\mathbf H}^{-1}\cdot\calr_{\mathbf H}&\approx& I_{\rm phys},\nn\\
\calr_{\mathbf H}\cdot\calr_{\mathbf H}^{-1}&\approx_S& I^{\rm phys}_S.\nn
\ea
\end{Lemma}
\begin{proof}
\proofapp
\end{proof}

\subsubsection{Reduction and embedding of observables}

As in the Page-Wootters case, we can exploit the quantum deparametrisation to embed evolving Heisenberg observables (recall their periodicity, Eq.~\eqref{periodicc})
\ba
\hat f_S^{\rm phys}(\tau)=U^\dag_S(\tau)\,\hat f_S^{\rm phys}\,U_S(\tau)\nn
\ea 
acting on the physical system Hilbert space $\ch_S^{\rm phys}$ into the algebra of Dirac observables $\cl(\ch_{\rm phys})$ via
\ba
\mathcal{E}_{\mathbf H}\left(\hat f_S^{\rm phys}(\tau)\right)\ce\calr_{\mathbf H}^{-1}\,\hat f_S^{\rm phys}(\tau)\,\calr_{\mathbf H},
\ea
Recall that, despite the apparent dependence of the reduction map $\calr_{\mathbf H}$ in Eq.~\eqref{Hredmap} on a clock reading, there is in fact no such dependence across its entire domain of definition, as shown in Eq.~\eqref{Hredstate}. The definition of the encoding map $\mathcal{E}_{\mathbf H}(\cdot)$ is therefore likewise independent of clock reading; dependence on clock reading results only from its argument, as the following theorem shows.

\begin{theorem}\label{thm_relobs2}
Let $\hat f_S^{\rm phys}(\tau)\in\cl(\ch_S^{\rm phys})$ be any evolving Heisenberg observable on the physical system Hilbert space.~Its embedding coincides weakly with the quantum relational Dirac observable in Eq.~\eqref{qrelobs},
\ba
\mathcal{E}_{\mathbf H}\left(\hat f_S^{\rm phys}(\tau)\right)\approx\hat F_{f_S^{\rm phys},T}(\tau).
\ea
Conversely, the quantum deparametrisation of a quantum relational Dirac observable weakly yields the corresponding relational Heisenberg observable on the physical system Hilbert space,
\ba
\calr_{\mathbf H}\,\hat F_{f_S^{\rm phys},T}(\tau)\,\calr_{\mathbf H}^{-1}\approx_S\hat f_S^{\rm phys}(\tau).
\ea
\end{theorem}
\begin{proof}
\proofapp
\end{proof}

\subsubsection{Preservation of physical expectation values and inner product}

The expectation values of physical observables are preserved by the embedding and deparametrisation maps.

\begin{theorem}\label{thm_expobs2}
Let $\hat f_S^{\rm phys}(\tau)\in\cl(\ch_S^{\rm phys})$ be any evolving Heisenberg observable.~Its expectation value in the inner product on $\ch_S^{\rm phys}$, given in Eq.~\eqref{sysPIP}, coincides with the expectation value of the corresponding relational observable evaluated in the physical inner product on $\ch_{\rm phys}$, given in Eq.~\eqref{PIP}, i.e.\
\ba
\braket{\phi_{\rm phys}|\,\hat F_{f_S^{\rm phys},T}(\tau)\,|\psi_{\rm phys}}_{\rm phys}&=&\braket{\phi_S^{\rm phys}|\,\hat f_S^{\rm phys}(\tau)\,|\psi_S^{\rm phys}}_{\ch_S^{\rm phys}}\nn\\
&=&\braket{\phi_S|\,\hat f_S^{\rm phys}(\tau)\,|\psi_S^{\rm phys}}_S,\nn
\ea
where 
\begin{itemize}
\item[(i)] physical states and physical system states are related by the deparametrisation, ${\ket{\psi_S^{\rm phys}}:=\calr_{\mathbf H}\,\ket{\psi_{\rm phys}}}$ and similarly for $\ket{\phi_S^{\rm phys}}$, and 
\item[(ii)] $\ket{\phi_S}\in\ch_S$ is any kinematical system state such that $\Pi_{\sigma_{S\vert C}}\,\ket{\phi_S} = \calr_{\mathbf H}(\tau)\ket{\phi_{\rm phys}}=\ket{\phi_S^{\rm phys}}\in\ch_S^{\rm phys}$.
\end{itemize}
\end{theorem}
The proof of Theorem~\ref{thm_expobs2} is analogous to that of Theorem~\ref{thm_expobs} upon invoking Lemma~\ref{lem_10} and Theorem~\ref{thm_relobs2}, and is thus omitted (see also \cite{hoehnHowSwitchRelational2018,Hoehn:2018whn,Hoehn:2019owq,HLSrelativistic} for the case of monotonic clocks). 

Theorem~\ref{thm_expobs2} entails that quantum deparametrisation defines an isometry.
\begin{corol}
Setting $\hat f_S^{\rm phys}=I_S$ in Theorem~\ref{thm_expobs2}, we find that quantum deparametrisation $\calr_{\mathbf H}:\ch_{\rm phys}\rightarrow\ch_S^{\rm phys}$ preserves the inner product, i.e. 
\ba
\braket{\phi_{\rm phys}|\psi_{\rm phys}}_{\rm phys}=\braket{\phi_S^{\rm phys}|\psi_S^{\rm phys}}_{\ch_S^{\rm phys}},\nn
\ea
where ${\ket{\psi_S^{\rm phys}}:=\calr_{\mathbf H}\,\ket{\psi_{\rm phys}}}$ and similarly for $\ket{\phi_S^{\rm phys}}$.
\end{corol}

This extends the equivalence between the relational dynamics in terms of relational Dirac observables on the physical Hilbert space and the relational Heisenberg picture of the quantum deparametrised theory from monotonic clocks (incl.\ relativistic ones)  \cite{hoehnHowSwitchRelational2018,Hoehn:2018whn,Hoehn:2019owq,HLSrelativistic} to periodic clocks.~We recall, however, that for periodic clocks, the physical Hilbert space may be `too small' to support non-trivial periodic S-observables yielding global Dirac observables as it was e.g.~the case for incommensurate oscillators (cfr.~Sec.~\ref{Sec:quantexamples}).~It is also clear that the relational Heisenberg picture of the quantum deparametrised theory is unitarily equivalent to the relational Schr\"odinger picture of the Page-Wootters formalism.

\section{Comparing dynamics with respect to periodic and aperiodic clocks}~\label{sec_clockchanges}

Given that a periodic clock entails a necessarily periodic relational dynamics, in contrast with an aperiodic clock, it is illustrative to consider a scenario in which we can compare dynamics with respect to each kind of clock.~To this end, we consider the case of a tripartite kinematical Hilbert space $\ch_{\rm kin}=\ch_A\otimes\ch_B\otimes\ch_S$ and a constraint operator of the form $\hat{C}_H = \hat{H}_A + \hat{H}_B + \hat{H}_S$, with $\hat{H}_i \in \cl (\ch_i)$ where $i=A,B,S$, we omit tensor factors of identity operators, and we assume for simplicity that each $\hat{H}_i$ is nondegenerate.~Furthermore, let $\hat{H}_A$ have a purely continuous spectrum, thus corresponding to an aperiodic clock, with a time parameter we denote by $\tau_{A}$ and a clock operator denoted by $\hat{Q}_A$ (which satisfies $e^{it \hat{H}_A}\hat{Q}_A e^{-it\hat{H}_A}=t\hat{I}_A+\hat{Q}_A$), and let $\hat{H}_B$ correspond to a periodic clock, whose time parameter we denote by $\phi_{B}$ and the first moment of the clock POVM (cf.~Eq.~\eqref{nthmoment}) is denoted by $\hat{\phi}_B$.

\subsection{Comparing relational observables}

Let us first consider the construction of the relational observables on $\ch_{\rm phys}$ with respect to the aperiodic clock $A$, i.e.~$\hat{F}_{f_{S},T_{A}}(\tau_{A})$ defined according to Eq.~(27) in~\cite{Hoehn:2019owq}.~Then $\hat{F}_{f_{S},T_{A}}(\tau_{A})$ is a strong Dirac observable, according to Theorem~1 in~\cite{Hoehn:2019owq}.~In this case the periodicity or aperiodicity of $\hat{F}_{f_{S},T_{A}}(\tau_{B})$ is determined by the respective periodicity or aperiodicity of $\hat{f}_{S}$.~Now, consider instead the relational observable with respect to the periodic clock $B$, i.e.~$\hat{F}_{f_{S},T_{B}}(\phi_{B})$, defined according to Eq.~\eqref{qrelobs}.~Then, according to Lemma~\ref{lem_ruin}, this is only a Dirac observable if $\hat{f}_{S}$ is weakly $t_\mathrm{max}$-periodic with respect to the unitary generated by $\hat{H}_{S}$ (where $t_\mathrm{max}$ is the period of the covariant time observable with respect to $\hat{H}_{B}$).~Thus, when this latter condition does not hold, in order to have a relational Dirac observable one must instead construct the relational observable with respect to the physical system observable associated with clock $B$, i.e.~$\hat{f}_{S}^{\rm phys}$ as defined in Eq.~\eqref{physsobs}, thus inducing $t_\mathrm{max}$-periodicity.~However, $\hat{f}_{S}^{\rm phys}$ is a distinct observable to the $\hat{f_{S}}$ with respect to which $\hat{F}_{f_{S},T_{A}}(\tau_{A})$ was constructed. Indeed, it is now generally an operator pertaining to both $S$ and $A$ because the `projector' $\Pi_{\sigma_{S|C}}$ in Eq.~\eqref{physsobs} now has to be replaced with the corresponding $\Pi_{\sigma_{AC|C}}$, which acts on both $AS$. The fact that $\hat{F}_{f_{S}^{\rm phys},T_{B}}(\phi_{B})$ and $\hat{F}_{f_{S},T_{A}}(\tau_{A})$ differ in periodicity (or lack thereof) in that case is therefore not surprising.

\subsection{Comparing perspectives}

Let us now compare the relational description obtained by performing a quantum reduction with respect to the aperiodic clock $A$ (see Sec.~V of~\cite{Hoehn:2019owq}) with one obtained by reducing with respect to the periodic clock $B$ (see Sec.~\ref{sec_reduction} above).~We denote the reduced Hilbert space with respect to clock $A$ by $\mathcal{H}_{BS\vert A}$, and the reduced space with respect to clock $B$ by $\mathcal{H}_{AS\vert B}$.

As in~\cite{Hoehn:2019owq} we can transform from the perspective of clock $A$ to that of clock $B$ by first applying the inverse reduction map with respect to clock $A$, and then the reduction map with respect to clock $B$ (using the Schr\"{o}dinger or Heisenberg picture version of each map as appropriate).~Then, for example, the relational Schr\"{o}dinger picture state $\ket{\psi_{BS\vert A}(\tau_{A})}\in\mathcal{H}_{BS\vert A}$ with respect to $A$ corresponds to the state $\ket{\psi_{AS\vert B}(\tau_{B})}=\Lambda^{A \to B}_{\rm \bf{S}}\ket{\psi_{BS\vert A}(\tau_{A})}\in\mathcal{H}_{AS\vert B}$ in the relational Schr\"{o}dinger picture with respect to clock $B$, with the frame change map $\Lambda^{A \to B}_{\rm \bf{S}}$ defined by
\begin{equation}
\Lambda^{A \to B}_{\rm \bf{S}} \ce \! \mathcal{R}_{\rm \bf{S}}(\tau_B) \circ \mathcal{R}^{-1}_{\rm \bf{S}}(\tau_A) ,
\end{equation}
where $\tau_A$ (respectively $\tau_B$) denotes the time read by clock $A$ (respectively clock $B$).~We recall that the inverse-reduction map $\mathcal{R}^{-1}_{\rm \bf{S}}(\tau_A)$ with respect to an aperiodic clock is as in Eq.~\eqref{PWinvred1}, except that $t_{\rm max}$ is replaced by $2\pi$ and the integration takes place over the entire real line~\cite{Hoehn:2019owq}.~An observable $\hat{O}_{BS|A}^{\rm phys}\in \cl (\ch_{BS\vert A})$ in the relational Schr\"{o}dinger picture with respect to clock $A$ then transforms to the observable $\hat{O}_{AS|B}^{\rm phys} (\tau_A, \! \tau_B)\in \cl (\ch_{AS\vert B})$ given by
\begin{align} \label{eObsTransSchr}
\hat{O}_{AS|B}^{\rm phys} (\tau_A, \! \tau_B) \!&=\! \Lambda^{A \to B}_{\rm \bf{S}} \hat{O}_{BS|A}^{\rm phys} \left(\Lambda^{A \to B}_{\rm \bf{S}} \right)^\dagger \\
&= \! \mathcal{R}_{\rm \bf{S}}(\tau_B) \circ \mathcal{E}_{\rm \bf{S}}^{\tau_A} \!\left( \hat{O}_{BS|A}^{\rm phys} \right) \! \circ \mathcal{R}^{-1}_{\rm \bf{S}}(\tau_B) \nn 
\end{align}
 Note that, despite the appearance of $\tau_A$ and $\tau_B$ in $\hat{O}_{AS|B}^{\rm phys} (\tau_A, \! \tau_B)$, it is in fact a Schr\"{o}dinger-picture observable (see~\cite{Hoehn:2019owq,HLSrelativistic} for a discussion).~In particular, the dependence on $\tau_B$ is due to the fact that $\hat{O}_{BS|A}^{\rm phys}$ may encode evolving degrees of freedom of the new clock $B$.~If this is not the case, then there is no $\tau_B$-dependence, as demonstrated in the following theorem (cf.~Theorem 7 in~\cite{Hoehn:2019owq}).

\begin{theorem}\label{thm_notbdep}
Consider an operator on $BS$ from the perspective of $A$, denoted $\hat{O}_{BS|A}^{\rm phys}\in \cl (\ch_{BS\vert A})$.~From the perspective of $B$, this operator is independent of $\tau_{B}$, so that $\hat{O}_{AS|B}^{\rm phys} (\tau_A, \! \tau_B)=\hat{O}_{AS|B}^{\rm phys} (\tau_A)\in \cl (\ch_{AS\vert B})$ if and only if $\left[I_{A}\otimes\hat{O}_{BS|A},I_{A}\otimes H_{B}\otimes I_{S}\right]\approx 0$. 
\end{theorem}
\begin{proof}
\proofapp
\end{proof}

Now, in order to ask how an observable on $S$ described with respect to $A$ looks from the perspective of $B$, we must assume that the reduced Hilbert space with respect to clock $A$ factorises along the same lines as the kinematical space, i.e.~$\mathcal{H}_{BS\vert A}\simeq\mathcal{H}_{B\vert A}\otimes\mathcal{H}_{S\vert A}$ for some $\mathcal{H}_{B\vert A}$ and $\mathcal{H}_{S\vert A}$, so that $S$ can meaningfully be called a subsystem in this perspective~\cite{Hoehn:2021wet,Hoehn:2023ehz}.~Let us then consider an observable in the relational Schr\"{o}dinger picture that only acts nontrivially on $\mathcal{H}_{S\vert A}$, i.e.~$\hat{O}_{BS|A}^{\rm phys}=I_{B|A}\otimes \hat{f}_{S|A}$ for some $\hat{f}_{S|A}$.~Using Eq.~\eqref{eObsTransSchr} to transform this to the relational Schr\"{o}dinger picture observable with respect to clock $B$, we obtain
\ba
&\hat{O}_{AS|B}^{\rm phys} (\tau_A )   = \mathcal{R}_{\rm \bf{S}}(0) \circ \mathcal{E}_{\rm \bf{S}}^{\tau_A} \!\left( I_{B|A}\otimes \hat{f}_{S|A} \right) \! \circ \mathcal{R}^{-1}_{\rm \bf{S}}(0) \nonumber
\ea
which depends on $\tau_A$, but is $\tau_{B}$-independent, in accordance with Theorem~\ref{thm_notbdep}.~The $\tau_{B}$-dependence of the statistics associated with the observable $\hat{O}_{AS|B}^{\rm phys} (\tau_A )$ is then entirely encoded in the state $\ket{\psi_{AS|B}(\tau_{B})}\ce\calr_{\mathbf S}(\tau_{B})\ket{\psi_{\rm phys}}$, and this state is $t_\mathrm{max}$-periodic according to Eq~\eqref{ePeriodicPWstate}.

This clarifies how an observable can have an aperiodic behaviour with respect to one clock, and periodic behaviour with respect to another; from the perspective of clock~$B$ all probabilities are periodic with respect to the clock reading $\tau_{B}$, and the (aperiodic) label $\tau_{A}$ simply selects the observable from the set $\lbrace \hat{O}_{AS|B}^{\rm phys} (\tau_A ) \rbrace_{\tau_{A}}$ whose evolution is to be considered.

\subsection{Unravelling a periodic quantum clock with respect to an aperiodic one}\label{sUnravQuant}

This setting also allows us to show that, just as in the classical case (cf.~Sec.~\ref{sUnravClass}), one cannot use a monotonic observable for a periodic quantum clock to construct a relational Dirac observable.~Let us define $U_{C_H}(t):=\exp(-it\hat{C}_H)$ for the tripartite system.~In this context, we can establish the following Lemma, which is a quantum version of Lemma \ref{lem_monoPNP}.
\begin{Lemma}\label{lem_monoPNPQ}
The quantisation of the relational observable $F_{T_B,Q_A}(\tau)$ of Lemma~\ref{lem_monoPNP}, which encodes the value of the monotonic clock $T_B$ of a periodic system relative to the value $\tau$ of the clock $Q_A$ of an aperiodic system is given by
\begin{equation}
\hat{F}_{T_B,Q_A}(\tau) = \tau\hat{I}-\hat{Q}_A+\hat{\phi}_{B} \ ,
\end{equation}
and it satisfies the property:
\begin{equation}\label{monoPNPQ-0}
\begin{aligned}
&\braket{\alpha^s_{C_H}\cdot\hat{F}_{T_B,Q_A}(\tau)}_{21}\\
&= \braket{\hat{F}_{T_B,Q_A}(\tau)- t_{{\rm max};B}\hat{Z}_B(s)}_{21}\,,
\end{aligned}
\end{equation}
where $\alpha^s_{C_H}\cdot\hat{F}_{T_B,Q_A}(\tau) := U_{C_H}^{\dagger}(s)\hat{F}_{T_B,Q_A}(\tau)U_{C_H}(s)$, and
\begin{equation}
\hat{Z}_B(s) := \frac{1}{ t_{{\rm max};B}}\int_0^{ t_{{\rm max};B}}d\phi_B\, \left\lfloor\frac{s+\phi_B}{ t_{{\rm max};B}}\right\rfloor\ket{\phi_B}\!\bra{\phi_B} \,,
\end{equation}
and $\braket{\hat{O}}_{21} = \braket{\psi_2|\hat{O}|\psi_1}$ for any operator $\hat{O}$, with $\ket{\psi_{1,2}}$ leading to wave functions that are integrable in $\phi$ (cf. Lemmas~\ref{lem_quantumPhiEvol} and~\ref{lem_clobsQ}). Thus, Eq.~\eqref{monoPNPQ-0} is a quantum version of the result of Lemma \ref{lem_monoPNP}.
\end{Lemma}
\begin{proof}
\proofapp
\end{proof}
This Lemma can be straightforwardly generalised for systems with an arbitrary number of periodic and aperiodic clocks.~The example below illustrates the property proved in Lemma \ref{lem_monoPNPQ}, and it serves as a quantum version of Example \ref{ex:monoPNP}.
\begin{example}[\bf Periodic and aperiodic clocks]\label{ex:monoPNP-Q}
 Let us consider the tripartite system for which $\mathcal{H}_A$ is the Hilbert space of an ideal clock, whereas $\mathcal{H}_B$ is that of a periodic clock, which we take to be either infinite-dimensional as a harmonic oscillator or a two-level system, and $\mathcal{H}_S$ that of a free particle.~The constraint operator reads $\hat{C}_H = \hat{H}_A+\hat{H}_B+\hat{H}_S$, where $\hat{H}_A = \hat{H}_I$, $\hat{H}_B = \hat{H}_O$ or $\hat{H}_{\rm qubit}$, and $\hat{H}_S = \hat{H}_P$ are the operators corresponding to the Hamiltonians given in Eqs.~\eqref{monoPNP-HI},~\eqref{monoPNP-HO} or \eqref{eq:qubitHSandfreeHS},~and~\eqref{monoPNP-HP}.~The clock operator associated with $\hat{H}_I$ can be straightforwardly defined to be $\hat{T}_I:=-\hat{q}_3$, so that $[\hat{T}_I,\hat{H}_I] = \mathrm{i}$, similarly to the classical definition given in Eq.~\eqref{TIclock}.~Furthermore, the clock states of $\hat{H}_P$ are (see, for instance, \cite{braunsteinGeneralizedUncertaintyRelations1996,HLSrelativistic,Chataignier:2019kof,Chataignier:2020fys,ChataignierT})
\begin{equation}\label{FreeParticleClock}
\ket{t_P;\sigma} = \int_{-\infty}^{\infty}d p_2\ \sqrt{\left|\frac{p_2}{2\pi \hbar m_2}\right|}\theta(-\sigma p_2)\exp\left(\frac{\mathrm{i} t_P p_2^2}{2m\hbar}\right)\ket{p_2}\,,
\end{equation}
where $\sigma=\pm1$ and $\theta(-\sigma p_2)$ is the Heaviside step function.~Thus, the particle's clock operator is the first moment of the POVM associated with Eq.~\eqref{FreeParticleClock}:
\begin{equation}
\hat{T}_P:= \sum_{\sigma=\pm1}\int_{-\infty}^{\infty}d t_P\ t_P\ket{t_P;\sigma}\bra{t_p;\sigma} \,,
\end{equation}
and it satisfies $[\hat{T}_P,\hat{H}_P] = \mathrm{i}$, similarly to the classical definition in Eq.~\eqref{TPclock} \cite{HLSrelativistic}.~Finally, as for the periodic clock $B$, the clock states for the two-level clock are given in Eq.~\eqref{eq:qubitclockstates} and the corresponding clock operator $\hat{\phi}_{\rm qubit}$ given in Eq.~\eqref{eq:qubitclockoperator} is not conjugate to $\hat{H}_{\rm qubit}$ (cf.~Eq.~\eqref{eq:phiCHCcomm}).~Similarly, the clock states for the harmonic oscillator can be defined as in Eq.~\eqref{discreteClockState}, so that the clock operator is given by Eq.~\eqref{nthmoment} for $n=1$.~Due to Lemma~\ref{lem_povm}, it fails to be conjugate to $\hat{H}_O$.~This is similar to the classical theory, in which the evolution of $\phi_C$ given in Eq.~\eqref{HOphase2} is not differentiable for all values of $s$ (cf.~Eq.~\eqref{solution}).~The quantum relational observables encoding the value of $\hat{T}_B$ relative to $\hat{T}_P$ and $\hat{T}_I$ when these read $\tau$ are the operators (omitting tensor factors of identity operators)
\begin{align}
\hat{F}_{T_B,T_P}(\tau) &= \tau\hat{I} - \hat{T}_P +\hat{\phi}_B \,,\\
\hat{F}_{T_B,T_I}(\tau) &= \tau\hat{I}+\hat{q}_3+\hat{\phi}_B \,,
\end{align}
with $B=O,\rm qubit$.~In the case of the harmonic oscillator, these are the analogues of the classical observables given in Eqs.~\eqref{FTOTP} and~\eqref{FTOTI}, while there is no classical analogue for the qubit case.~In either case, the above operators fail to commute with $\hat{C}_H$ because $\hat{\phi}_B$ is not conjugate to $\hat{H}_B$.~Instead, owing to Lemmas~\ref{lem_quantumPhiEvol} and~\ref{lem_monoPNPQ}, they obey the property 
\begin{equation}\notag
\braket{\alpha^s_{C_H}\cdot\hat{F}_{T_B,T_P}(\tau)}_{21} = \braket{\hat{F}_{T_O,T_P}(\tau)- t_{{\rm max};B}\hat{Z}_B(s)}_{21} \,,
\end{equation}
and similarly for $\hat{F}_{T_B,T_I}(\tau)$, for any pair of states $\ket{\psi_{1,2}}$ that leads to wave functions that are integrable in $\phi$.~When $B$ is a harmonic oscillator, this property is the analogue of the classical result given in Eq.~\eqref{FTOtransient}.
\end{example}

\section{Discussion}\label{Sec:discussion}

Given their operational relevance, periodic clocks have, of course, been discussed in the literature on relational dynamics before \cite{Smith:2017pwx,Alonso-Serrano:2023gir,Wendel:2020hqv,Martinez:2023fsd,Bojowald:2021uqo,Dittrich:2015vfa,Dittrich:2016hvj,Bojowald:2010qw,Cafasso:2024zqa,Dittrich:2007th,pageEvolutionEvolutionDynamics1983,Rovelli:1990jm,Rovelli:1989jn,Kiefer:1989va,Kiefer:1990ms,Kiefer:1993cqg,Favalli:2020gmx}, however, mostly using specific examples. A systematic treatment, on the other hand, paralleling the classical and quantum theories, describing the relation between different approaches in this context, as well as contrasting with the treatment of monotonic clocks had been missing.~Our work aims to fill this gap. 

Rather than providing an exhaustive comparison with earlier work (which would be difficult due to the spectrum of different clock variables in the literature), we briefly compare with and comment on only two interesting recent directions.

Let us start with \cite{Wendel:2020hqv,Bojowald:2021uqo,Martinez:2023fsd}, which explores the possibility of a fundamental period of time, but otherwise shares a similar aim as ours: making sense of relational evolution with respect to a periodic reference degree of freedom.~The term ``periodic'' in \cite{Wendel:2020hqv,Bojowald:2021uqo,Martinez:2023fsd} generically refers to the confinement of the clock degrees of freedom to a finite (energy-dependent) range which thus have turning points.~This behaviour is rooted in non-trivial self-interactions in a Hamiltonian constraint of the form (in our notation) $C_H=H_S(q_S,p_S)^2-p_C^2-W(q_C)$ with a quadratic \cite{Wendel:2020hqv,Bojowald:2021uqo} or exponential \cite{Martinez:2023fsd} clock potential $W(q_C)$.~Such a constraint is also encompassed by our setup, as there are no interactions between $C$ and $S$.~However, the key difference between \cite{Wendel:2020hqv,Bojowald:2021uqo,Martinez:2023fsd} and the present manuscript is that the former use $q_C$ as the initial clock variable, whereas our formalism, applied to this case, would first construct a clock using a covariant POVM for the entire Hamiltonian $H_C=-p_C^2-W(q_C)$, which amounts to using the angle variable $\phi_C$ as a time keeper. That is, our time variable is conjugate to the constraint (on a dense subspace) and monotonic for each clock cycle, whereas this is not the case in \cite{Wendel:2020hqv,Bojowald:2021uqo,Martinez:2023fsd}, which features turning points.~To avoid energy-dependent turning points of $q_C$, the latter works also introduce winding numbers, as we do for the angle variable to accommodate different cycles. 

The different time variables lead to a further drastic difference: in \cite{Wendel:2020hqv,Bojowald:2021uqo,Martinez:2023fsd}, due to the self-interaction $W(q_C)$,
the $S$-evolution relative to $q_C$ (as well as its unwound version) is governed by a time-dependent Hamiltonian, which is obtained by solving the constraint for $p_C$. The evolving states are then constructed by concatenating the different branches of the wave functions.~By contrast, in our case, the relational $S$-dynamics obtained via Page-Wootters reduction is governed by the time-independent $\hat H_S$ (cfr.~Eq.~\eqref{schrod}) and the relational Schr\"odinger state evolution is periodic up to phase (cfr.~Eq.~\eqref{ePeriodicPWstate}), so that no concatenation is needed. This means in particular, that the contemplations about a fundamental period of time in \cite{Wendel:2020hqv,Bojowald:2021uqo} apply to a distinct time evolution and cannot be easily translated into our formulation. Indeed, the proposal in \cite{Wendel:2020hqv,Bojowald:2021uqo} relies on the evolution in their unwound clock not being exactly periodic (thanks to the time dependence in the Hamiltonian), whereas the gauge-invariant relational evolution in our variable, for the same system, is exactly periodic by fiat.

Given the finite dimensionality of the Hilbert space of a periodic quantum clock in some cases (such as a qubit or spin), one may consider a different path to the one taken here, constructing a relational quantum theory using only a discrete set of time states~\cite{baumann2022noncausal,diaz2023parallel,Favalli:2020gmx}.~This leads to a discrete form of the history state given in Eq.~\eqref{eHistory}, and can be understood in terms of the Feynman-Kitaev clock construction, where this state represents some computation~\cite{Feynman1985,Kitaev2002,Breuckmann2014,Caha2018}.~In this construction, one begins by assuming the desired form of this state, in contrast with our work which begins with the constraint given by Eq.~\eqref{WDW} and examines the relational dynamics (including the history state) that arises therefrom.

\section{Conclusion}\label{Sec:conclusion}

In summary, we have presented a systematic framework for exploring relational dynamics with periodic clocks, emphasising the parallels between classical and quantum observables relative to periodic clocks.

A key ingredient in our analysis was the resort to angle variables to model the periodic clock, as opposed to its often used configuration variable which may experience turning points. In the quantum theory, this results in $\rm{U}(1)$-covariant POVM observables, which enjoy many useful properties to model a wide variety of periodic clocks. Both classically and quantumly, the angle variable is conjugate to the clock Hamiltonian, except where the clock completes its cycles, and this simplifies many steps in the construction.

We have demonstrated that classical and quantum relational observables relative to a periodic clock are not global Dirac observables, but only transiently invariant per clock cycle, unless the quantity they encode is periodic along the entire flow generated by the constraint.~By considering a partial group-averaging procedure for quantum relational observables, we have established the equivalence of the clock-neutral picture of Dirac quantisation with the quantum deparameterisation and Page-Wootters approaches in this context, and thus extended the `trinity' of relational quantum dynamics to include periodic clocks.~This means in particular that the relational dynamics in all three formulations is necessarily periodic.~We also showed that a full average over the gauge flow generated by the constraint does not improve the situation:~when the full group is the translation group, this leads to divergences where the partial average is better behaved, and when the full group is $\rm{U}(1)$ (possibly with a distinct representation and periodicity than that of the periodic clock), then one recovers the same relational Dirac observables as with the partial average.~Specifically, the equivalence with the Page-Wootters formalism also shows that one can obtain every gauge-invariant operator via a partial group average.

We further showed that a na\"ive application of the Page-Wootters formalism with respect to a periodic clock results in divergent conditional probabilities for systems with continuous energy spectra, but we have found that this can be resolved by invoking the clock-neutral Dirac picture.~Crucially, the use of the gauge-invariant physical inner product of the clock-neutral picture, instead of the usually invoked conditional one, allowed us to correct the otherwise ill-defined conditional probabilities.

We have finally considered a scenario including both a periodic and an aperiodic clock, and shown how to resolve the apparent tension when the periodic behaviour of some system with respect to the former clock appears aperiodic with respect to the latter clock.

Despite the prevalence of periodic clocks in the literature, a systematic treatment for describing dynamics relative to them had been missing.~Our analysis shows the care that must be taken when constructing the relational theory in this context, compared to the more straightforward context of aperiodic clocks.~Our systematic treatment of this operationally important category of clocks thus represents a step forward in the program of temporal quantum reference frames.~At the same time, this class of models constitutes technically a special case, and physically many other types may be relevant beyond laboratory situations.~For example, much work remains to be done on more generic cases, such as encompassing interactions~\cite{Hohn:2011us,Smith:2017pwx,Smith:2019imm,Cafasso:2024zqa,castro-ruizTimeReferenceFrames2019,Hoehn:2023axh} and non-integrability~\cite{Dittrich:2016hvj,Hohn:2011us}, both of which appear pertinent in the context of gravity and may preclude simple periodic clock observables as explored here.

\section*{Acknowledgements}

The authors thank Alexander R.~H.~Smith for many useful discussions on the topic, as well as Martin Bojowald and Claus Kiefer for discussion and providing useful references. The work of L.C.\ was supported by the Basque Government Grant \mbox{IT1628-22}, and by the Grant PID2021-123226NB-I00 (funded by MCIN/AEI/10.13039/501100011033 and by “ERDF A way of making Europe”). It is also partly funded by the IKUR 2030 Strategy of the Basque Government.  The work of P.A.H.\ and F.M.M.\ was supported in part by funding from the Okinawa Institute of Science and Technology Graduate University.~This project/publication was also made possible through the support of the ID\# 62312 grant from the John Templeton Foundation, as part of the \href{https://www.templeton.org/grant/the-quantum-information-structure-of-spacetime-qiss-second-phase}{\textit{`The Quantum Information Structure of Spacetime'} Project (QISS)}, as well as Grant ID\# 62423 from the John Templeton Foundation.~The opinions expressed in this project/publication are those of the author(s) and do not necessarily reflect the views of the John Templeton Foundation. P.A.H.\ was further supported by the Foundational Questions Institute under grant number FQXi-RFP-1801A. M.P.E.L.\ acknowledges support from ERC-2021-COG 101043705 ``Cocoquest''. ~F.M.M.'s research at Western University is also supported by Francesca Vidotto's Canada Research Chair in the Foundation of Physics, and NSERC Discovery Grant ``Loop Quantum Gravity:~from Computation to Phenomenology''.~Western University is located in the traditional territories of Anishinaabek, Haudenosaunee,  L\=unaap\'eewak and Chonnonton Nations. 

\bibliography{periodic-trinity}

\onecolumngrid
\appendix

\section{\label{app:examples}Further examples}

\begin{example}[\bf Harmonic oscillator] \label{ex:HOfloor}
Let us illustrate the validity of Eq.~\eqref{solution} for the angle variable defined in Eq.~\eqref{HOphase2}. A classical solution for the harmonic oscillator with initial conditions $(t_0,p_{t0})$ can be written as   
\begin{equation}\label{HOsol-arc}
\begin{aligned}
t(s) &= \mathrm{sgn}(t_0)\sqrt{t_0^2+\frac{p_{t0}^2}{m_t^2\omega_t^2}}\cos\left(\omega_t (s+\varphi_0)\right)\,,\\
p_t(s)&= -\mathrm{sgn}(t_0)\sqrt{m_t^2\omega_t^2t_0^2+p_{t0}^2}\sin\left(\omega_t (s+\varphi_0)\right) \,,
\end{aligned}
\end{equation}
where $\varphi_0 = \varphi(t_0,p_{t0})$, and $\varphi$ was defined in Eq.~\eqref{HOphase}. The evolution of the angle variable $\phi_C$ given in Eq.~\eqref{HOphase2} for the harmonic oscillator is obtained by evaluating $\phi_C$ on the solution~\eqref{HOsol-arc}. If
\begin{equation}\label{HOinterval-svarphi}
(2k-1)\frac{\pi}{2\omega_t} < s+\varphi_0 < (2k+1)\frac{\pi}{2\omega_t} \ ,\ k\in\mathbb{Z} \,,
\end{equation}
then
\begin{equation}
\begin{aligned}
\varphi(t(s),p_t(s)) &= \frac{1}{\omega_t}\mathrm{arctan}\left(-\frac{p_t(s)}{m_t\omega_t t(s)}\right)\\
&=s+\varphi_0-\f{k\pi}{\omega_t} \ ,
\end{aligned}
\end{equation}
since the image of $\mathrm{arctan}$ is $(-\pi/2, \pi/2)$. Due to Eq.~\eqref{HOinterval-svarphi}, we can write $s+\varphi_0 = [2(k+\epsilon)-1]\pi/(2\omega_t)$, where $\epsilon\in(0,1)$. Together with Eq.~\eqref{HOsol-arc}, this leads to
\begin{align}
t(s)&=\mathrm{sgn}(t_0)\sqrt{t_0^2+\frac{p_{t0}^2}{m_t^2\omega_t^2}}(-1)^k\sin\left(\epsilon \pi\right)\,,\\
\mathrm{sgn}(t(s)) &= \mathrm{sgn}(t_0)(-1)^k \ ,
\end{align}
 and [recall $\phi_C(t = 0,p_t)$ is undefined]
\begin{equation}
\begin{aligned}
\phi_C(s)&\equiv\phi_C(t(s),p_t(s))\\
&= \varphi(t(s),p_t(s))+\f{\pi}{\omega_t}-\f{\pi}{2\omega_t}\mathrm{sgn}(t(s))\\
&=s+\varphi_0+\f{\pi}{\omega_t}-\frac{k\pi}{\omega_t}-\f{\pi}{2\omega_t}\mathrm{sgn}(t_0)(-1)^k \ . \label{HO-phiC-evol}
\end{aligned}
\end{equation}
For $k=2n$ ($k$ even), this reduces to
\begin{equation}
\phi_C(s) = s+\phi_C(0)-n\phi_{\rm max} \ ,
\end{equation}
where $\phi_{\rm max}=2\pi/\omega_t$, and we obtain
\begin{equation}
\begin{aligned}
\left\lfloor\frac{s+\phi_C(0)}{2\pi/\omega_t}\right\rfloor &= \left\lfloor\frac{2n+\epsilon}{2}+\frac14(1-\mathrm{sgn}(t_0))\right\rfloor = n \,,
\end{aligned}
\end{equation}
as $0<\epsilon<1$. For $k=2\tilde{n}-1$ ($k$ odd), Eq.~\eqref{HO-phiC-evol} becomes
\begin{equation}
\phi_C(s) = s +\phi_C(0)-\left(\tilde{n}-\frac12-\frac12\mathrm{sgn}(t_0)\right)\phi_{\rm max} \ ,
\end{equation}
where $n = \tilde{n}-(1+\mathrm{sgn}(t_0))/2$ is an integer that satisfies
\begin{equation}
\begin{aligned}
\left\lfloor\frac{s+\phi_C(0)}{2\pi/\omega_t}\right\rfloor &= \left\lfloor\frac{2\tilde{n}-1+\epsilon}{2}+\frac14(1-\mathrm{sgn}(t_0))\right\rfloor\\
&=\left\{
\begin{matrix}
\tilde{n}-1 = n & (\mathrm{sgn}(t_0)=1)\\
\tilde{n} = n & (\mathrm{sgn}(t_0)=-1)
\end{matrix}
\right.\ .
\end{aligned}
\end{equation}
Collecting the above results, we see that Eq.~\eqref{solution} is recovered for the angle variable of the harmonic oscillator.
\end{example}

\begin{example}[\bf Two particles on a circle]\label{ex_bad}
  To emphasise the importance of analyticity of $f_S$ in the power series expansion Eq.~\eqref{relobs}, we provide an example where the latter fails. As in \cite{Dittrich:2015vfa,Dittrich:2016hvj}, consider two free particles on a circle with fixed total energy $E>0$:
\ba 
C_H=\f{p_t^2}{2m_t}+\f{p^2}{2m}-E.\nn
\ea 
The configuration space is a torus $\mathbb T^2$ and so we have $t+1\sim t$ and $q+1\sim q$. As our unravelled clock function we choose $T(s)=s+\phi_C$, where $\phi_C$ is given in Eq.~\eqref{compactclock}. Directly solving the equations of motion for $q(s)$ and replacing $s$ by $s=\tau-\phi_C$ yields the relational observable
\ba 
F_{q,T}(\tau)=\f{p}{m}(\tau-\phi_C)+q-n_q(\tau;\phi_C,q,p),\label{Fwind}
\ea 
where 
\ba 
n_q(\tau;\phi_C,q,p)\ce\Big\lfloor \f{p}{m}s+q\Big\rfloor_{s=\tau-\phi_C}\label{n2}
\ea 
is the winding number of the second particle (i.e.\ of system $S$). By contrast, the power series in Eq.~\eqref{relobs} yields
\ba 
\tilde F_{q,T}(\tau)=\f{p}{m}(\tau-\phi_C)+q\;,
\ea 
which, however, is only correct on the $n_q=0$ cycle and beyond it takes value outside $[0,1]$, in conflict with $q\in[0,1]$. The reason for this failure is clear: $f_S=q$ is not analytic from one cycle of $S$ to the next. 

Note that $F_{q,T}(\tau)$ satisfies the transient invariance property \eqref{transientinv} and if $\f{m_tp}{mp_t}\in\mathbb N$, in which case $q(s)$ is periodic by a unit fraction of $\phi_{\rm max}=\f{m_t}{p_t}$, it is invariant along the entire gauge orbit.~It follows from the discussion in \cite{Dittrich:2015vfa} that $F_{q,T}(\tau)$ is \emph{not} continuous on $\cc$. This has to do with the fact that $\cc$ contains trajectories with $\f{m_tp}{mp_t}\notin\mathbb{Q}$ which densely fill the torus. In particular, $n_q$ will not be continuous on $\cc$ in directions transversal to a dynamical orbit and so $F_{q,T}(\tau)$ will fail to be differentiable in those directions.

This example highlights why it is important to work with analytic functions $f_S$ when using the power series expansion Eq.~\eqref{relobs}.~As emphasised in the main text, we shall use the power series to quantise the relational dynamics with respect to periodic clocks.~We thus restrict ourselves to systems $S$ which feature a Poisson subalgebra $\ca_S$ of analytic functions that also separates the points in $\cp_S$ and so can be used to coordinatise $\cp_S$.~Instead, we refer the reader to \cite{Dittrich:2015vfa,Dittrich:2016hvj} for how to quantise models such as in the example we just considered.
\end{example}

\section{\label{app:proof}Proofs of lemmas and theorems in the main text}

\noindent{\bf\hyperref[lem_clobs]{Lemma \ref{lem_clobs}}.}
For an arbitrary system phase space function $f_S:\cp_S\rightarrow\mathbb{R}$, the relational observables $F_{f_S,T}(\tau)$ satisfy the \emph{transient invariance property}
\ba
\alpha_{C_H}^{s}\cdot F_{f_S,T}(\tau)=\alpha_{C_H}^{z\phi_{\rm max}-\phi_C^0}\cdot F_{f_S,T}(\tau)\,,\label{app:transientinv}
\ea
with $z\phi_{\rm max}-\phi_C^0\leq s<(z+1)\phi_{\rm max}-\phi_C^0$, for $z\in\mathbb Z$, and
\ba 
\alpha_{C_H}^{z\phi_{\rm max}-\phi_C^0}\cdot F_{f_S,T}(\tau)=F_{f_S,T}(\tau+z\phi_{\rm max}).\q\;\,
\ea

\begin{proof}
First note that, for arbitrary phase space functions $f,g:\cp_{\rm kin}\rightarrow\mathbb{R}$ and arbitrary $x\in\cp_{\rm kin}$, we have for their point-wise product
\ba 
\alpha_{C_H}^s\left[f\cdot g\right](x)&=&\left[f\cdot g\right]\left(\alpha_{C_H}^s(x)\right)\nn\\
&=&f\left(\alpha_{C_H}^s(x)\right)\cdot g\left(\alpha_{C_H}^s(x)\right)\,.\nn
\ea 
We can thus write 
\ba 
\alpha_{C_H}^{s}\cdot F_{f_S,T}(\tau)&=&\sum_{n=0}^\infty\alpha_{C_H}^{s}\left[\f{(\tau-\phi_C)^n}{n!}\right]\cdot\alpha_{C_H}^{s}\left[\{f_S,H_S\}_n\right]\,.\nn
\ea 
Now suppose $z\phi_{\rm max}-\phi_C^0\leq s<(z+1)\phi_{\rm max}-\phi_C^0$ for $z\in\mathbb Z$. Since $\phi_C$ is $\phi_{\rm max}$-periodic, Eq.~\eqref{solution} entails for this case
\ba 
\alpha_{C_H}^{s}\left[(\tau-\phi_C)^n\right]=\left(\tau-\phi_C^0+z\phi_{\rm max}-s\right)^n\,.\nn
\ea 
Hence, invoking Eq.~\eqref{alpha},
\ba 
\alpha_{C_H}^{s}\cdot F_{f_S,T}(\tau)&=&\sum_{n,m=0}^\infty\f{(\tau-\phi_C^0+z\phi_{\rm max}-s)^n\,s^m}{n!\,m!}\{f_S,H_S\}_{n+m}\,.\nn
\ea 
Using $\f{d}{d\tau}F_{f_S,T}(\tau)=\big\{F_{f_S,T}(\tau),H_S\big\}$, this gives
\ba 
\alpha_{C_H}^{s}\cdot F_{f_S,T}(\tau)&=&\sum_{n,m=0}^\infty\f{d^m}{d\tau^m}\f{(\tau-\phi_C^0+z\phi_{\rm max}-s)^n}{n!}\f{s^m}{m!}\{f_S,H_S\}_n\nn\\
&=&\sum_{m=0}^\infty\f{s^m}{m!}\f{d^m}{d\tau^m}\,F_{f_S,T}(\tau+z\phi_{\rm max}-s)\nn\\
&=&F_{f_S,T}(\tau+z\phi_{\rm max})\,.\nn
\ea 
Finally, we notice that for $s=z\,\phi_{\rm max}-\phi_C^0$,  $z\in\mathbb{Z}$, we have 
\ba 
\alpha_{C_H}^{z\phi_{\rm max}-\phi_C^0}\left[(\tau-\phi_C)^n\right]=\tau^n\,,\nn
\ea 
so that, repeating the steps as above,
\ba 
\alpha_{C_H}^{z\phi_{\rm max}-\phi_C^0}\cdot F_{f_S,T}(\tau)&=&\sum_{n,m=0}^\infty\f{\tau^n(z\phi_{\rm max}-\phi_C^0)^m}{n!\,m!}\{f_S,H_S\}_{n+m}\nn\\
&=&\sum_{n,m=0}^\infty\f{d^m}{d\tau^m}\f{(\tau+\phi_C^0-\phi_C^0)^n}{n!}\f{((z\phi_{\rm max}-\phi_C^0)^m}{m!}\{f_S,H_S\}_n\nn\\
&=&F_{f_S,T}(\tau+z\phi_{\rm max})\,.\nn
\ea 
\end{proof}

\noindent{\bf\hyperref[lem_monoPNP]{Lemma \ref{lem_monoPNP}}.}
The relational observable $F_{T,Q}(\tau)$ that encodes the value of the unravelled (monotonic) clock $T$ of a periodic system [cf.\ Eq.\ \eqref{globalclock}] relative to the value $\tau$ of the clock $Q$ of an aperiodic system [cf.\ Definition \ref{def:nonperclock}] satisfies the \emph{transient invariance property}
\ba
\alpha_{C_H}^{s}\cdot F_{T,Q}(\tau)=\alpha_{C_H}^{z\phi_{\rm max}-\phi_C^0}\cdot F_{T,Q}(\tau)\,,
\ea
with $z\phi_{\rm max}-\phi_C^0\leq s<(z+1)\phi_{\rm max}-\phi_C^0$, for $z\in\mathbb Z$, and
\ba 
\alpha_{C_H}^{z\phi_{\rm max}-\phi_C^0}\cdot F_{T,Q}(\tau)=F_{T,Q}(\tau-z\phi_{\rm max}).\q\;\,
\ea
\begin{proof}
From Definition \ref{def:nonperclock} together with Eq.~\eqref{globalclock} and the analogue of Eq.~\eqref{relobs} for the aperiodic clock $Q$, the relational observable of $T$ relative to $Q$ is found to be
\begin{equation}\label{app-relobs-TQ}
F_{T,Q}(\tau) = \tau-Q^0+\phi_C^0 \ .
\end{equation}
Using $\alpha^s_{C_H}\cdot Q^0 = s+Q^0$ together with Eq.~\eqref{solution}, we have
\begin{equation}
\begin{aligned}
\alpha^s_{C_H}\cdot F_{T,Q}(\tau) &= \tau-(s+Q^0)+s+\phi_C^0-\left\lfloor\frac{s+\phi_C^0}{\phi_{\rm max}}\right\rfloor\phi_{\rm max}\\
&=\tau-Q^0+\phi_C^0-\left\lfloor\frac{s+\phi_C^0}{\phi_{\rm max}}\right\rfloor\phi_{\rm max} \,.
\end{aligned}
\end{equation}
If $z\phi_{\rm max}\leq s+\phi_C^0<(z+1)\phi_{\rm max}$  for $z\in\mathbb{Z}$, then
\begin{equation}
\begin{aligned}
\alpha_{C_H}^s\cdot F_{T,Q}(\tau) &= \tau-Q^0+\phi_C^0-z\phi_{\rm max} =F_{T,Q}(\tau-z\phi_{\rm max}) = \alpha^{z\phi_{\rm max}-\phi_C^0}_{C_H}\cdot F_{T,Q}(\tau) \ .
\end{aligned}
\end{equation}
\end{proof}

\noindent{\bf\hyperref[lem_povm]{Lemma \ref{lem_povm}}.}
The $n^{\rm th}$-moment operators of the covariant clock POVM are not conjugate to the clock Hamiltonian for $n>0$
\ba
[\hat\phi_C^{(n)},\hat H_C] = i\,n\,\hat\phi_C^{(n-1)}-i\,(t_{\rm max})^{n-1}\,\ket{0}\!\bra{0}\,.
\ea
For $n=0$, we clearly have $[\hat\phi_C^{(0)},\hat H_C]=0$.
\begin{proof}
Using $\big[\ket{\varepsilon_i}\!\bra{\varepsilon_j},\hat H_C\big]=(\varepsilon_j-\varepsilon_i)\,\ket{\varepsilon_i}\!\bra{\varepsilon_j}$, it is easy to check that
\ba
\big[\ket{\phi}\!\bra{\phi},\hat H_C\big]=-i\,\partial_\phi\,\ket{\phi}\!\bra{\phi}.\nn
\ea
Invoking Eq.~\eqref{nthmoment} therefore gives
\ba
[\hat\phi_C^{(n)},\hat H_C] &=& \f{1}{t_{\rm max}}\int_0^{t_{\rm max}}\,d\phi\,\phi^n\,\big[\ket{\phi}\!\bra{\phi},\hat H_C\big]\nn\\
&=&-\f{i}{t_{\rm max}}\int_0^{t_{\rm max}}\,d\phi\,\phi^n\,\partial_\phi\,\ket{\phi}\!\bra{\phi}
\ea
Partial integration, taking into account the boundary terms and Eq.~\eqref{id} then yield the claim. 
\end{proof}

\noindent{\bf\hyperref[lem_quantumPhiEvol]{Lemma \ref{lem_quantumPhiEvol}}}.
Let $\ket{\psi_{1,2}}$ be states in the clock Hilbert space such that $\braket{\phi|\psi_{1,2}} = \psi_{1,2}(\phi)$ are integrable functions, and let $\hat{\phi}_C(s):=U^{\dagger}_C(s)\hat{\phi}_CU_C(s)$ with $\braket{\hat{\phi}_C(s)}_{21}:=\braket{\psi_2|\hat{\phi}_C(s)|\psi_1}$. Then, the clock operator $\hat{\phi}_C$ obeys the following evolution law:
\ba
\braket{\hat{\phi}_C(s)}_{21} = \frac{1}{t_{\rm max}}\int_0^{t_{\rm max}}d\phi \ \left(s+\phi-t_{\rm max}\left\lfloor\frac{s+\phi}{t_{\rm max}}\right\rfloor\right)\psi_2^*(\phi)\psi_1(\phi)\,,
\ea
which we take to be the quantum analogue of the classical evolution given by Eq.~\eqref{solution}, with $t_{\rm max}$ being the counterpart to the classical period $\phi_{\rm max}$.
\begin{proof}
Without loss of generality, we consider that $z t_{\rm max}\leq s < (z+1)t_{\rm max}$ for $z\in\mathbb{Z}$. Equivalently, we can write $s=(z+\epsilon)t_{\rm max}$ with $\epsilon\in[0,1)$. From Eqs.~\eqref{statecov} and \eqref{id}, we then find
\begin{equation}
U^{\dagger}_C(s)\ket{\phi}\!\bra{\phi}U_C(s) = \ket{\phi-s}\!\bra{\phi-s}=\ket{\phi-\epsilon t_{\rm max}}\!\bra{\phi-\epsilon t_{\rm max}} \,,
\end{equation}
which implies that the first moment defined from Eq.~\eqref{nthmoment} with $n=1$ satisfies
\begin{equation}\label{quantumPhiEvol0}
\begin{aligned}
\hat{\phi}_C(s)&:=U_C^{\dagger}(s)\hat{\phi}_CU_C(s)\\
&= \frac{1}{ t_{\rm max}}\int_0^{ t_{\rm max}}d\phi\ \phi \ket{\phi-\epsilon t_{\rm max}}\!\bra{\phi-\epsilon t_{\rm max}}\\
&=\frac{1}{ t_{\rm max}}\int_{-\epsilon t_{\rm max}}^{ t_{\rm max}-\epsilon t_{\rm max}}d\phi\ (\phi+\epsilon t_{\rm max}) \ket{\phi}\!\bra{\phi}\\
&=\frac{1}{ t_{\rm max}}\int_{0}^{ t_{\rm max}}d\phi\ (\phi+\epsilon t_{\rm max}) \ket{\phi}\!\bra{\phi}+\frac{1}{ t_{\rm max}}\int_{-\epsilon t_{\rm max}}^{0}d\phi\ (\phi+\epsilon t_{\rm max}) \ket{\phi}\!\bra{\phi}\\
&\ \ \ -\frac{1}{ t_{\rm max}}\int_{ t_{\rm max}-\epsilon t_{\rm max}}^{ t_{\rm max}}d\phi\ (\phi+\epsilon t_{\rm max}) \ket{\phi}\!\bra{\phi} \,.
\end{aligned}
\end{equation}
Using the periodicity property given in Eq.~\eqref{id} and adjusting integration variables and limits, we may write
\begin{equation}\label{quantumPhiEvolFloor}
\begin{aligned}
&\frac{1}{ t_{\rm max}}\int_{-\epsilon t_{\rm max}}^{0}d\phi\ (\phi+\epsilon t_{\rm max}) \psi_2^*(\phi)\psi_1(\phi)-\frac{1}{ t_{\rm max}}\int_{ t_{\rm max}-\epsilon t_{\rm max}}^{ t_{\rm max}}d\phi\ (\phi+\epsilon t_{\rm max}) \psi_2^*(\phi)\psi_1(\phi)\\
&= -\int_{ t_{\rm max}-\epsilon t_{\rm max}}^{ t_{\rm max}}d\phi\ \psi_2^*(\phi)\psi_1(\phi)\\
&=-\int_{0}^{ t_{\rm max}}d\phi\ \left\lfloor\frac{\epsilon t_{\rm max}+\phi}{ t_{\rm max}}\right\rfloor\psi_2^*(\phi)\psi_1(\phi)\,,
\end{aligned}
\end{equation}
where the last equality follows from the fact that $\epsilon\in[0,1)$. By substituting Eq.~\eqref{quantumPhiEvolFloor} into Eq.~\eqref{quantumPhiEvol0} and using the property $\lfloor z+x\rfloor = z+\lfloor x\rfloor$ for $z\in\mathbb{Z}$, we finally obtain
\begin{equation}
\begin{aligned}
\braket{\hat{\phi}_C(s)}_{21} &= \frac{1}{ t_{\rm max}}\int_0^{ t_{\rm max}}d\phi\ \left((z+\epsilon) t_{\rm max}+\phi- t_{\rm max}\left\lfloor\frac{(z+\epsilon) t_{\rm max}+\phi}{ t_{\rm max}}\right\rfloor\right)\psi_2^*(\phi)\psi_1(\phi)\\
&= \frac{1}{ t_{\rm max}}\int_0^{ t_{\rm max}}d\phi\ \left(s+\phi- t_{\rm max}\left\lfloor\frac{s+\phi}{ t_{\rm max}}\right\rfloor\right)\psi_2^*(\phi)\psi_1(\phi) \,.
\end{aligned}
\end{equation}\vspace{-0.1cm}
\end{proof}

\noindent{\bf\hyperref[lem_ruin]{Lemma \ref{lem_ruin}}.}
Suppose $\hat H_C$ has discrete, non-degenerate spectrum and $\hat\phi_C^{(n)}$ is the $n^{\rm th}$-moment operator of the covariant and periodic clock POVM in Eq.~\eqref{nthmoment}. Then
\ba
[\hat\phi_C^{(n)},\hat H_C]\,\ket{\psi_{\rm phys}} = i\,n\,\hat\phi_C^{(n-1)}\,\ket{\psi_{\rm phys}}\label{wantthat}
\ea
only holds for $\ket{\psi_{\rm phys}}\equiv0$.
\begin{proof}
Lemma~\ref{lem_povm} implies that, in order for Eq.~\eqref{wantthat} to be satisfied, we must have $\braket{\phi=0|\psi_{\rm phys}}=0$. By Eqs.~\eqref{discreteClockState} and~\eqref{crap}, this means the following expression must vanish:
\ba
\sum_{\varepsilon_k\in\spec(\hat H_C)}\,e^{-ig(\varepsilon_k)}\braket{\varepsilon_k|\psi_{\rm phys}}\underset{(\ref{crap})}{=}\sum_{-\varepsilon_k\in\sigma_{S\vert C}}\sum_{\sigma_{-\varepsilon_k}}e^{-ig(\varepsilon_k)}\psi_{\rm kin}(\varepsilon_k,-\varepsilon_k,\sigma_{-\varepsilon_k})\,\ket{-\varepsilon_k,\sigma_{-\varepsilon_k}}_S\,.\nn
\ea
Since $\ket{-\varepsilon_k,\sigma_{-\varepsilon_k}}_S$ are part of a (possibly improper) basis for $\ch_S$, this is only possible for $\psi_{\rm kin}(\varepsilon_k,-\varepsilon_k,\sigma_{-\varepsilon_k})\equiv0$, $\forall\,-\varepsilon_k\in\sigma_{S\vert C}$.
\end{proof}

\noindent{\bf\hyperref[lem_noDirac]{Theorem \ref{lem_noDirac}}.}
The commutator between the quantisation of relational observables relative to periodic clocks in Eq.~\eqref{qrelobs} and the constraint evaluates to
\ba
[\hat F_{f_S,T}(\tau),\hat C_H]=-\f{i}{t_{\rm max}}\ket{0}\!\bra{0}\otimes U_S^\dag(\tau)\left[U_S(t_{\rm max})\hat f_SU_S^\dag(t_{\rm max})-\hat f_S\right]U_S(\tau).\nn
\ea

Furthermore, $\hat F_{f_S,T}(\tau)$ is a \emph{weak} quantum Dirac observable, i.e.\ $[\hat F_{f_S,T}(\tau),\hat C_H]\approx0$, where $\approx$ is the weak equality, if and only if $\hat f_S$ is \emph{weakly} $t_{\rm max}$-periodic, i.e.\ if and only if 
\ba
{I_C\otimes U_S(t_{\rm max})\hat f_SU_S^\dag(t_{\rm max})\approx I_C\otimes \hat f_S}.
\ea
In all other cases, $\hat F_{f_S,T}(\tau)$ is neither a weak nor a strong quantum Dirac observable.
\begin{proof}
We expand the first line in Eq.~\eqref{qrelobs} in the $n^{\rm th}$-moment operators to find
\ba
[\hat F_{f_S,T}(\tau),\hat C_H]&=&\sum_{n=0}^\infty\f{i^n}{n!}\,\sum_{k=0}^n\,\binom{n}{k}(-\tau)^k\,\left[\hat\phi_C^{(n-k)}\otimes [\hat f_S,\hat H_S]_n,\hat H_C+\hat H_S\right]\nn\\
&=&\sum_{n=0}^\infty\f{i^n}{n!}\,\sum_{k=0}^n\,\binom{n}{k}(-\tau)^k\left([\hat\phi_C^{(n-k)},\hat H_C]\otimes [\hat f_S,\hat H_S]_n+\hat\phi_C^{(n-k)}\otimes [\hat f_S,\hat H_S]_{n+1}\right).\label{a1}
\ea
Let us focus on the first term. Lemma~\ref{lem_povm} implies
\ba
\sum_{k=0}^n\,\binom{n}{k}(-\tau)^k\,[\hat\phi_C^{(n-k)},\hat H_C]&=&\sum_{k=0}^{n-1}\,\binom{n}{k}(-\tau)^k\left(i(n-k)\,\hat\phi_C^{(n-1-k)}-i(t_{\rm max})^{n-1-k}\,\ket{0}\!\bra{0}\right)\nn\\
&=&i\,n\,\sum_{k=0}^{n-1}\,\binom{n-1}{k}(-\tau)^k\,\hat\phi_C^{(n-1-k)}-\f{i}{t_{\rm max}}\left((t_{\rm max}-\tau)^n-(-\tau)^n\right)\,\ket{0}\!\bra{0}.\label{a2}
\ea
It is easy to see that the first term on the r.h.s.\ of Eq.~\eqref{a2}, when reinserted into Eq.~\eqref{a1}, cancels the second term on the r.h.s.\ in Eq.~\eqref{a1} upon a relabeling of the summation index. We are then left with the term proportional to $\ket{0}\!\bra{0}$ in Eq.~\eqref{a2} coming from the correction term to the canonical commutation relations in Lemma~\ref{lem_povm}. Inserting it into Eq.~\eqref{a1} gives
\ba
[\hat F_{f_S,T}(\tau),\hat C_H]=-\f{i}{t_{\rm max}}\ket{0}\!\bra{0}\otimes\sum_{n=0}^\infty\f{i^n}{n!}\,\left((t_{\rm max}-\tau)^n-(-\tau)^n\right)\,[\hat f_S,\hat H_S]_n,\nn
\ea
which upon invoking the Baker-Cambpell-Hausdorff relation yields the first claim. 

It is clear that $[\hat F_{f_S,T}(\tau),\hat C_H]=0$, in which case $\hat F_{f_S,T}(\tau)$ is a strong Dirac observable, if and only if $\hat f_S$ is $t_{\rm max}$-periodic, i.e.\ if $U_S(t_{\rm max})\hat f_SU_S^\dag(t_{\rm max})=\hat f_S$. Furthermore, since the proof of Lemma~\ref{lem_ruin} shows that $\braket{0|\psi_{\rm phys}}=0$ only if $\ket{\psi_{\rm phys}}=0$, it is also evident that 
$
[\hat F_{f_S,T}(\tau),\hat C_H]\ket{\psi_{\rm phys}}=0
$
for arbitrary $\ket{\psi_{\rm phys}}\in\ch_{\rm phys}$ if and only if 
\ba
I_C\otimes U_S^\dag(\tau)\left[U_S(t_{\rm max})\hat f_SU_S^\dag(t_{\rm max})-\hat f_S\right]U_S(\tau) \ket{\psi_{\rm phys}}
\underset{(\ref{WDW})}{=} U_{CS}^\dag(\tau)\left( I_C\otimes U_S(t_{\rm max})\hat f_SU_S^\dag(t_{\rm max})-I_C\otimes\hat f_S\right)\ket{\psi_{\rm phys}}=0,\nn
\ea
which is equivalent to Eq.~\eqref{condobs} and to $\hat F_{f_S,T}(\tau)$ being a weak Dirac observable.
Hence, when this condition does not hold, $\hat F_{f_S,T}(\tau)$ is neither a weak nor strong quantum Dirac observable.
\end{proof}

\noindent{\bf\hyperref[lem_clobsQ]{Lemma \ref{lem_clobsQ}.}}
Let $\ket{\psi_{1}}$ be a (physical) state and $\ket{\psi_2}$ a kinematical state such that $\bra{\phi}\otimes\braket{q|\psi_{1,2}} = \psi_{1,2}(\phi,q)$ are integrable functions of $\phi$ for any choice of basis $\ket{q}$ in the system Hilbert space. Given the quantum relational observables defined in Eq.~\eqref{qrelobs}, let $\alpha^s_{C_H}\cdot\hat{F}_{f_S,T}(\tau):= U_{CS}^{\dagger}(s)\hat{F}_{f_S,T}(\tau)U_{CS}(s)$ and $\braket{\alpha^s_{C_H}\cdot\hat{F}_{f_S,T}(\tau)}_{21}:=\braket{\psi_2|\alpha^s_{C_H}\cdot\hat{F}_{f_S,T}(\tau)|\psi_1}$.~(Note that the latter expression invokes the physical inner product Eq.~\eqref{PIP}.)~Then, the quantum relational observables obey the following property:
\begin{equation}\label{clobsQ-0}
\braket{\alpha^s_{C_H}\cdot\hat{F}_{f_S,T}(\tau)}_{21}= \bra{\psi_2}\frac{1}{ t_{\rm max}}\int_0^{ t_{\rm max}}d\phi\,\ket{\phi}\!\bra{\phi}\otimes\hat{f}_S\left(\left(\tau+ t_{\rm max}\left\lfloor\frac{s+\phi}{ t_{\rm max}}\right\rfloor\right)-\phi\right)\ket{\psi_1}\,,
\end{equation}
where $\hat{f}_S(t) = U_S^\dagger(t)\hat{f}_SU_S(t)$ with $U_S(t)=\exp(-it\hat{H}_S)$.~We take Eq.~\eqref{clobsQ-0} to be the quantum version of the transient invariance property of classical relational observables established in Lemma \ref{lem_clobs}.
\begin{proof}
Using the second line of Eq.~\eqref{qrelobs}, we may write
\begin{equation}\label{clobsQ-1}
U_{CS}^{\dagger}(s)\hat{F}_{f_S,T}(\tau)U_{CS}(s)= \frac{1}{ t_{\rm max}}\int_0^{ t_{\rm max}}d\phi\,\ket{\phi-s}\!\bra{\phi-s}\otimes\hat{f}_S\left(\tau-(\phi-s)\right) \,.
\end{equation}
Without loss of generality, we consider that $z t_{\rm max}\leq s<(z+1) t_{\rm max}$ or, equivalently, $s=(z+\epsilon) t_{\rm max}$ with $\epsilon\in[0,1)$. Then, due to Eq.~\eqref{id}, Eq.~\eqref{clobsQ-1} becomes
\begin{equation}\label{clobsQ-2}
\begin{aligned}
U_{CS}^{\dagger}(s)\hat{F}_{f_S,T}(\tau)U_{CS}(s)&= \frac{1}{ t_{\rm max}}\int_0^{ t_{\rm max}}d\phi\,\ket{\phi-\epsilon t_{\rm max}}\!\bra{\phi-\epsilon t_{\rm max}}\otimes\hat{f}_S\left((\tau+z t_{\rm max})-(\phi-\epsilon t_{\rm max})\right)\\
&= \frac{1}{ t_{\rm max}}\int_{-\epsilon t_{\rm max}}^{(1-\epsilon) t_{\rm max}}d\phi\,\ket{\phi}\!\bra{\phi}\otimes\hat{f}_S\left((\tau+z t_{\rm max})-\phi\right)\\
&=\frac{1}{ t_{\rm max}}\int_{(1-\epsilon) t_{\rm max}}^{ t_{\rm max}}d\phi\,\ket{\phi}\!\bra{\phi}\otimes\hat{f}_S\left((\tau+z t_{\rm max})-(\phi- t_{\rm max})\right)\\
&\ \ \ +\frac{1}{ t_{\rm max}}\int_{0}^{(1-\epsilon) t_{\rm max}}d\phi\,\ket{\phi}\!\bra{\phi}\otimes\hat{f}_S\left((\tau+z t_{\rm max})-\phi\right) \,.
\end{aligned}
\end{equation}
As $\epsilon\in[0,1)$, we note that
\begin{equation}
\left\lfloor\frac{s+\phi}{ t_{\rm max}}\right\rfloor = \left\{
\begin{matrix}
z & \text{if}\ 0\leq\phi < (1-\epsilon) t_{\rm max}\\
z+1 & \text{if}\ (1-\epsilon) t_{\rm max}\leq\phi <  t_{\rm max}
\end{matrix}
\right.\,,
\end{equation}
which implies that Eq.~\eqref{clobsQ-2} can be concatenated into the result given in Eq.~\eqref{clobsQ-0} for a pair of wave functions that are integrable in $\phi$.
\end{proof}

\noindent{\bf\hyperref[lem_fullG]{Lemma \ref{lem_fullG}}.}
The $G$-twirl over the full group generated by the constraint $\hat{C}_H$ yields
\begin{eqnarray}
   \f{1}{|H|} \cg_G\left( \ket{\tau}\!\bra{\tau} \otimes \hat{f}_S\right)=\f{t_{\rm max}}{N_G} \,\mathcal{G}_{[0,t_{\rm max})}\left( \ket{\tau}\!\bra{\tau} \otimes \hat{f}_S^{H}\right)=\f{t_{\rm max}}{N_G} \,U_{CS}(\tau)\,\hat F_{f^{H}_S,T}(\tau)\,U_{CS}^\dag(\tau)\,,\nn
\end{eqnarray}
where
\begin{equation}
    \hat{f}_S^{H}=\cg_H\left(\hat{f}_S\right)=\f{1}{|H|}\sum_{z\in\mathbb{Z}_G}\,U_S(zt_{\rm max})\,\hat f_S\,U_S^\dag(zt_{\rm max})\nn
\end{equation}
is the averaging over the isotropy group $H=\mathbb{Z}_G$ of clock $C$ with respect to the full group $G$ and $|H|$ its cardinality. In the case that $G=\rm{U}(1)$, this assumes that an integer multiple of clock cycles fits into one cycle of $G$. (The normalisation constant is $N_G=\tilde{t}_{\rm max}$ for $G=\rm{U}(1)$, where $\tilde{t}_{\rm max}$ is the analog of $t_{\rm max}$ in Eq.~\eqref{point}, but for the constraint $\hat{C}_H$, and $N_G=2\pi$ for $G=(\mathbb{R},+)$ \cite{Hoehn:2019owq}.)

\begin{proof}
We have
\ba
\cg_G\left(\ket{\tau}\!\bra{\tau} \otimes \hat{f}_S\right) &=&\f{1}{N_G}\int_G d\phi\, U_{CS}(\phi)\left( \ket{\tau}\!\bra{\tau} \otimes \hat{f}_S\right)\,U^\dag_{CS}(\phi)\nn\\
&=&\f{1}{N_G}\sum_{z\in\mathbb{Z}_G}\int_{t_{\rm max} z}^{t_{\rm max}(z+1)} d\phi \, U_{CS}(\phi)\left( \ket{\tau}\!\bra{\tau} \otimes \hat{f}_S\right)\,U_{CS}^\dag(\phi)\nn\\
&=&\f{1}{N_G}\sum_{z\in\mathbb{Z}_G}U_{CS}( z t_{\rm max})\int_{0}^{t_{\rm max}} d\phi \, U_{CS}(\phi)\left( \ket{\tau}\!\bra{\tau} \otimes \hat{f}_S\right)\,U_{CS}^\dag(\phi) U_{CS}^\dag(z t_{\rm max})\nn\\
&=&\f{1}{N_G}\sum_{z\in\mathbb{Z}_G}\left(I_C\otimes U_S(z t_{\rm max})\right)\int_{0}^{t_{\rm max}} d\phi \, U_{CS}(\phi)\left( \ket{\tau}\!\bra{\tau} \otimes \hat{f}_S\right)\,U_{CS}^\dag(\phi)\left(I_C\otimes U^\dag_S(z t_{\rm max})\right)\nn\\
&=&\f{|\mathbb{Z}_G|}{N_G}\int_{0}^{t_{\rm max}} d\phi \, U_{CS}(\phi)\left( \ket{\tau}\!\bra{\tau} \otimes \mathcal{G}_{H}\left(\hat{f}_S\right)\right)\,U_{CS}^\dag(\phi)\nn\\
&=&|\mathbb{Z}_G|\,\f{t_{\rm max}}{N_G} \,\mathcal{G}_{[0,t_{\rm max})}\left( \ket{\tau}\!\bra{\tau} \otimes \mathcal{G}_{H}\left(\hat{f}_S\right)\right)\nn\\
&=&|\mathbb{Z}_G|\,\f{t_{\rm max}}{N_G} \,U_{CS}(\tau)\,\hat F_{\mathcal{G}_{H}(f_S),T}(\tau)\,U_{CS}^\dag(\tau),\label{fullG}
\ea
where in the last line, we made use of Eqs.~\eqref{qrelobs} and~\eqref{projrep}.
\end{proof}

\noindent{\bf\hyperref[lem_sysphysIP]{Lemma \ref{lem_sysphysIP}}.}
Let $E,E'\in\sigma_{S\vert C}$. Then 
\ba
\braket{E,\sigma_E|E',\sigma_{E'}}_{\ch_S^{\rm phys}} = \delta_{E,E'}\delta_{\sigma_E,\sigma_{E'}}.\nn
\ea
\begin{proof}
For case (a) the statement is trivial. In case (b), we have to choose an arbitrary normalizable state in $\ch_S$ that projects under $\Pi_{\sigma_{S\vert C}}$ to $\ket{E,\sigma_E}$, $E\in\sigma_{S\vert C}$. Let $I_E\subset\spec(\hat H_S)$ be an interval that contains the eigenvalue $E\in\sigma_{S\vert C}$, but not any other $E'\in\sigma_{S\vert C}$, and let $\chi(\tilde E,\sigma_{\tilde E})$ be an arbitrary square integrable function such that $\chi(\tilde E=E,\sigma_E)=1$. Then 
\ba
\ket{\chi_E}=\int_{I_E} d\tilde E\,\chi(\tilde E,\sigma_{\tilde E})\,\ket{\tilde E,\sigma_{\tilde E}}_S\nn
\ea
is such a state. Using it in the physical system inner product, Eq.~\eqref{sysPIP}, we find
\ba
\braket{E,\sigma_E|E',\sigma_{E'}}_{\ch_S^{\rm phys}} &&=\braket{\chi_E| E',\sigma_{E'}}_S\nn\\
&&=\int_{I_E} d\tilde E\,\chi^*(\tilde E,\sigma_{\tilde E})\delta(\tilde E-E')\delta_{\sigma_{\tilde E},\sigma_{E'}}
\ea
It is clear that this expression is zero if $E'\notin I_E$ and equal to $\delta_{\sigma_E,\sigma_{E'}}$ otherwise (in which case $E=E'$).
\end{proof}

\noindent{\bf\hyperref[lem_5]{Lemma \ref{lem_5}}.}
The reduction maps satisfy for all admissible unravelled clock readings $\tau$
\ba
\calr_{\mathbf S}^{-1}(\tau)\cdot\calr_{\mathbf S}(\tau)&\approx& I_{\rm phys}\,,\nn\\
\calr_{\mathbf S}(\tau)\cdot\calr_{\mathbf S}^{-1}(\tau)&\approx_S& I_S^{\rm phys},\nn
\ea
where $\approx$ denotes a weak equality, i.e.\ equality on the physical Hilbert space $\ch_{\rm phys}$, $\approx_S$ denotes the system weak equality, i.e.\ equality on $\ch_S^{\rm phys}$, and $I_{\rm phys}$ and $I_S^{\rm phys}$ are the identities on $\ch_{\rm phys}$ and $\ch_S^{\rm phys}$, respectively.

\begin{proof}
We begin with the first identity. Pick any $\ket{\psi_{\rm phys}}\in\ch_{\rm phys}$. Using Eqs.~\eqref{PWred} and~\eqref{PWinvred}, we have
\ba
\calr_{\mathbf S}^{-1}(\tau)\cdot\calr_{\mathbf S}(\tau)\,\ket{\psi_{\rm phys}}&&=\calr_{\mathbf S}^{-1}(\tau_C)\cdot\calr_{\mathbf S}(\tau_C)\,\ket{\psi_{\rm phys}}\nn\\
&&=\f{1}{t_{\rm max}}\int_0^{t_{\rm max}}d\phi\,\ket{\phi}\!\bra{\tau_C}\otimes U_S(\phi-\tau_C)\,\ket{\psi_{\rm phys}}\nn\\
&&=\f{1}{t_{\rm max}}\int_0^{t_{\rm max}}d\phi\,\ket{\phi}\!\bra{\tau_C} U_C^\dag(\phi-\tau_C)\otimes I_S\,\ket{\psi_{\rm phys}}\nn\\
&&=\f{1}{t_{\rm max}}\int_0^{t_{\rm max}}d\phi\,\ket{\phi}\!\bra{\phi} \otimes I_S\,\ket{\psi_{\rm phys}}\nn\\
&&\underset{(\ref{resolid})}{=}\ket{\psi_{\rm phys}}.\nn
\ea
In the third line, we made use of Eq.~\eqref{WDW}.

Next, we begin proving the second identity,
\ba
\calr_{\mathbf S}(\tau)\cdot\calr_{\mathbf S}^{-1}(\tau)
&&\;=\f{1}{t_{\rm max}}\int_0^{t_{\rm max}}d\phi\,\braket{\tau|\phi} \,U_S(\phi-\tau)\nn\\
&&\underset{(\ref{id3})}{=}\sum_{\varepsilon_j \in \spec(\hat H_C)} \,\,\,\,{{\intsum}_{E\in\spec(\hat H_S)}}\,\sum_{\sigma_E}\,e^{i\tau(\varepsilon_j+E)}\f{1}{t_{\rm max}}\int_0^{t_{\rm max}}d\phi\,e^{-i\phi(\varepsilon_j+E)}\ket{E,\sigma_E}\!\bra{E,\sigma_E}.\nn
\ea
Noting that the identity in Eq.~\eqref{id2} applies to the $\phi$-integral in the last line, provided we can guarantee that $E\in\sigma_{S\vert C}$, we find upon multiplying both sides from the right with $\Pi_{\sigma_{S\vert C}}$ given in Eq.~\eqref{physsproj}
\ba
\calr_{\mathbf S}(\tau)\cdot\calr_{\mathbf S}^{-1}(\tau)\,\Pi_{\sigma_{S\vert C}}\underset{(\ref{id2})}{=}\sum_{E\in\sigma_{S\vert C}}\,\sum_{\sigma_E}\ket{E,\sigma_E}\!\bra{E,\sigma_E}=\Pi_{\sigma_{S\vert C}},\nn
\ea
which proves the second claim. 
\end{proof}

\noindent{\bf\hyperref[thm_relobs]{Theorem \ref{thm_relobs}}.}
Let $\hat f_S^{\rm phys}\in\cl(\ch_S^{\rm phys})$ be a physical system observable. Its embedding coincides weakly with the quantisation of the relational observables in Eq.~\eqref{qrelobs},
\ba
\mathcal{E}_{\mathbf S}^\tau\left(\hat f_S^{\rm phys}\right)\approx\hat F_{f_S^{\rm phys},T}(\tau),
\ea
which in this case \emph{are} weak quantum Dirac observables, i.e.\ $[\hat F_{f_S^{\rm phys},T}(\tau),\hat C_H]\approx0$.

Conversely, the reduction of a relational observable associated with a physical system observable $\hat f_S^{\rm phys}$ coincides with that observable on the physical system Hilbert space $\ch_S^{\rm phys}$,
\ba
\calr_{\mathbf S}(\tau)\,\hat F_{f_S^{\rm phys},T}(\tau)\,\calr_{\mathbf S}^{-1}(\tau)\approx_S \hat f_S^{\rm phys}.
\ea
\begin{proof}
Eqs.~\eqref{PWred},~\eqref{PWinvred} and~\eqref{embed} entail
\ba
\mathcal{E}_{\mathbf S}^\tau\left(\hat f_S^{\rm phys}\right)=U_{CS}^\dag(\tau)\Big[\f{1}{t_{\rm max}}\int_0^{t_{\rm max}}d\phi\,U_{CS}(\phi)\left(\ket{\tau}\!\bra{\tau}\otimes\hat f_S^{\rm phys}\right)\Big],\nn
\ea
which, owing to $U_{CS}(\tau)\ket{\psi_{\rm phys}}=\ket{\psi_{\rm phys}}$, is weakly equivalent to the expression in Eq.~\eqref{qrelobs}. Owing to Eq.~\eqref{periodicc}, Theorem~\ref{lem_noDirac} tells us that $\hat F_{f_S^{\rm phys},T}(\tau)$ \emph{is} a weak quantum Dirac observable.

Conversely, inserting the expressions in Eqs.~\eqref{qrelobs},~\eqref{PWred} and~\eqref{PWinvred}, one finds
\ba
\calr_{\mathbf S}(\tau)\,\hat F_{f_S^{\rm phys},T}(\tau)\,\calr_{\mathbf S}^{-1}(\tau)\,\Pi_{\sigma_{S\vert C}} =\left(\f{1}{t_{\rm max}}\int_0^{t_{\rm max}}d\phi\braket{\tau|\phi}U_S(\phi-\tau)\right)\hat f_S^{\rm phys}\left(\f{1}{t_{\rm max}}\int_0^{t_{\rm max}}d\phi'\braket{\phi|\phi'}U_S(\phi'-\phi)\right)\,\Pi_{\sigma_{S\vert C}}.\nn
\ea
Invoking the second part of the proof of Lemma~\ref{lem_5}  yields the claim.
\end{proof}

\noindent{\bf\hyperref[thm_expobs]{Theorem \ref{thm_expobs}}.}
Let $\hat f_S^{\rm phys}\in\cl(\ch_S^{\rm phys})$ be a physical system observable. The expectation value of the corresponding relational observable evaluated in the physical inner product on $\ch_{\rm phys}$, given in Eq.~\eqref{PIP}, coincides with the expectation value of $\hat{f}_S^{\rm phys}$ evaluated in the inner product on $\ch_S^{\rm phys}$, given in Eq.~\eqref{sysPIP}, i.e.
\ba
\braket{\phi_{\rm phys}|\,\hat F_{f_S^{\rm phys},T}(\tau)\,|\psi_{\rm phys}}_{\rm phys}=\braket{\phi_S^{\rm phys}(\tau)|\,\hat f_S^{\rm phys}\,|\psi_S^{\rm phys}(\tau)}_{\ch_S^{\rm phys}}=\braket{\phi_S(\tau)|\,\hat f_S^{\rm phys}\,|\psi_S^{\rm phys}(\tau)}_S,\nn
\ea
where 
\begin{itemize}
\item[(i)] physical states and physical system states are related by Page-Wootters reduction, ${\ket{\psi_S^{\rm phys}(\tau)}:=\calr_{\mathbf S}(\tau)\,\ket{\psi_{\rm phys}}}$ and similarly for $\ket{\phi_S^{\rm phys}(\tau)}$, and 
\item[(ii)] $\ket{\phi_S(\tau)}:=U_S(\tau)\ket{\phi_S}$ is any kinematical system state $\ket{\phi_S}\in\ch_S$ such that $\Pi_{\sigma_{S\vert C}}\,\ket{\phi_S(\tau)} = \calr_{\mathbf S}(\tau)\ket{\phi_{\rm phys}}=\ket{\phi_S^{\rm phys}(\tau)}\in\ch_S^{\rm phys}$.
\end{itemize}

\begin{proof}
Invoking the definitions of the physical inner product in Eq.~\eqref{PIP} and of the encoding map in Eq.~\eqref{embed}, as well as Theorem~\ref{thm_relobs}, yields
\be
\braket{\phi_{\rm phys}|\,\hat F_{f_S^{\rm phys},T}(\tau)\,|\psi_{\rm phys}}_{\rm phys}=\braket{\phi_{\rm kin}|\,\calr_{\mathbf S}^{-1}(\tau)\,\hat f_S^{\rm phys}\,|\psi_S^{\rm phys}(\tau)}_S,\nn
\ee
where $\ket{\phi_{\rm kin}}$ is any state in the equivalence class of kinematical states that project under $\Pi_{\rm phys}$ to the same physical state $\ket{\phi_{\rm phys}}$.
All that remains to be shown is that $\bra{\phi_{\rm kin}}\,\calr_{\mathbf S}^{-1}(\tau)\,\Pi_{\sigma_{S\vert C}}=\bra{\phi_S^{\rm phys}(\tau)}=\bra{\phi_{\rm phys}}\calr_{\mathbf S}^\dag(\tau)$. But this is easy to check.~Recalling the expression in Eq.~\eqref{PWinvred1} for the inverse reduction map and using Eq.~\eqref{kinstate}, we compute
\ba
\bra{\phi_{\rm kin}}\,\calr_{\mathbf S}^{-1}(\tau)\,\Pi_{\sigma_{S\vert C}}&&=\f{1}{t_{\rm max}}\int_0^{t_{\rm max}}d\phi\bra{\phi_{\rm kin}}\left(\ket{\phi}\otimes\Pi_{\sigma_{S\vert C}}U_S(\phi-\tau)\right)\nn\\
&&=\sum_{\varepsilon_k\in\spec(\hat H_C)}\sum_{E\in\sigma_{S\vert C}}\sum_{\sigma_E}\phi^*_{\rm kin}(\varepsilon_k,E,\sigma_E)\bra{E,\sigma_E}e^{ig(\varepsilon_k)+iE\tau}\nn\\
&&\q\q\q\q\times \f{1}{t_{\rm max}}\int_0^{t_{\rm max}}d\phi \,e^{-i(\varepsilon_k+E)\phi}\nn\\
&&\underset{(\ref{id2})}{=} \sum_{E\in\sigma_{S\vert C}}\sum_{\sigma_E}\phi^*_{\rm kin}(-E,E,\sigma_E)\bra{E,\sigma_E}e^{ig(-E)+iE\tau}\nn\\
&&\underset{(\ref{schrodexp})}{=}\bra{\phi_S^{\rm phys}(\tau)}.\nn
\ea
Recalling the definition of the physical system inner product in Eq.~\eqref{sysPIP}, this proves the claim.
\end{proof}

\noindent{\bf\hyperref[lem_6]{Lemma \ref{lem_6}}.}
Let $G$ be the group generated by the constraint $\hat C_H$. If $G=\rm{U}(1)$, then the conditional inner product equals the physical inner product, i.e.
\ba
\braket{\phi_{\rm phys}|\ket{\tau}\!\bra{\tau}\otimes I_S |\psi_{\rm phys}}_{\rm kin}=\braket{\phi_{\rm phys}|\psi_{\rm phys}}_{\rm phys}.\nn
\ea
However, if $G=\mathbb{R}$, then the inner conditional product $\braket{\phi_{\rm phys}|\ket{\tau}\!\bra{\tau}\otimes I_S |\psi_{\rm phys}}_{\rm kin}$ diverges.

\begin{proof}
Let $\mathbb{Z}_G$ denote the set of integers counting the clock cycles which fit into one period of $G$ (hence labelling the isotropy group $H$), and recall that when $G=\mathbb{R}$, $\mathbb{Z}_G=\mathbb{Z}$. Furthermore let $N_G$ be the normalisation factor associated with the average over the group $G$ (cf.\ Lemma~\ref{lem_fullG}), and note that we can write
\ba
\Pi_{\rm phys}=\f{1}{N_G}\sum_{z\in\mathbb{Z}_G}\int_{zt_{\rm max}}^{(z+1)t_{\rm max}}d\phi\,U_{CS}(\phi).\nn
\ea
We then find
\ba
\braket{\phi_{\rm phys}|\ket{\tau}\!\bra{\tau}\otimes I_S |\psi_{\rm phys}}_{\rm kin}&&=\braket{\phi_{\rm kin}|\Pi_{\rm phys}\left(\ket{\tau}\!\bra{\tau}\otimes I_S\right) |\psi_{\rm phys}}_{\rm kin}\nn\\
&&=\f{1}{N_G}\sum_{z\in\mathbb{Z}_G}\braket{\phi_{\rm kin}|\int_{zt_{\rm max}}^{(z+1)t_{\rm max}}d\phi\,U_{CS}(\phi)\left(\ket{\tau}\!\bra{\tau}\otimes I_S\right) |\psi_{\rm phys}}_{\rm kin}\nn\\
&&=\f{t_{\rm max}}{N_G}\sum_{z\in\mathbb{Z}_G}\bra{\phi_{\rm kin}}U_{CS}(zt_{\rm max})
\f{1}{t_{\rm max}}\int_{0}^{t_{\rm max}}d\phi\,U_{CS}(\phi)\left(\ket{\tau}\!\bra{\tau}\otimes I_S\right) \ket{\psi_{\rm phys}}_{\rm kin}\nn\\
&&=\f{t_{\rm max}}{N_G}\sum_{z\in\mathbb{Z}_G}\bra{\phi_{\rm kin}}U_{CS}(zt_{\rm max})\ket{\psi_{\rm phys}}_{\rm kin}\nn\\
&&=\f{t_{\rm max}}{N_G}\sum_{z\in\mathbb{Z}_G}\braket{\phi_{\rm kin}|\psi_{\rm phys}}_{\rm kin}\nn\\
&&=|\mathbb{Z}_G|\f{t_{\rm max}}{N_G}\,\braket{\phi_{\rm phys}|\psi_{\rm phys}}_{\rm phys},\label{eProofLem10}
\ea
where in the fourth equality we have made use of the first part of the proof of Lemma~\ref{lem_5}. Now, consider the case where $G=\mathbb{R}$. Then $N_G=2\pi$ and $|\mathbb{Z}_G|$ diverges, and therefore so too does $\braket{\phi_{\rm phys}|\ket{\tau}\!\bra{\tau}\otimes I_S |\psi_{\rm phys}}_{\rm kin}$, proving the second statement of the lemma. To prove the first statement, consider the case where $G=\rm{U}(1)$, and let $\ket{-E}_C\otimes\ket{E,\sigma_E}_S$, denote a zero-eigenvector of $\hat{C}_{H}$.~This corresponds to case (a) of Sec.~\ref{ssec_diracgen}, and so this eigenvector is normalisable and further also an eigenvector of $\Pi_{\rm phys}$ (cf.\ Eq.~\eqref{impproj})
\ba
\ket{-E}_C\otimes\ket{E,\sigma_E}_S &=& \Pi_{\rm phys} \ket{-E}_C\otimes\ket{E,\sigma_E}_S \nn \\
&= & \f{1}{N_G}\sum_{z\in\mathbb{Z}_G}\int_{zt_{\rm max}}^{(z+1)t_{\rm max}}d\phi\,U_{CS}(\phi) \ket{\varepsilon_{C},-\varepsilon_{S}} \nn \\
&=& \f{1}{N_G} \left(\sum_{z\in\mathbb{Z}_G}\int_{zt_{\rm max}}^{(z+1)t_{\rm max}}d\phi \right) \ket{\varepsilon_{C},-\varepsilon_{S}} \nn \\
&=& \f{|\mathbb{Z}_G|t_{\rm max}}{N_G} \ket{-E}_C\otimes\ket{E,\sigma_E}_S, \label{eNormalisationNote}
\ea
and therefore $N_G=|\mathbb{Z}_G|t_{\rm max}$, which itself is the period of the $\rm{U}(1)$-representation generated by $\hat{C}_H$. Inserting this into Eq.~\eqref{eProofLem10}, we find that
\ba
\braket{\phi_{\rm phys}|\ket{\tau}\!\bra{\tau}\otimes I_S |\psi_{\rm phys}}_{\rm kin} = \braket{\phi_{\rm phys}|\psi_{\rm phys}}_{\rm phys},
\ea
concluding the proof.
\end{proof}

\noindent{\bf\hyperref[lem_triv1]{Lemma \ref{lem_triv1}}.}
On solutions to the constraint in Eq.~\eqref{WDW}, the inverse $\ct_C(\ch_{\rm phys})\rightarrow\ch_{\rm phys}$ of the constraint trivialisation map is given by
\ba
\ct_{C}^{(-1)}=\f{1}{t_{\rm max}}\int_0^{t_{\rm max}} d\phi\,\ket{\phi}\!\bra{\phi}\otimes e^{-i\phi(\hat H_S+\varepsilon_*I_S)},\label{trivinv}
\ea
so that
\ba
\ct_{C}^{(-1)}\cdot\ct_{C}\approx I_{\rm phys}.\nn
\ea

\begin{proof}
Invoking Eqs.~\eqref{discreteClockState} and~\eqref{id3} yields
\ba
\ct^{(-1)}_{C}\cdot\ct_C = \f{1}{t_{\rm max}^2}\sum_{\varepsilon_j,\varepsilon_k,\varepsilon_l\in\spec(\hat H_C)}\,e^{i(g(\varepsilon_k)-g(\varepsilon_l))}\ket{\varepsilon_k}_C\!\bra{\varepsilon_l}\otimes\int_0^{t_{\rm max}} d\phi d\phi'\,e^{-i\phi'(\hat H_S+\varepsilon_*-\varepsilon_j+\varepsilon_l)}\,e^{i\phi(\hat H_S+\varepsilon_*-\varepsilon_j+\varepsilon_k)}.\label{tinvt}
\ea
Next, let $\ket{\psi_{\rm phys}}$ be an arbitrary physical state. We use Eqs.~\eqref{crap} and~\eqref{id2} and $\braket{\varepsilon_k|\varepsilon_l}=\delta_{\varepsilon_k,\varepsilon_l}$ to find
\ba
&&\ket{\varepsilon_k}_C\!\bra{\varepsilon_l}\otimes\f{1}{t_{\rm max}^2}\int_0^{t_{\rm max}} d\phi d\phi'\,e^{-i\phi'(\hat H_S+\varepsilon_*-\varepsilon_j+\varepsilon_l)}\,e^{i\phi(\hat H_S+\varepsilon_*-\varepsilon_j+\varepsilon_k)}\,\ket{\psi_{\rm phys}}\nn\\
&&\q=\begin{cases}
  0    & \text{if $-\varepsilon_l\notin\sigma_{S\vert C}$}, \\
\delta_{\varepsilon_*,\varepsilon_j}\delta_{\varepsilon_k,\varepsilon_l}\,  \sum_{\sigma_{-\varepsilon_l}}\,\psi_{\rm kin}(\varepsilon_l,-\varepsilon_l,\sigma_{-\varepsilon_l})\,\ket{\varepsilon_k}_C\ket{-\varepsilon_l,\sigma_{-\varepsilon_l}}_S    & \text{if $-\varepsilon_l\in\sigma_{S\vert C}$}.
\end{cases}\nn
\ea
In conjunction with Eqs.~\eqref{tinvt} and~\eqref{crap}, this entails
\ba
\ct^{(-1)}_{C}\cdot\ct_C\,\ket{\psi_{\rm phys}}=\ket{\psi_{\rm phys}}.\nn
\ea
\end{proof}

\noindent{\bf\hyperref[lem_8]{Lemma \ref{lem_8}}.}
The map $\ct_C$ (weakly) trivialises the constraint in Eq.~\eqref{WDW} to the clock degrees of freedom on $\ct_C(\ch_{\rm phys})$, i.e.
\ba
\ct_C\,\hat C_H\,\ct_C^{(-1)}\overset{*}{\approx}\left(\hat H_C-\varepsilon_*\right)\otimes I_S.\nn
\ea
Furthermore, it transforms physical states in Eq.~\eqref{crap} into a product form (relative to the tensor factorization of $\ch_{\rm kin}$):
\ba
&&\ct_C\,\ket{\psi_{\rm phys}} =e^{ig(\varepsilon_*)} \ket{\varepsilon_*}_C\otimes\ket{\psi_S^{\rm phys}}\nn
\ea
with $\ket{\psi_S^{\rm phys}}\in\ch_S^{\rm phys}$ given by Eqs.~\eqref{physsysstate} and~\eqref{Heisstate}.

\begin{proof}
To prove the first statement, we use Lemma~\ref{lem_povm} to compute
\ba
[\ct_C,\hat H_C\otimes I_S] &=&\sum_{n=0}^\infty\,\f{i^n}{n!}\,[\hat\phi^{(n)},\hat H_C]\otimes\left(\hat H_S+\varepsilon_*I_S\right)^n\nn\\
&=&-I_C\otimes\left(\hat H_S+\varepsilon_*I_S\right)\,\ct_C-\f{i}{t_{\rm max}}\sum_{n=1}^\infty\,\f{(it_{\rm max})^n}{n!}\,\ket{0}_C\!\bra{0}\otimes \left(\hat H_S+\varepsilon_*I_S\right)^n\label{commutator}\\
&=&-I_C\otimes\left(\hat H_S+\varepsilon_*I_S\right)\,\ct_C-\f{i}{t_{\rm max}}\ket{0}_C\!\bra{0}\otimes \left(e^{it_{\rm max}(\hat H_S+\varepsilon_*I_S)}-I_S\right).\nn
\ea
Taking the form of the constraint in Eq.~\eqref{WDW} into account and noting that $I_C\otimes \hat H_S$ commutes with $\ct_C$, the commutator in Eq.~\eqref{commutator} implies
\ba
\ct_C\,\hat C_H\,\ct_C^{(-1)}=\left(\hat H_C-\varepsilon_*I_C\right)\otimes I_S\,\ct_C\cdot\ct_C^{(-1)}-\f{i}{t_{\rm max}}\ket{0}_C\!\bra{0}\otimes \left(e^{it_{\rm max}(\hat H_S+\varepsilon_*I_S)}-I_S\right)\ct_C^{(-1)}.\nn
\ea
Next, we apply this relation to trivialised physical states, i.e.\ states of the form $\ct_C\ket{\psi_{\rm phys}}$ for some $\ket{\psi_{\rm phys}}\in\ch_{\rm phys}$. First, using Lemma~\ref{lem_triv1}, as well as Eqs.~\eqref{point} and $-\varepsilon_*\in\sigma_{S|C}$, yields
\ba
\ket{0}_C\!\bra{0}\otimes \left(e^{it_{\rm max}(\hat H_S+\varepsilon_*I_S)}-I_S\right)\ct_C^{(-1)}\cdot\ct_C\,\ket{\psi_{\rm phys}}=0.\nn
\ea
Then invoking Corollary~\ref{cor_triv1} gives the first statement of the Lemma. The second statement about the form of the trivialised physical states easily follows from Eqs.~\eqref{crap} and~\eqref{trivialisation} upon invoking Eqs.~\eqref{discreteClockState} and~\eqref{id2}.
\end{proof}

\noindent{\bf\hyperref[lem_10]{Lemma \ref{lem_10}}.}
The quantum deparametrisation map weakly equals the Page-Wootters reduction map and a system time evolution,
\ba
\calr_{\mathbf H}\approx U_S^\dag(\tau)\cdot \calr_{\mathbf S}(\tau),
\ea
while their inverses satisfy the strong relation for all $\tau$
\ba
\calr^{-1}_{\mathbf H}=\calr_{\mathbf S}^{-1}(\tau)\cdot U_S(\tau).\label{qrinv2}
\ea
In particular,
\ba
\calr_{\mathbf H}^{-1}\cdot\calr_{\mathbf H}&\approx& I_{\rm phys},\nn\\
\calr_{\mathbf H}\cdot\calr_{\mathbf H}^{-1}&\approx_S& I^{\rm phys}_S.\nn
\ea

\begin{proof}
The first statement follows from Eqs.~\eqref{schrodexp}--\eqref{Heisstate} and~\eqref{Hredstate} applied to an arbitrary physical state. The second statement follows from inserting Eq.~\eqref{trivinv} into Eq.~\eqref{QRinv}, $
\braket{\phi|\varepsilon_*}_C=e^{-i(g(\varepsilon_*)-\varepsilon_*\phi)}
$
and comparing with Eq.~\eqref{PWinvred1}. The invertibility properties are then implied by Lemma~\ref{lem_5}.
\end{proof}

\noindent{\bf\hyperref[thm_relobs2]{Theorem \ref{thm_relobs2}}.}
Let $\hat f_S^{\rm phys}(\tau)\in\cl(\ch_S^{\rm phys})$ be any evolving Heisenberg observable on the physical system Hilbert space. Its embedding coincides weakly with the quantum relational Dirac observable in Eq.~\eqref{qrelobs},
\ba
\mathcal{E}_{\mathbf H}\left(\hat f_S^{\rm phys}(\tau)\right)\approx\hat F_{f_S^{\rm phys},T}(\tau).
\ea
Conversely, the quantum deparametrisation of a quantum relational Dirac observable weakly yields the corresponding relational Heisenberg observable on the physical system Hilbert space,
\ba
\calr_{\mathbf H}\,\hat F_{f_S^{\rm phys},T}(\tau)\,\calr_{\mathbf H}^{-1}\approx_S\hat f_S^{\rm phys}(\tau).
\ea

\begin{proof}
Using Lemma~\ref{lem_10}, we have
\ba
\mathcal{E}_{\mathbf H}\left(\hat f_S^{\rm phys}(\tau)\right)&\approx&\calr_{\mathbf S}^{-1}(\tau) U_S(\tau)\,\hat f_S^{\rm phys}(\tau)\,U_S^\dag(\tau')\calr_{\mathbf S}(\tau')\nn\\
&\approx&\mathcal{E}_{\mathbf S}^\tau(\hat f_S^{\rm phys}),\nn
\ea
where we have used that $U_S(\tau-\tau')\calr_{\mathbf S}(\tau')\approx\calr_{\mathbf S}(\tau)$ and that the choice of $\tau$ in Eq.~\eqref{qrinv2} is arbitrary. The first statement then follows from Theorem~\ref{thm_relobs}. Conversely, since $\hat F_{f_S^{\rm phys},T}(\tau)$ maps physical states to physical states, using Lemma~\ref{lem_10} we have that $\hat F_{f_S^{\rm phys},T}(\tau)\,\calr_{\mathbf H}^{-1}\,\ket{\psi_S}\in\ch_{\rm phys}$ for all $\ket{\psi_S}\in\ch_S^{\rm phys}$. Hence, we can use Lemma~\ref{lem_10} to find
\ba
\calr_{\mathbf H}\,\hat F_{f_S^{\rm phys},T}(\tau)\,\calr_{\mathbf H}^{-1}&&\approx_SU_S^\dag(\tau')\calr_{\mathbf S}(\tau')\hat F_{f_S^{\rm phys},T}(\tau)\calr_{\mathbf S}^{-1}(\tau) U_S(\tau)\nn\\
&&\approx_S\hat f_S^{\rm phys}(\tau)\nn
\ea
using that the expression is (weakly) independent of $\tau'$, and invoking once more Theorem~\ref{thm_relobs}.
 \end{proof}

\noindent{{\bf\hyperref[thm_notbdep]{Theorem \ref{thm_notbdep}.}}
Consider an operator on $BS$ from the perspective of $A$, denoted $\hat{O}_{BS|A}^{\rm phys}\in \cl (\ch_{BS\vert A})$. From the perspective of $B$, this operator is independent of $\tau_{B}$, so that $\hat{O}_{AS|B}^{\rm phys} (\tau_A, \! \tau_B)=\hat{O}_{AS|B}^{\rm phys} (\tau_A)\in \cl (\ch_{AS\vert B})$ if and only if $\left[I_{A}\otimes\hat{O}_{BS|A},I_{A}\otimes H_{B}\otimes I_{S}\right]\approx 0$.
\begin{proof}
As noted in the main text, we have
\begin{align}
\hat{O}_{AS|B}^{\rm phys} (\tau_A, \! \tau_B) \!&=\! \Lambda^{A \to B}_{\rm \bf{S}} \hat{O}_{BS|A}^{\rm phys} \left(\Lambda^{A \to B}_{\rm \bf{S}} \right)^\dagger \\
&= \! \mathcal{R}_{\rm \bf{S}}(\tau_B) \circ \mathcal{E}_{\rm \bf{S}}^{\tau_A} \!\left( \hat{O}_{BS|A}^{\rm phys} \right) \! \circ \mathcal{R}^{-1}_{\rm \bf{S}}(\tau_B) \nn .
\end{align}
Now, combining Eqs.~\eqref{PWred} and~\eqref{PWinvred1} in the main text with Eq.~(43) in~\cite{Hoehn:2019owq}, we can write the above as
\begin{align} \label{eq_notbdep1}
\hat{O}_{AS|B}^{\rm phys} (\tau_A, \! \tau_B) \!&=\! \bra{\tau_{B}}_{B} \delta(\hat{C}_H)\! \left( \! \ket{\tau_A} \! \bra{\tau_A}\!  \otimes \! \hat{O}_{BS|A}^{\rm phys}\!  \right) \left[ \f{1}{t_{\rm max}}\int_{0}^{t_{\rm max}} d\phi\,\ket{\phi}_{B}\otimes e^{-i(\hat{H}_{A}+\hat{H}_{S})(\phi-\tau_{B})} \right]  ,
\end{align}
where $t_{\rm max}$ is the period of clock $B$, and 
\ba
\delta(\hat{C}_H) \ce \frac{1}{2\pi} \int_\mathbb{R} dr \; e^{i (\hat{H}_{A}+\hat{H}_{B}+\hat{H}_{S})r}  \nn
\ea
is equal to the improper projector defined in Eq.~\eqref{impproj}. Now, noting that $\ket{\tau_B}=e^{-i \hat{H}_{B}\tau_B}\ket{0_{B}}$, and $[\delta(\hat{C}_H),\hat{H}_{B}]=0$, we can rewrite Eq.~\eqref{eq_notbdep1} as
\begin{align}
\hat{O}_{AS|B}^{\rm phys} (\tau_A, \! \tau_B) \!&=\! \bra{0}_{B} \delta(\hat{C}_H)\! \left( \! \ket{\tau_A} \! \bra{\tau_A}\!  \otimes \! e^{i \hat{H}_{B}\tau_B}\hat{O}_{BS|A}^{\rm phys}\!  \right) e^{i(\hat{H}_{A}+\hat{H}_{S})\tau_{B}} \left[ \f{1}{t_{\rm max}}\int_{0}^{t_{\rm max}} d\phi\,\ket{\phi}_{B}\otimes e^{-i(\hat{H}_{A}+\hat{H}_{S})\phi} \right] \nn \\ 
&= \bra{0}_{B} \delta(\hat{C}_H)\! \left( \! \ket{\tau_A} \! \bra{\tau_A}\!  \otimes \! e^{i \hat{H}_{B}\tau_B}\hat{O}_{BS|A}^{\rm phys}\, e^{-i \hat{H}_{B}\tau_B}  \right) e^{i\hat{C}_{H}\tau_{B}} \left[ \f{1}{t_{\rm max}}\int_{0}^{t_{\rm max}} d\phi\,\ket{\phi}_{B}\otimes e^{-i(\hat{H}_{A}+\hat{H}_{S})\phi} \right]. \label{eq_notbdep2}
\end{align}
We will now prove an intermediary result, namely that $\hat{C}_{H}\left[ \f{1}{t_{\rm max}}\int_{0}^{t_{\rm max}} d\phi\,\ket{\phi}_{B}\otimes e^{(\hat{H}_{A}+\hat{H}_{S})\phi} \right]=0$ on $\ch_{AS|B}$. To do this, first note that we can construct a basis for $\ch_{AS|B}$ using eigenstates of $\hat{H}_{A}+\hat{H}_{S}$ of the form $\ket{E_{A},(-E_{A}-E_{B})_S}$, where $E_{A}\in\spec(\hat{H}_{A})\cap\spec(-\hat{H}_B-\hat{H}_S)$ and $E_{B}\in\spec(\hat{H}_{B})\cap\spec(-\hat{H}_A-\hat{H}_S)$. Consider then
\ba
&&\hat{C}_{H}\left[ \f{1}{t_{\rm max}}\int_{0}^{t_{\rm max}} d\phi\,\ket{\phi}_{B}\otimes e^{-i(\hat{H}_{A}+\hat{H}_{S})\phi} \right]\ket{E_{A},(-E_{A}-E_{B})_S}\nn \\
&&\qquad\qquad\qquad\qquad=\f{1}{t_{\rm max}}\int_{0}^{t_{\rm max}} d\phi\,e^{iE_{B}\phi} \hat{C}_{H}\ket{\phi}_{B}\otimes\ket{E_{A},(-E_{A}-E_{B})_S} \nn \\
&&\qquad\qquad\qquad\qquad= \f{1}{t_{\rm max}}\int_{0}^{t_{\rm max}} d\phi\,e^{iE_{B}\phi} \hat{C}_{H} \left[ \sum_{\varepsilon_j \in \spec(\hat H_C)}  e^{ig(\varepsilon_j)} e^{-i\varepsilon_j \phi} \ket{\varepsilon_j } \right] \otimes\ket{E_{A},(-E_{A}-E_{B})_S} \nn \\
&&\qquad\qquad\qquad\qquad= \sum_{\varepsilon_j \in \spec(\hat H_C)} e^{ig(\varepsilon_j)} \left[ \f{1}{t_{\rm max}}\int_{0}^{t_{\rm max}} d\phi\, e^{i(E_{B}-\varepsilon_j)\phi} \right] \hat{C}_{H} \ket{E_{A},(\varepsilon_j)_{B},(-E_{A}-E_{B})_S} \nn \\
&&\qquad\qquad\qquad\qquad= \sum_{\varepsilon_j \in \spec(\hat H_C)} e^{ig(\varepsilon_j)} \delta_{E_{B}\varepsilon_j} \hat{C}_{H} \ket{E_{A},(\varepsilon_j)_{B},(-E_{A}-E_{B})_S} \nn \\
&&\qquad\qquad\qquad\qquad=  e^{ig(E_{B})}  \hat{C}_{H} \ket{E_{A},E_{B},(-E_{A}-E_{B})_S} \nn \\
&&\qquad\qquad\qquad\qquad=  0 , \nn
\ea
where we have used Eq.~\eqref{discreteClockState} in the third line and Eq.~\eqref{id2} in the fifth line.
Consequently,
\ba
e^{i\hat{C}_{H}\tau_{B}} \left[ \f{1}{t_{\rm max}}\int_{0}^{t_{\rm max}} d\phi\,\ket{\phi}_{B}\otimes e^{-i(\hat{H}_{A}+\hat{H}_{S})\phi} \right] = \left[ \f{1}{t_{\rm max}}\int_{0}^{t_{\rm max}} d\phi\,\ket{\phi}_{B}\otimes e^{-i(\hat{H}_{A}+\hat{H}_{S})\phi} \right] \quad \text{on } \ch_{AS|B} , \nn
\ea
and applying this to Eq.~\eqref{eq_notbdep2}, we have
\begin{align}
\hat{O}_{AS|B}^{\rm phys} (\tau_A, \! \tau_B) \!&=\! \bra{0}_{B} \delta(\hat{C}_H)\! \left( \! \ket{\tau_A} \! \bra{\tau_A}\!  \otimes \! e^{i \hat{H}_{B}\tau_B}\hat{O}_{BS|A}^{\rm phys} \, e^{-i \hat{H}_{B}\tau_B}  \right) \left[ \f{1}{t_{\rm max}}\int_{0}^{t_{\rm max}} d\phi\,\ket{\phi}_{B}\otimes e^{-i(\hat{H}_{A}+\hat{H}_{S})\phi} \right]. \label{eq_notbdep3}
\end{align}
Now, if $\left[I_{A}\otimes\hat{O}_{BS|A},I_{A}\otimes H_{B}\otimes I_{S}\right]\approx 0$, then $e^{i \hat{H}_{B}\tau_B}\hat{O}_{BS|A}^{\rm phys}\, e^{-i \hat{H}_{B}\tau_B} =\hat{O}_{BS|A}^{\rm phys}$ and from Eq.~\eqref{eq_notbdep3} we can see that $\hat{O}_{AS|B}^{\rm phys} (\tau_A, \! \tau_B)$ does not depend on $\tau_{B}$. To prove the converse statement, note that Eq.~\eqref{eq_notbdep3} implies that $\frac{d}{d\tau_{B}}\hat{O}_{AS|B}^{\rm phys} (\tau_A, \! \tau_B)=0$ if and only if $I_{A}\otimes\frac{d}{d\tau_{B}} \left( e^{i \hat{H}_{B}\tau_B}\hat{O}_{BS|A}^{\rm phys}\, e^{-i \hat{H}_{B}\tau_B}\right)\approx 0$, and that the latter is true if and only if $\left[I_{A}\otimes\hat{O}_{BS|A},I_{A}\otimes H_{B}\otimes I_{S}\right]\approx 0$, thus concluding the proof of the theorem.
\end{proof}

\noindent{\bf\hyperref[lem_monoPNPQ]{Lemma \ref{lem_monoPNPQ}}}
The quantisation of the relational observable $F_{T_B,Q_A}(\tau)$ of Lemma~\ref{lem_monoPNP}, which encodes the value of the monotonic clock $T_B$ of a periodic system relative to the value $\tau$ of the clock $Q_A$ of an aperiodic system is given by
\begin{equation}\label{monoPNPQ-m1}
\hat{F}_{T_B,Q_A}(\tau) = \tau\hat{I}-\hat{Q}_A+\hat{\phi}_{B} \ ,
\end{equation}
and it satisfies the property:
\begin{equation}\label{monoPNPQ-0-app}
\braket{\alpha^s_{C_H}\cdot\hat{F}_{T_B,Q_A}(\tau)}_{21} = \braket{\hat{F}_{T_B,Q_A}(\tau)- t_{{\rm max};B}\hat{Z}_B(s)}_{21}\,,
\end{equation}
where $\alpha^s_{C_H}\cdot\hat{F}_{T_B,Q_A}(\tau) := U_{C_H}^{\dagger}(s)\hat{F}_{T_B,Q_A}(\tau)U_{C_H}(s)$, and
\begin{equation}
\hat{Z}_B(s) := \frac{1}{ t_{{\rm max};B}}\int_0^{ t_{{\rm max};B}}d\phi_B\, \left\lfloor\frac{s+\phi_B}{ t_{{\rm max};B}}\right\rfloor\ket{\phi_B}\!\bra{\phi_B} \,,
\end{equation}
and $\braket{\hat{O}}_{21} = \braket{\psi_2|\hat{O}|\psi_1}$ for any operator $\hat{O}$, with $\ket{\psi_{1,2}}$ leading to wave functions that are integrable in $\phi$ (cf. Lemmas~\ref{lem_quantumPhiEvol} and~\ref{lem_clobsQ}). Thus, Eq.~\eqref{monoPNPQ-0-app} is a quantum version of the result of Lemma \ref{lem_monoPNP}.
\begin{proof}
Eq.~\eqref{monoPNPQ-m1} is a direct quantisation of Eq.~\eqref{app-relobs-TQ}. Using Lemma \ref{lem_quantumPhiEvol}, it is straightforward to obtain
\begin{align}
\braket{\alpha^s_{C_H}\cdot\hat{F}_{T_B,Q_A}(\tau)}_{21}&=\bra{\psi_2}\left[\tau\hat{I}-\hat{Q}_A-s\hat{I}+\frac{1}{ t_{{\rm max};B}}\int_0^{ t_{{\rm max};B}}d\phi_B\,\left(s+\phi_B- t_{{\rm max};B}\left\lfloor\frac{s+\phi_B}{ t_{{\rm max};B}}\right\rfloor\right)\ket{\phi_B}\!\bra{\phi_B}\right]\ket{\psi_1}\notag\\
&=\bra{\psi_2}\left[\tau\hat{I}-\hat{Q}_A-s\hat{I}+s\hat{I}+\hat{\phi}_{B}-\int_0^{ t_{{\rm max};B}}d\phi_B\,\left\lfloor\frac{s+\phi_B}{ t_{{\rm max};B}}\right\rfloor\ket{\phi_B}\!\bra{\phi_B}\right]\ket{\psi_1}\notag\\
&=\bra{\psi_2}\left[\left(\tau\hat{I}- t_{{\rm max};B}\hat{Z}_B(s)\right)-\hat{Q}_A+\hat{\phi}_{B}\right]\ket{\psi_1}\notag\\
&= \braket{\hat{F}_{T_B,Q_A}(\tau)- t_{{\rm max};B}\hat{Z}_B(s)}_{21}\,.
\end{align}
\end{proof}

\section{Failure of the standard conditional probabilities in the Page-Wootters formalism} \label{sFailureCondProbPW}

We recall the correct definition of the conditional probability densities in the Page-Wootters formalism (Sec.~\ref{sCorrectCondProbPW})
\ba \label{eProbDens}
P(f_S|\tau)&\ce&\f{\braket{\psi_{\rm phys}|\,\hat F_{\ket{f_S^{\rm phys}}\!\bra{f_S^{\rm phys}},T}(\tau)\,|\psi_{\rm phys}}_{\rm phys}}{\braket{\psi_{\rm phys}\,|\psi_{\rm phys}}_{\rm phys}}\nn\\
&=&\f{\braket{\psi_S^{\rm phys}(\tau)|\,\ket{f_S^{\rm phys}}\!\bra{f_S^{\rm phys}}\,|\psi_S^{\rm phys}(\tau)}_{\ch_S^{\rm phys}}}{\braket{\psi_S^{\rm phys}(\tau)|\psi_S^{\rm phys}(\tau)}_{\ch_S^{\rm phys}}},\nn
\ea
and contrast them with the following na\"ive definition with respect to the conditional inner product, as is usually the case
\ba \label{eNaivePDens}
    \tilde{P}(f_S|\tau) &\ce& \f{\braket{\psi_{\rm phys}|\, (\ket{\tau}\!\bra{\tau}\otimes \ket{f_S^{\rm phys}}\!\bra{f_S^{\rm phys}}) \,|\psi_{\rm phys}}_{\rm kin}}{\braket{\psi_{\rm phys}|\, (\ket{\tau}\!\bra{\tau}\otimes I_{S}) \,|\psi_{\rm phys}}_{\rm kin}} . \nn
\ea
We will now show how the latter definition fails when $G=\mathbb{R}$, but coincides with the correct definition for $G=\rm{U}(1)$.~First note that there exists some $\ket{\phi_{\rm phys}}$ such that $\ket{f_S^{\rm phys}}=\calr_{\mathbf S}(\tau)\ket{\phi_{\rm phys}}$.~Then
\ba
    \braket{\psi_{\rm phys}|\, (\ket{\tau}\!\bra{\tau}\otimes \ket{f_S^{\rm phys}}\!\bra{f_S^{\rm phys}}) \,|\psi_{\rm phys}}_{\rm kin} &=& | \braket{f_S^{\rm phys}|\psi^{\rm phys}_{S}(\tau)}_{\ch_{S}} |^{2} \nn\\
    &=& | \braket{\phi_{\rm phys}|\, (\ket{\tau}\!\bra{\tau}\otimes I_{S}) \,|\psi_{\rm phys}}_{\rm kin} |^{2}
\ea
and therefore
\ba \label{eNaiveCondProb2}
    \tilde{P}(f_S|\tau) &\ce& \f{|\braket{\phi_{\rm phys}|\, (\ket{\tau}\!\bra{\tau}\otimes I_{S}) \,|\psi_{\rm phys}}_{\rm kin} |^{2}
    }{\braket{\psi_{\rm phys}|\, (\ket{\tau}\!\bra{\tau}\otimes I_{S}) \,|\psi_{\rm phys}}_{\rm kin}} .
\ea
Now, in the case where $G=\rm{U}(1)$, for the numerator we have 
\ba
    \braket{\phi_{\rm phys}|\, (\ket{\tau}\!\bra{\tau}\otimes I_{S}) \,|\psi_{\rm phys}}_{\rm kin} &=& \braket{\phi_{\rm phys}|\psi_{\rm phys}}_{\rm phys} \nn\\
    &=& \braket{f_S^{\rm phys}|\psi_{S}^{\rm phys}(\tau)}_{\ch_{S}^{\rm phys}}
\ea
where we have used Lemma~\ref{lem_6} in the first equality and Corollary~\ref{corolInProdsPW} in the second equality, and for the denominator we have similarly
\ba
    \braket{\psi_{\rm phys}|\, (\ket{\tau}\!\bra{\tau}\otimes I_{S}) \,|\psi_{\rm phys}}_{\rm kin} &=& \braket{\psi_{\rm phys}|\psi_{\rm phys}}_{\rm phys} \nn\\
     &=& \braket{\psi^{\rm phys}_{S}(\tau)|\psi^{\rm phys}_{S}(\tau)}_{\ch_{S}^{\rm phys}}\,.
\ea
Thus, for $G=\rm{U}(1)$, the na\"ive conditional probability density is the correct one:
\ba
    \tilde{P}(f_S|\tau) &=& \f{\braket{\psi^{\rm phys}_{S}(\tau)|f_S^{\rm phys}}_{\ch_{S}^{\rm phys}} \braket{f_S^{\rm phys}|\psi^{\rm phys}_{S}(\tau)}_{\ch_{S}^{\rm phys}}}{\braket{\psi^{\rm phys}_{S}(\tau)|\psi^{\rm phys}_{S}(\tau)}_{\ch_{S}^{\rm phys}}} . \nn\\
    &=& P(f_S|\tau).
\ea
Considering instead the case where $G=\mathbb{R}$, we can apply Eq.~\eqref{eProofLem10} to Eq.~\eqref{eNaiveCondProb2} to obtain
\ba
    \tilde{P}(f_S|\tau) &=& |\mathbb{Z}_G|\f{t_{\rm max}}{N_G}\,\frac{|\braket{\phi_{\rm phys}|\psi_{\rm phys}}_{\rm phys}|^{2}}{\braket{\psi_{\rm phys}|\psi_{\rm phys}}_{\rm phys}}.
\ea
Recalling that for non-compact clocks, $N_G=2\pi$ and $|\mathbb{Z}_G|=|\mathbb{Z}|$ diverges, we see that $\tilde{P}(f_S|\tau)$ likewise diverges in this case.
This highlights the importance of the correct definition of the conditional probabilities in the Page-Wootters formalism when using a periodic clock.

\end{document}